\documentclass{article}

\PassOptionsToPackage{numbers, compress}{natbib}
\usepackage{amsmath,amssymb,amsthm,fullpage,mathrsfs,pgf,tikz,caption,subcaption,mathtools}
\usepackage[ruled,vlined,linesnumbered]{algorithm2e}
\usepackage{natbib}
\usepackage{comment}
\usepackage{amsmath,amssymb,amsthm,mathtools}
\usepackage[utf8]{inputenc} 
\usepackage[T1]{fontenc}    
\usepackage{hyperref}       
\usepackage{url}            
\usepackage{booktabs}       
\usepackage{amsfonts}       
\usepackage{nicefrac}       
\usepackage{microtype}      
\usepackage[margin=1in]{geometry}
\usepackage{bbm}
\usepackage{enumitem}
\usepackage{xcolor}
\usepackage{cleveref}

\newcommand{\maxh}[1]{}
\newcommand{\snote}[1]{}
\newcommand{\mnote}[1]{}

\newtheorem{theorem}{Theorem}[section]
\newtheorem{corollary}[theorem]{Corollary}
\newtheorem{proposition}[theorem]{Proposition}
\newtheorem{lemma}[theorem]{Lemma}
\newtheorem{definition}[theorem]{Definition}

\newtheorem{claim}[theorem]{Claim}

\newcommand{\R}{\mathbb{R}}

\newcommand{\id}[1]{\mathbbm{1}_{#1}}
\newcommand{\eps}{\varepsilon}
\newcommand{\val}{\text{val}}

\newcommand{\A}{\mathcal{A}}

\DeclareMathOperator{\poly}{poly}

\DeclareMathOperator*{\pE}{\widetilde{\mathbb{E}}}
\DeclareMathOperator*{\e}{\mathbb{E}}

\newcommand{\E}[2]{\mathbb{E}_{{#1}}\left[{#2}\right]}

\newcommand\uinner[2]{\langle #1, #2 \rangle}

\newcommand{\ind}{\mathbbm{1}}

\newcommand\norm[1]{\left\lVert#1\right\rVert}

\newcommand{\pe}[1]{\widetilde{\mathbb{E}}\left[ {#1} \right]}
\newcommand{\tO}{\widetilde{O}}

\title{High Dimensional Expanders: Eigenstripping, Pseudorandomness, and Unique Games}

\author{
   Mitali Bafna\thanks{Department of Computer Science, Harvard University, MA 02138. Email: \texttt{mitalibafna@g.harvard.edu}. Supported in part by the Simons Investigator Award to Madhu Sudan.}
   \and
   Max Hopkins\thanks{Department of Computer Science and Engineering, UCSD, CA 92092. Email: \texttt{nmhopkin@eng.ucsd.edu}. Supported by NSF Award DGE-1650112.}
   \and
     Tali Kaufman\thanks{Department of Computer Science, Bar-Ilan University. Email: \texttt{kaufmant@mit.edu}. Supported by ERC and BSF.}
     \and
     Shachar Lovett\thanks{Department of Computer Science and Engineering, UCSD, CA 92092. Email: \texttt{slovett@cs.ucsd.edu}. Supported by NSF Award CCF-1953928.}
}

\begin{document}
\maketitle

\begin{abstract}
    Higher order random walks (HD-walks) on high dimensional expanders (HDX) have seen an incredible amount of study and application since their introduction by Kaufman and Mass (ITCS 2016), yet their broader combinatorial and spectral properties remain poorly understood. We develop a combinatorial characterization of the spectral structure of HD-walks on two-sided local-spectral expanders (Dinur and Kaufman FOCS 2017), which offer a broad generalization of the well-studied Johnson and Grassmann graphs. Our characterization, which shows that the spectra of HD-walks lie tightly concentrated in a few combinatorially structured strips, leads to novel structural theorems such as a tight $\ell_2$-characterization of edge-expansion, as well as to a new understanding of local-to-global graph algorithms on HDX.
    
    Towards the latter, we introduce a novel spectral complexity measure called \textit{Stripped Threshold Rank}, and show how it can replace the (much larger) threshold rank as a parameter controlling the performance of algorithms on structured objects. Combined with a sum-of-squares proof for the former $\ell_2$-characterization, we give a concrete application of this framework to algorithms for unique games on HD-walks, where in many cases we improve the state of the art (Barak, Raghavendra, and Steurer FOCS 2011, and Arora, Barak, and Steurer JACM 2015) from nearly-exponential to polynomial time (e.g.\ for sparsifications of Johnson graphs or of slices of the $q$-ary hypercube). Our characterization of expansion also holds an interesting connection to hardness of approximation, where an $\ell_\infty$-variant for the Grassmann graphs was recently used to resolve the 2-2 Games Conjecture (Khot, Minzer, and Safra FOCS 2018). We give a reduction from a related $\ell_\infty$-variant to our $\ell_2$-characterization, but it loses factors in the regime of interest for hardness where the gap between $\ell_2$ and $\ell_\infty$ structure is large. Nevertheless, our results open the door for further work on the use of HDX in hardness of approximation and their general relation to unique games.
    
\end{abstract}

\newpage

\tableofcontents

\newpage

\section{Introduction}\label{sec:intro}
Since their introduction by Kaufman and Mass \cite{kaufman2016high} in 2016, higher order random walks (HD-walks) on high dimensional expanders (HDX) have seen an explosion of research and application throughout theoretical computer science, perhaps most famously in approximate sampling \cite{anari2019log,alev2020improved,anari2020spectral,chen2020rapid,chen2021optimal,chen2021rapid,feng2021rapid,jain2021spectral,liu2021coupling,blanca2021mixing}, but also in CSP-approximation \cite{alev2019approximating}, error correction \cite{dinur2019list,dikstein2020locally,jeronimo2020unique,jeronimo2021near}, and agreement testing \cite{dinur2017high,dikstein2019agreement,kaufman2020local}. These breakthroughs, while evidence of the importance of HD-walks, tend to have a fairly narrow focus in their analysis of the walks themselves---most rely only on proving (often one-sided) bounds on spectral expansion. On the other hand, there are many reasons to study combinatorial and spectral structure beyond the second eigenvalue. Such results are often useful, for instance, when designing graph algorithms (e.g.\ techniques relying on threshold rank \cite{barak2011rounding,guruswami2011lasserre,gharan2013new}), and have even shown up in recent breakthroughs in hardness of approximation \cite{khot2017independent,dinur2018towards,dinur2018non,barak2018small,khot2018small,subhash2018pseudorandom}, where such an analysis of the Johnson and Grassmann graphs (themselves HD-walks) proved crucial for the resolution of the 2-2 Games Conjecture. Despite so many interesting connections and the recent flurry of work on HDX, these structures remain relatively unexplored and poorly understood.

Building on work of \cite{kaufman2020high,dikstein2018boolean}, we make progress on this problem, building a combinatorial characterization of the spectral structure of HD-walks on \textit{two-sided local-spectral expanders}, a variant of HDX introduced by Dinur and Kaufman \cite{dinur2017high} to generalize the Johnson and Grassmann schemes\footnote{The Johnson Scheme consists of matrices indexed by $k$-sets of $[n]$ which depend only on intersection size. HD-walks on the complete complex are exactly the non-negative elements of the Johnson Scheme. The Grassmann scheme is an analogous object over vector spaces.} (and their corresponding agreement tests). We show how our characterization leads to a new understanding of local-to-global algorithms on HDX, including the introduction of a novel spectral parameter generalizing threshold rank, and give a concrete application to efficient algorithms for Unique Games on HD-walks. While the structural theorems we develop lie in a different regime than the one needed to recover hardness results like the proof of 2-2 Games, we open the door for future work connecting HDX and hardness of approximation.
\subsection{Contributions}
We start with an informal overview of our contributions, broken down into three main sections. We give a more thorough exposition of these results in \Cref{sec:results} along with some relevant background, and overview their proofs in \Cref{sec:proof-overview}. The remainder of the paper is devoted to background, proof details, and further discussion.

\paragraph{Eigenstripping and ST-Rank:} We develop the interplay of spectral and combinatorial structure on HD-walks, and introduce a spectral parameter called \textit{Stripped Threshold Rank} (ST-Rank) that generalizes threshold rank and controls the performance of local algorithms on such graphs. In more detail, we prove that the spectrum of any HD-walk can be decomposed into small, disjoint intervals we call ``eigenstrips,'' where the corresponding eigenvectors in each strip share the same combinatorial structure.
\begin{theorem}[Eigenstripping (informal)]
Let $M$ be an HD-walk on $k$-sets of an $\gamma$-HDX. Then the spectrum of $M$ lies in $k+1$ eigenstrips of width $O_{k,M}(\gamma^{1/2})$,\footnote{This was recently improved to $O_k(\gamma)$ by Zhang \cite{Zhang2020}.} each corresponding combinatorially to a level in the underlying HDX.\footnote{There are $k+1$ rather than $k$ strips since $0$ is a dimension of the complex (corresponding to the empty set).}
\end{theorem}
The ST-Rank of an HD-walk (\Cref{def:ST-rank-intro}) then measures the number of spectral strips with eigenvalues above some smallness threshold. This generalizes the concept of standard threshold rank, which measures the \textit{total} number of such eigenvalues, to coarser eigendecompositions. Threshold rank itself is well known to control the performance of many graph algorithms \cite{arora2015subexponential,barak2011rounding,guruswami2011lasserre,gharan2013new,hopkins2020subexponential}. At their core, these techniques often boil down to various methods of enumerating over eigenvectors with large eigenvalues. On structured objects like HD-walks, we argue one should instead enumerate over \textit{strips} of eigenvectors with matching combinatorial structure, and therefore that \textit{ST-rank} is the relevant parameter. Since many important objects share similar combinatorial and spectral characteristics to HD-walks (e.g. noisy hypercube, $q$-ary hypercube...), we expect ST-Rank to have far-reaching applications beyond those studied in this work.

\paragraph{Edge-expansion in HD-walks:} The combinatorial and spectral machinery we develop allows us to characterize a fundamental graph property on HD-walks: edge-expansion. Given a graph $G_M=(V,E)$ (corresponding to a random walk $M$), the edge expansion of a subset $S \subseteq V$ measures the expected probability of leaving $S$ after a single application of $M$. It is hard to overstate the importance of edge-expansion throughout theoretical computer science. While the most widely used characterizations of expansion are for families where all sets expand (expander graphs) or all \textit{small} sets expand (small set expanders), more involved characterizations for objects such as the Johnson and Grassmann graphs have seen recent use for both hardness of approximation \cite{khot2017independent,dinur2018towards,dinur2018non,barak2018small,subhash2018pseudorandom} and algorithms \cite{bafna2020playing}. We prove that HD-walks exhibit similar expansion characteristics to the Johnson graphs in two key aspects: the expansion of \textit{locally-structured} sets called links, and the \textit{structure} of non-expanding sets.

Applying our machinery to the former reveals a close connection between local expansion and ST-rank: 
\begin{theorem}[Local Expansion vs. Eigenstripping (informal)]\label{thm:intro-local-vs-global}
The (non)-expansion of any $i$-link\footnote{Links are grouped by level of the complex. The set of links at each level give increasingly finer decompositions of the HDX into local parts. See \Cref{sec:results} for details.} is almost exactly the eigenvalue corresponding to the $i$th level's eigenstrip.
\end{theorem}
This connection is particularly useful in the construction and analysis of spectral `local-to-global' graph algorithms, where it helps tie performance guarantees to ST-rank rather than threshold rank. Surprisingly, to our knowledge, \Cref{thm:intro-local-vs-global} novel even for the Johnson graphs (though in this case it follows easily from known results). In fact, this special case alone already gives a better understanding of recent local-to-global algorithm for unique games \cite{bafna2020playing}.

A corollary of this result, on the other hand, extends a very well known fact on the Johnson graphs: \textit{locally structured sets expand poorly}. This raises a natural question: are \emph{all} (small) non-expanding sets explained by local structure? Before answering, we address a subtle point: what exactly does it mean to be ``explained'' by local structure? There are a number of reasonable ways to formalize this notion, each with its own use. Following \cite{bafna2020playing}, we resolve an \textit{$\ell_2$-variant} of this conjecture stating: 
\begin{theorem}[Characterizing Expansion in HD-walks (informal $\ell_2$-variant)]
Any non-expanding set in an HD-walk has \textbf{high variance} across low-level links.
\end{theorem}

Similar results used in hardness of approximation \cite{khot2017independent,dinur2018towards,dinur2018non,barak2018small,khot2018small,subhash2018pseudorandom}, on the other hand, rely on an $\ell_\infty$-variant that replaces variance with \textit{maximum}. While our bound is exactly tight in the former regime, it (necessarily) loses a factor in the latter in cases where there is a significant gap between $\ell_2$ and $\ell_\infty$ structure. Proving a tight bound directly on the $\ell_\infty$-variant for HD-walks remains an important open question given their close connections to the structures used in hardness of approximation.

\paragraph{Playing Unique Games:} Unique Games are a well-studied class of 2-CSP that underlie a central open question in computational complexity and algorithms called the \textit{Unique Games Conjecture} (UGC). The UGC stipulates that distinguishing between almost satisfiable (value $\geq 1-\eps$) and highly unsatisfiable (value $\leq \eps$) instances of Unique Games is NP-hard. It implies optimal hardness of approximation for a host of optimization problems such as Max-Cut~\cite{khot2007optimal}, Vertex Cover~\cite{khot2008vertex}, and in fact all CSPs~\cite{raghavendra2008optimal}, but is still not known to be true or false. In a recent breakthrough, Khot, Minzer, and Safra \cite{subhash2018pseudorandom} (following a long line of work \cite{khot2017independent,dinur2018towards,dinur2018non,barak2018small,khot2018small}) used an $\ell_\infty$ characterization of expansion on the Grassmann graphs to prove a weaker variant known as the 2-2 Games Conjecture: it is NP-hard to distinguish $(\frac{1}{2} - \varepsilon)$-satisfiable instances from $\epsilon$-satisfiable instances of unique games. On the other hand, in contrast to hard CSPs such as 3SAT, a long line of research establishes sub-exponential time approximation algorithms for unique games~\cite{arora2015subexponential} and furthermore shows that approximating unique games is easy on various restricted classes of graphs such as expanders~\cite{makarychev2010play}, perturbed random graphs~\cite{KMM11}, certifiable small-set expanders and Johnson graphs~\cite{bafna2020playing}. 

While our $\ell_2$-characterization of expansion may not be as useful for hardness of approximation as the $\ell_\infty$-variant, it gets its chance to shine in the latter context: \textit{algorithms}. For most classes of constraint graphs, the best known algorithms for unique games depend on the threshold rank \cite{arora2015subexponential,barak2011rounding,guruswami2011lasserre} of the graph. Using a sum-of-squares variant of our $\ell_2$-characterization and the connection between local expansion and stripped eigenvalues, we build a local-to-global algorithm for unique games on HD-walks (based off the recent paradigm of \cite{bafna2020playing}) that depends instead on ST-rank.
\begin{theorem}[Playing Unique Games on HD-walks (informal)]\label{intro-intro:UG}
For any $\varepsilon \in (0,.01)$, there exists an algorithm $\mathcal{A}$ with the following guarantee. If $\mathcal{I}$ is a $(1-\varepsilon)$-satisfiable instance of unique games whose constraint graph is an HD-walk $M$\footnote{Note that self-adjoint random walks can be equivalently viewed as undirected weighted graphs, and therefore also as underlying constraint graphs for unique games.} on $k$-sets of an HDX, then:
\begin{enumerate}
    \item $\mathcal{A}$ outputs a $\text{poly}(k^{-r},\varepsilon)$-satisfying solution.
    \item $\mathcal{A}$ runs in time at most $N^{\tilde{O}\left({k \choose r}\varepsilon^{-1}\right)}$
\end{enumerate}
where $r = R_{1 - O(\eps)}(M)$ is the ST-rank of $M$ above threshold $1 - O(\eps)$, and there are $N$ total $k$-sets.
\end{theorem}

When $k = O(1)$, the ST-rank above any threshold $\tau$ is substantially smaller than the threshold rank above $\tau$ for (non-expanding) HD-walks, and \Cref{intro-intro:UG} therefore obtains a strict improvement over previous results based on threshold rank for this family. In many cases, such as sparsifications of (constant-level) Johnson graphs or related objects like slices of the noisy $q$-ary hypercube, our algorithm reduces the best-known runtime from nearly-exponential to polynomial. Besides its independent algorithmic interest, the result carries a number of connections to other open problems both in and outside of the study of unique games. The algorithm sheds some (limited) light,\footnote{The result certainly does not rule out reductions using such structure, but states that one must be careful in doing so that the resulting constraint graphs fall outside the parameter regime for our algorithm.} for instance, on requisite structure for candidate hardness reductions that use direct product testing \cite{khot2016candidate} or potential attempts to use agreement tests based on local-spectral expanders \cite{dinur2017high}. It gives some hope for better unique games algorithms on graphs like the hypercube~\cite{AKKT15} as well, which can be viewed as an HD-walk on level $k = O(\log N)$ on a weaker (one-sided) notion of high dimensional expansion than we study. Outside of unique games the result has some further connections to error correcting codes, where approximation algorithms for general CSPs on HDX underlie recent breakthroughs in efficient decoding for locally testable codes \cite{jeronimo2020unique,jeronimo2021near}. In fact, algorithms specifically for unique games over expanders have already seen similar use for efficient decoding of direct-product codes based on HDX structure \cite{dinur2021list}.



\section{Our Results}\label{sec:results}
We now move to a more in-depth exposition of our main results. First, however, we cover some background regarding local-spectral expanders and higher order random walks. For a full treatment of these and related objects, see \Cref{sec:prelims}. Local-spectral expansion is a robust notion of local connectivity on \textbf{pure simplicial complexes} introduced by Dinur and Kaufman \cite{dinur2017high}. A $d$-dimensional pure simplicial complex is the downward closure of a $d$-uniform hypergraph, that is a collection:
\[
X = X(0) \cup \ldots \cup X(d)
\]
where $X(d) \subseteq {[n] \choose d}$ is a $d$-uniform hypergraph, and $X(i)$ consists of all $\tau \in {[n] \choose i}$ such that $\tau \subset T$ for some $T \in X(d)$ (and $X(0)=\{\emptyset\}$). We note that this notation is off-by-one from much of the HDX literature, where $X(i)$ is instead given by sets of size $i+1$, a notation rooted in the topological view of simplicial complexes. The former definition is more common in combinatorial works (e.g. the sampling literature \cite{anari2020spectral}), and more natural for our purposes.

Simplicial complexes come equipped with natural local structure called \textbf{links}. For every $i$ and $\tau \in X(i)$, the link of $\tau$ is the restriction of the complex to faces containing $\tau$, that is:
\[
X_\tau = \{ \sigma~:~ \sigma \cap \tau = \emptyset, \ \sigma \cup \tau \in X\}.
\]
We call the link of an i-face an $i$-link, and when clear from context, will also use $X_\tau$ to denote the set of faces at a given level (e.g.\ in $X(d)$) containing $\tau$. Following Dinur and Kaufman \cite{dinur2017high}, we say a complex is a \textbf{two-sided $\gamma$-local-spectral expander} if the graph underlying every link is a $\gamma$-spectral expander (that is all non-trivial eigenvalues are smaller than $\gamma$ in absolute value). Finally, it is worth noting that we actually study the more general set of \textbf{weighted pure simplicial complexes} (see \Cref{sec:prelims}), but for simplicity restrict to the unweighted case for the moment as it requires essentially no modification.

Higher order random walks \cite{kaufman2016high} are analogues of the walk associated with standard graphs (given by the normalized adjacency matrix) that moves from vertex to vertex via an edge. In a higher order random walk, one applies a similar process at any level of the complex---moving for instance between two edges via a triangle, or two triangles via a pyramid. We call walks between $k$-faces of a complex \textbf{$k$-dimensional HD-walks} (see \Cref{def:intro-HD-walk} for formal definition), and study a broad set of walks that capture important structures such as sparsifications of the Johnson and Grassmann schemes.\footnote{We note that the Grassmann scheme comes from applying HD-walks to the Grassmann poset, which is not a simplicial complex.} For simplicity, throughout the introduction we will often focus on two natural classes of HD-walks which see the most use in the literature: the \textbf{canonical walks $N_k^i$} which walk between $k$-faces via a neighboring $(k+i)$-face, and the \textbf{partial-swap walks $S_k^{i}$} which do the same but only move between $k$-faces of fixed intersection size $k-i$. While this latter class may seem less natural at first, notice that on the complete complex they are exactly the well-studied Johnson graph $J(n,k,k-i)$ (the graph on ${[n] \choose k}$ where $(v,w) \in E$ iff $|v \cap w| = k-i$).

\subsection{Eigenstripping and ST-Rank}
With background out of the way, we start our results in earnest with a more in-depth discussion of eigenstripping and ST-rank: the spectral and combinatorial structure of HD-walks.
We prove that the spectra of $k$-dimensional HD-walks lie in $k+1$ tightly concentrated ``eigenstrips'', where the $i$th strip corresponds combinatorially to functions lifted from the $i$th level of the complex by averaging.
\begin{theorem}[Spectrum of HD-Walks (Informal \Cref{thm:approx-ortho} + \Cref{prop:hdx-eig-vals} + \Cref{prop:eig-decrease})]\label{thm:new-intro-hdx-spectra}
Let $M$ be an HD-walk on the $k$th level of a two-sided $\gamma$-local-spectral expander. Then the spectra of $M$ is highly concentrated in $k+1$ strips:
\[
\text{Spec}(M) \in \{1\} \cup \bigcup_{j=1}^k \left[\lambda_i(M) - e, \lambda_i(M) + e \right]
\]
where the $\lambda_i(M)$ are decreasing constants depending only on $M$, and the error term satisfies $e \leq O_{k,M}(\sqrt{\gamma})$. Moreover, the span of eigenvectors with eigenvalues in the strip $\lambda_i(M) \pm e$ are (approximately) functions in $X(i)$ lifted to $X(k)$ by averaging.
\end{theorem}
Recently, Zhang \cite{Zhang2020} provided a quantitative improvement of \Cref{thm:new-intro-hdx-spectra} such that $e \leq O_k(\gamma)$. 
\Cref{thm:new-intro-hdx-spectra} can be seen as a marriage (and generalization) of previous results studying $N_k^1$ of Kaufman and Oppenheim \cite{kaufman2020high}, and Dikstein, Dinur, Filmus, and Harsha (DDFH) \cite{dikstein2018boolean}. The former prove that $N_k^1$ exhibits spectral eigenstripping, while the latter introduce the corresponding combinatorial structure but lack the machinery to tie it directly to these strips. We fill in the gap by proving a general linear algebraic theorem of independent interest. 
\begin{theorem}[Approximate Eigendecompositions Imply Eigenstripping (Informal \Cref{thm:approx-ortho})]\label{thm:intro-approx-ortho}
Let $M$ be a self-adjoint operator over an inner product space $V$, and $V=V^1 \oplus \ldots \oplus V^k$ a decomposition satisfying $\forall 1 \leq i \leq k, f_i \in V^i$:
\[
\norm{Mf_i - \lambda_i f_i} \leq c_i\norm{f_i}
\]
for some family of constants $(\{\lambda_i\}_{i=1}^k,\{c_i\}_{i=1}^k)$. Then as long as the $c_i$ are sufficiently small, the spectra of $M$ is concentrated around each $\lambda_i$:
\[
\text{Spec}(M) \subseteq \bigcup_{i=1}^k \left [ \lambda_i - e, \lambda_i + e\right ] =  I_{\lambda_i},
\]
where $e = O_{k,\lambda}\left(\sqrt{ \max_i\{c_i\}}\right)$. This was recently improved to $e \leq O_k(\max_i\{c_i\})$ by Zhang \cite{Zhang2020}.
\end{theorem}
\Cref{thm:new-intro-hdx-spectra} then follows from proving that DDFH's decomposition is an approximate eigendecomposition for all HD-walks (DDFH only show this holds for $N_k^1$).

In the analysis of graphs, it is often useful to bound the number of eigenvalues above some smallness threshold. This parameter, called the \textbf{threshold rank}, often plays a key role, for instance, in graph algorithms, including many methods for generic 2CSP approximation \cite{arora2015subexponential,barak2011rounding,guruswami2011lasserre,gharan2013new, hopkins2020subexponential}. Painted in broad strokes, these algorithms usually boil down to some method of eigenvalue enumeration, whether performed directly or implicitly in the proof of key structural lemmas. Given the spectral structure of HD-walks (or generally any operator with a natural approximate eigendecomposition), we argue that it is generally more natural to enumerate over \textit{strips} with large eigenvalues instead. This motivates a natural spectral complexity measure we call \textbf{stripped threshold rank}.
\begin{definition}[Stripped Threshold Rank]\label{def:ST-rank-intro}
Let $M$ be a linear operator over a vector space $V$ with decomposition $V=\bigoplus_i V^i$ denoted $\mathscr D$, where each $V^i$ is the span of some set of eigenvectors. Given $\delta \in \mathbb{R}$, the stripped threshold rank (ST-Rank) with respect to $\delta$ and $\mathscr D$ is:
\[
R_\delta(M,\mathscr D) = \left|\{ V^i \in \mathscr D : \exists f \in V^i, Mf = \lambda f, \lambda > \delta\}\right|.
\]
In this work, $\mathscr D$ will always correspond to the eigenstrips given by \Cref{thm:new-intro-hdx-spectra} so we drop it from the notation.
\end{definition}
The ST-Rank of any (non-expanding) $k$-dimensional HD-walk is always substantially smaller than its standard threshold rank (which is at least $\text{poly}(n)$, and often as large as $n^{\Omega(k)}$). As a general algorithmic paradigm, if one can leverage the combinatorial structure of eigenstrips to replace eigenvalue enumeration, it is possible to achieve substantially better performance on structured graphs. In \Cref{sec:results-UG}, we discuss a concrete instantiation of this framework to unique games, where we gain substantial improvements over state of the art algorithms \cite{barak2011rounding,arora2015subexponential} over such constraint graphs.

Before moving on to such results, however, it is natural to ask whether we can give a finer-grained characterization of ST-rank for HD-walks. In \Cref{sec:hdx-spectra} we show how to explicitly compute the approximate eigenvalue $\lambda_i(M)$ corresponding to each strip based upon the structure of $M$. While the full result requires further background, its specification to canonical and partial-swap walks has a nice combinatorial interpretation. Recall that the canonical walk $N_k^i$ walks between $k$-faces through a shared $(k+i)$-face. We say it has \textbf{depth} $i/(k+i)$ since it traverses $i$ of the $k+i$ relevant levels of the complex. Similarly, we say the partial-swap walk $S_k^i$ has depth $i/k$, since it swaps $i$ out of $k$ elements (up to factors in $\gamma$, this can also be viewed as walking between $k$ faces via a shared $k-i$ face, where one traverses $i$ out of the $k$ relevant levels). We prove that the stripped eigenvalues corresponding to the canonical and partial-swap walks decay exponentially fast with a base rate dependent on depth.
\begin{theorem}[ST-rank of HD-walks (\Cref{prop:canon-spectra-HDX} + \Cref{cor:swap-HDX-spectra})]\label{thm:intro-hdx-swap-canon}
Let $M$ be a canonical or partial-swap walk of depth $0 \leq \beta \leq 1$ on a sufficiently strong two-sided local-spectral expander. Then the eigenvalues corresponding to the eigenstrips of $M$ decay exponentially fast:
\[
\lambda_i(M) \leq e^{-\beta i}.
\]
The ST-Rank $R_\delta(M)$ is then at most:
\[
R_{\delta}(M) \leq \frac{\ln\left(\frac{1}{\delta}\right)}{\beta}
\]
\end{theorem}
In other words, walks that reach deep into the complex (e.g. $N^{k/2}_k$) have constant ST-Rank, whereas shallow walks like $N^1_k$ have ST-rank $\Omega(k)$.
\subsection{Characterizing (non)-Expansion in HD-Walks}
We now move to our main structural application, characterizing edge expansion in HD-walks. Edge expansion is a fundamental combinatorial property of graphs with applications across many areas of theoretical computer science, including (as we will soon discuss) both hardness and algorithms for unique games. For a subset of vertices $S \subseteq V$, edge expansion measures the (normalized) fraction of edges which leave $S$.
\begin{definition}[Edge Expansion]
Given a graph $G(V,E)$, the edge expansion of a subset $S \subset V$ is:
\[
\phi(G,S) = \frac{E(S,V \setminus S)}{E(S,V)}
\]
where $E(S,T)$ counts the number of edges crossing from $S$ to $T$ (double-counting edges in the intersection). When convenient, we denote non-expansion, $1-\phi(G,S)$, as $\bar{\phi}(G,S)$, and drop $G$ from the notation when clear from context.
\end{definition}
It is often useful to characterize exactly which sets in a graph expand. The two most widely used characterizations are when all\footnote{By this we really mean all sets of size at most $|V|/2$.} sets expand (expanders), or when all \textit{small} sets expand (small-set expanders). However, many important structures fall outside such a simple characterization. The Johnson and Grassmann graphs, for instance, are well known to have small non-expanding sets. Characterizing the structure of non-expansion in these graphs was crucial not only for the resolution of the 2-2 Games Conjecture \cite{subhash2018pseudorandom}, but also for recent algorithms for unique games \cite{bafna2020playing}.

Using the spectral machinery developed in the previous section, we give a tight
characterization of expansion in HD-walks. As discussed, we focus mainly on two key aspects: the expansion of \emph{links}, and the structure of non-expanding sets. We start with the former, where links have long been the prototypical example of small, non-expanding sets for the Johnson graphs. In fact, we prove this stems from a much stronger connection between local expansion and spectral structure in HD-walks.
\begin{theorem}[Local (non)-Expansion = Global Spectrum (Informal \Cref{thm:local-vs-global})]\label{thm:new-intro-local-vs-global}
Let $M$ be a $k$-dimensional HD-walk on a sufficiently strong $d$-dimensional two-sided $\gamma$-local-spectral expander with $d>k$,\footnote{The lower bound on non-expansion still holds for $d=k$, which is sufficient for most applications of this result. It's also worth noting that the condition $d>k$ itself is generally trivially satisfied, since all known two-sided local-spectral expanders are in fact cutoffs of larger complexes.} and let $\lambda_i(M)$ be approximate eigenvalues corresponding to $M$'s $k+1$ eigenstrips. For all $0 \leq i \leq k$ and $\tau \in X(i)$, let $X_\tau$ denote the set of $k$-faces that contain $\tau$. Then $X_\tau$ is expanding if and only if $\lambda_i(M)$ is small:
\[
\bar{\phi}(X_\tau) \in \lambda_{i}(M) \pm O_k(\gamma).
\]
\end{theorem}
Since HD-walks (like the Johnson graphs) are generally poor spectral expanders, \Cref{thm:new-intro-local-vs-global} implies the existence of small, non-expanding sets (namely $1$-links\footnote{Since we generally consider the regime where $|X(1)| \gg k$, a basic averaging argument gives that there must exist a $o(1)$-size $1$-link. One may also note that in a $\gamma$-local-spectral expander, no $1$-link can have more than $O_k(\gamma)$ mass, so all links are small if $\gamma \leq o(1)$.}). Our stronger characterization of local expansion is of independent interest beyond this simple corollary, however, due to its important connections to local-to-global algorithms and ST-rank. In particular, essentially all work on high-dimensional expanders relies on a strategy known as the \textit{local-to-global paradigm}, where a global property (e.g. mixing, agreement testing, etc.) is reduced to examining a local version on links. These arguments often reduce to showing that, in some relevant sense, interaction \textit{between} links is minimal. \Cref{thm:new-intro-local-vs-global} offers exactly such a statement for levels of the complex which correspond to eigenspaces with large eigenvalues---since links at these levels are non-expanding, they don't have much interaction. This gives a holistic local-to-global approach for spectral graph algorithms on HDX, since the eigenspaces corresponding to \textit{small} eigenvalues are traditionally easy to handle through other means. Later we will see an algorithmic application of this approach to unique games, where the non-expansion of links allows us to patch together local solutions, and the connection to global spectra ties performance of the algorithm to ST-rank.

Returning to the topic of expansion, \Cref{thm:new-intro-local-vs-global} shows that as long as $M$ is not a spectral expander, there always exist small, non-expanding sets in the form of links. This raises a natural question: are \textit{all} non-expanding sets explained by links? Before answering, let's spend a little time formalizing this. Given a subset $S \subset X(k)$, let $L_{S,i}$ be a function on $i$-faces measuring the deviation of $S$ from its average across links at level $i$:
\[
\forall \tau \in X(i): L_{S,i}(\tau) = \underset{X_\tau}{\mathbb{E}}[\mathbbm{1}_S]-\mathbb{E}[\mathbbm{1}_S].
\]
The statement ``non-expansion is explained by links'' then really amounts to saying that we expect $L_{S,i}$ to be far from $0$ in some sense if $S$ is non-expanding. There are a number of reasonable ways to formalize this notion, each with its own use. Following \cite{bafna2020playing}, we mainly focus on an $\ell_2$-variant that is particularly useful for designing local-to-global algorithms: if $S$ is non-expanding, then $\norm{L_{S,i}}_2^2$ (or equivalently the \textit{variance} of $S$ across $i$-links) must be large. It is perhaps easier to think about this result in terms of its contrapositive. Call a set \textbf{$\ell_2$-pseudorandom} if its variance across $i$-links is \textit{small}.
\begin{definition}[$\ell_2$-Pseudorandom Sets]
Let $X$ be a pure simplicial complex. We call a subset $S \subset X(k)$ $(\varepsilon,\ell)$-$\ell_2$-pseudorandom if
\[
\forall i \leq \ell: \norm{L_{S,i}}_2^2 \leq \varepsilon\mathbb{E}[\mathbbm{1}_S]
\]
\end{definition}
We remark briefly on the choice of normalization by $\mathbb{E}[\id{S}]$ on the right. While mostly a formality, we will see this is in fact a natural choice both in comparison to the $\ell_\infty$-variant (which differs in some natural sense by at least a factor of $\mathbb{E}[\id{S}]$), and when considering expansion which is itself normalized by $\mathbb{E}[\id{S}]$. Indeed with this in mind, we prove that $(\varepsilon,\ell)$-$\ell_2$-pseudorandom sets expand near-perfectly. 
\begin{theorem}[$\ell_2$-Pseudorandom Sets Expand (Informal \Cref{thm:hdx-expansion})]\label{intro:new-l2-pr-sets-expand}
Given a $k$-dimensional HD-walk $M$ on a sufficiently strong two-sided $\gamma$-local-spectral expander and small constants $\alpha,\delta,\varepsilon>0$, we have that any $(\varepsilon,R_{\delta}(M))$-pseudorandom set $S$ of density $\alpha$ expands near-perfectly:
\[
\phi(S) \geq 1 - \alpha - \delta - O_k(\varepsilon) - O_k(\gamma)
\]
\end{theorem}
We note that (the formal version of) this result is exactly tight (in the limit of $\gamma \to 0$). Passing back to the original (now contrapositive) form of the statement, \Cref{intro:new-l2-pr-sets-expand} immediately implies that any non-expanding set must have large variance across links at a level determined by the ST-rank of $M$. On the algorithmic side, this informally translates to the statement that most interesting structure lies on low-level links, which is crucial to any local-to-global algorithm.

Before formalizing this algorithmic intuition in the case of unique games, it's worth taking a moment to consider the implication of \Cref{intro:new-l2-pr-sets-expand} to a different formalization of this problem with recent applications to hardness of approximation: the $\ell_\infty$-variant. In this characterization, the (squared) $\ell_2$-norm of $L_{S,i}$ is replaced with its maximum. In other words, the $\ell_\infty$-variant characterization posits that every non-expanding set must be highly concentrated in some \textit{individual} link. We prove that \Cref{intro:new-l2-pr-sets-expand} actually holds in this regime as well by a simple reduction. If we analogously define $\ell_\infty$-pseudorandom sets (see \Cref{def:pseudorandom-infty}), it is not hard to show that any $\ell_\infty$-pseudorandom set must also be $\ell_2$-pseudorandom. This results in a tight version of \Cref{intro:new-l2-pr-sets-expand} for the $\ell_\infty$-regime, but only when $\max(L_{S,i})$ is close to $\frac{1}{\mathbb{E}[\id{S}]}\norm{L_{S,i}}^2$. Practically, the interesting regime in hardness of approximation is when these two quantities are far apart. In this case our $\ell_\infty$ to $\ell_2$ reduction necessarily loses important factors, so we cannot recover any results from the hardness of approximation literature.

Nevertheless, it is worth stating that as an immediate corollary of our $\ell_2$-based analysis we get an $\ell_\infty$-characterization of expansion in HD-walks: any non-expanding set must be non-trivially concentrated in a link.
\begin{corollary}[Non-expanding Sets Correlate with Links (Informal \Cref{cor:hdx-non-expansion})]\label{cor:new-non-expansion}
Let $M$ be a $k$-dimensional HD-walk on a sufficiently strong two-sided $\gamma$-local-spectral expander. Then if $S \subset X(k)$ is a set of density $\alpha$ and expansion at most:
\[
\phi(S) < 1 - \alpha - \delta - O_k(\gamma)
\]
for some $\delta>0$, $S$ must be non-trivially correlated with some $i$-link for $1 \leq i \leq r=R_{\delta/2}(M)$:
\[
\exists 1 \leq i \leq r, \tau \in X(i): \underset{X_\tau}{\mathbb{E}}[\mathbbm{1}_S] \geq \alpha + \frac{c_{\delta,r}}{{k \choose i}}
\]
where $c_{\delta,r}$ depends only on $\delta$ and the ST-rank $r$.
\end{corollary}
Proving a $k$-independent $\ell_\infty$-characterization for HD-walks, either directly or via a stronger $k$-dependent reduction to the $\ell_2$-regime, remains an interesting open problem. Such a characterization is known on the Johnson graphs (and a number of related objects \cite{keevash2019hypercontractivity,filmus2020hypercontractivity}), and may shed light on similar structure in the Grassmann graphs used to prove the 2-2 Games Conjecture \cite{subhash2018pseudorandom}. We discuss some additional subtleties in this direction in \Cref{sec:q-eposet}.

\subsection{Playing Unique Games}\label{sec:results-UG}
One motivation for studying spectral structure and non-expansion in HD-walks stems from a simple class of 2-CSPs known as unique games, a central object of study in hardness-of-approximation since Khot's introduction of the Unique Games Conjecture (UGC) \cite{khot2002power} nearly 20 years ago. We study affine unique games which are known to be as hard as unique games~\cite{khot2007optimal}:

\begin{definition}[Affine Unique Games]\label{def:ug}
An instance $I = (G,\mathcal{S})$ of affine unique games over alphabet $\Sigma = \{0,\ldots,m-1\}$ is a weighted undirected graph $G(V,E)$, and set of affine shifts $\mathcal{S} =\{s_{uv} \in \Sigma\}_{(u,v) \in E}$. The value of $I$, $\text{val}(I)$, is the maximum fraction of satisfied constraints over all possible assignments $\Sigma^V$:
\[\max_{X \in \Sigma^V}\underset{(u,v) \sim E}{\Pr}[X_{u} - X_v = s_{uv} (\text{mod } |\Sigma|)],
\]
where edges are drawn corresponding to their weight. For an individual assignment $X$, we refer to this expectation as $\text{val}_I(X)$. Any weighted undirected graph $G$ is uniquely associated with a self-adjoint random-walk matrix $M$ where $M(u,v) = \Pr_E[(u,v)]/\sum_{w}\Pr_E[(u,w)]$ (see \Cref{sec:ug-rwalk}). We will usually view the constraint graph in this manner instead, referring to unique games instances over random-walks as $I = (M,\Pi)$.
\end{definition}
Informally, the UGC states that for sufficiently small constants $\varepsilon,\delta$, there exists an alphabet size such that distinguishing between instances of unique games with value $1-\varepsilon$ and $\delta$ is NP-hard. A positive resolution to the UGC would resolve the hardness-of-approximation of many important combinatorial optimization problems, including CSPs \cite{raghavendra2008optimal}, vertex-cover \cite{khot2008vertex}, and a host of others such as \cite{khot2002power,khot2007optimal,guruswami2008beating,chitta2011approximate,khot2014characterization,khot2015unique}.

Building on the recent framework of Bafna, Barak, Kothari, Schramm, and Steurer \cite{bafna2020playing}, we show how our structural theorems combine with the notion of ST-rank to give efficient algorithms for unique games over HD-walks.
\begin{theorem}[Playing Unique Games on HD-walks (Informal \Cref{thm:unique-games})]\label{thm:informal-intro-unique-games}
For any $\varepsilon \in (0,.01)$, there exists an algorithm $\mathcal{A}$ with the following guarantee. If $\mathcal{I} = (M,\mathcal{S})$ is an instance of affine unique games with value at least $1-\varepsilon$ over $M$, a complete $k$-dimensional HD-walk on a $d$-dimensional two-sided $\gamma$-local-spectral expander with $\gamma \ll o_k(1)$ and $d>k$,\footnote{Formally, we note $\gamma$ also has dependence on $M$ (see \Cref{thm:unique-games}), though this can be removed in special cases like the Johnson scheme. We also note again that $d>k$ is generally a trivial condition for strong two-sided local-spectral expanders.} then $\mathcal{A}$ outputs a $\poly(\tau)$-satisfying assignment in time $|X(k)|^{\poly(1/\tau)}$, where $\tau = \left(\frac{\varepsilon}{{k \choose r(\varepsilon)}}\right)$ and $r(\varepsilon)=R_{1-O(\varepsilon)}(M)$ is the ST-rank of $M$.
\end{theorem}
It is worth giving a quantitative comparison of this result to the best known algorithms based on threshold rank \cite{arora2015subexponential,barak2011rounding,guruswami2011lasserre}. These algorithms run in time roughly exponential in the $(1 - O(\eps))$-threshold rank of $G$ to get similar soundness guarantees when $k=O(1)$.\footnote{This corresponds to the setting in which the dimension of the HD-walk is fixed, and the number of vertices in the complex grows.} Due to the poor threshold rank of HD-walks, this generally amounts to nearly exponential time ($2^{|X(k)|^{\poly(\eps)}}$), whereas our algorithm runs in time roughly $|X(k)|^{2^{O(k)}}$. In the regime of constant $k$, \Cref{thm:informal-intro-unique-games} therefore gives a polynomial time algorithm for such instances achieving an exponential improvement over threshold-rank based algorithms. We describe a few examples of such families below.

Since \Cref{thm:informal-intro-unique-games} can be a bit hard to interpret without additional knowledge of the high dimensional expansion literature, we end our discussion of algorithms for unique games with a few concrete examples. Perhaps the most basic example generalizing the Johnson graphs that fits into our framework are slices of the $q$-ary noisy hypercube. For constant-level slices, we improve over BRS from nearly exponential to polynomial. Further, the result is robust in the sense that it continues to hold even when the underlying complex is perturbed. This results in algorithms, for instance, for dense random sparsifications of these slices, or slices taken from a negatively correlated distribution over $\mathbb{Z}_q^n$ (or more generally slices of any spectrally-independent spin-system \cite{anari2020spectral}). Another interesting class of graphs we see significant speed-ups on are algebraic sparsifications of the Johnson graphs stemming from bounded-degree constructions of two-sided local-spectral expanders such as (cutoffs of) Lubotsky, Samuels, and Vishne's \cite{lubotzky2005explicit} Ramanujan complexes, or Kaufman and Oppenheim's \cite{kaufman2018construction} coset complex expanders. The resulting speedups on this class of graphs is somewhat more surprising than the above, since they exhibit substantially different structure in other aspects (they are, for instance, of bounded degree unlike the Johnson graphs).

\subsection{Connections to Hardness and the Grassmann Graphs}\label{sec:q-eposet}
We finish the discussion of our results by taking a deeper look at connections with recent progress on the UGC, and argue that our framework opens an avenue for further progress. The resolution of the 2-2 Games Conjecture \cite{subhash2018pseudorandom} hinged on a characterization of non-expanding sets on the Grassmann graph not dissimilar to what we have shown for two-sided local-spectral expanders. While we have focused above on HDX which are simplicial complexes, our work extends to a broader set of objects introduced by DDFH \cite{dikstein2018boolean} called expanding posets which includes class of objects includes expanding subsets of the Grassmann poset we call $q$-eposets (we refer the reader to \cite{dikstein2018boolean} for definitions and further discussion). 


\Cref{intro:new-l2-pr-sets-expand} and \Cref{cor:new-non-expansion} extend naturally to HD-walks on $q$-eposets. We state the latter result here since it follows without too much difficulty from the same arguments as in this paper, but the full details (and further generalizations to expanding posets) will appear in a companion paper.
\begin{corollary}[Non-expansion in $q$-eposet]\label{cor:intro-non-expansion-grassmann}
Let $(X,\Pi)$ be a two-sided $\gamma$-$q$-eposet with $\gamma$ sufficiently small, $M$ a $k$-dimensional HD-walk on $(X,\Pi)$. Then if $S \subset X(k)$ is a set of density $\alpha$ and expansion:
\[
\phi(S) < 1 - \alpha - O_{q,k}(\gamma) - \delta
\]
for some $\delta>0$ and $r=R_{\delta/2}(M)$, $S$ must be non-trivially correlated with some $i$-link for $1 \leq i \leq r$:
\[
\exists 1 \leq i \leq r, \tau \in X(i): \underset{X_\tau}{\mathbb{E}}[\mathbbm{1}_S] \geq \alpha + \frac{c_{\delta,r}}{{k \choose i}_q}
\]
where $c_{\delta,r}$ depends only on $\delta$ and the ST-rank $r$, and ${k \choose i}_q$ is the standard q-binomial coefficient:
\[
{k \choose i}_q = \prod\limits_{j=0}^{i-1}\frac{q^{k-j}-1}{q^{i-j}-1},
\]
\end{corollary}
Since the Grassmann graphs are simply partial-swap walks on the Grassmann poset,\footnote{Seeing that they are HD-walks is non-trivial, and follows from the $q$-analog of work in \cite{alev2019approximating}.} \Cref{cor:intro-non-expansion-grassmann} provides a direct connection to the proof of the 2-2 Games Conjecture \cite{subhash2018pseudorandom}. Unfortunately, due to the dependence on $k$, this result is too weak to recover the proof. 

The dependence of both \Cref{cor:new-non-expansion} and \Cref{cor:intro-non-expansion-grassmann} on $k$ is a subtle but important point, so we'll finish the section by discussing it in a bit more detail. The particular dependence we get in the $\ell_2$-regime, ${k \choose i}$ for the Johnson and ${k \choose i}_q$ for the Grassmann, is tight and can be understood by examining the variance of $i$-links. While $i$-links are the prototypical example of an $\Omega(1)$-pseudorandom set in the $\ell_\infty$ regime, one can show that they are actually about $1/{k \choose i}$-pseudorandom in the $\ell_2$-regime (or $1/{k \choose i}_q$ for the Grassmann). Since links are non-expanding, any bound in the $\ell_2$-regime must have matching dependence on $k$ to make up for this fact. Indeed, one can use the same argument to show that a $k$-independent bound cannot exist in $\ell_2$-regime (even if we relax dependence on pseudorandomness, see \Cref{prop:ell2-k-dependence}). As a result, any $\ell_\infty$ to $\ell_2$ reduction like ours (which has no dependence on $k$) will always result in a final bound depending on $k$.

Finally, it's worth noting that while the $\ell_\infty$-regime escapes this particular issue since links are $\Omega(1)$-pseudorandom (and indeed that for the Johnson graphs, a $k$-independent bound is known \cite{khot2018small}), there is an additional consideration for the Grassmann graphs: there exist small, non-expanding $\ell_\infty$-pseudorandom sets \cite{dinur2018non}. The proof of the 2-2 Games Conjecture therefore relies on a finer-grained definition of local structure than links (called zoom-in zoom-outs) \cite{subhash2018pseudorandom}. While we cannot hope to apply exactly the same techniques to analyze this variant, we view our method's generality and simplicity as evidence that a deeper understanding of higher order random walks may be key to further progress on the UGC.

\section{Proof Overview}\label{sec:proof-overview}
In this section we give a proof overview of our main results. Throughout the section we will assume the complex is endowed with a uniform distribution, and will (usually) ignore error terms in the spectral parameter $\gamma$. The full details for weighted complexes and a careful treatment of the error terms is given in the main body along with a number of further generalizations.

\subsection{Eigenstripping and the HD-Level-Set Decomposition}
We start with a discussion of the techniques underlying \Cref{thm:new-intro-hdx-spectra} (Eigenstripping). At its core, this (and indeed all of) our results rely upon now-standard machinery for working on simplicial complexes called the averaging operators.
\begin{definition}[The Averaging Operators]
Let $X$ be a $d$-dimensional pure simplicial complex. For any $0 \leq k \leq d$, denote the space of functions $f: X(k) \to \mathbb{R}$ by $C_k$. The ``Up'' operator lifts $f \in C_k$ to $U_kf \in C_{k+1}$:
\[
\forall y \in X(k+1): U_kf(y) = \frac{1}{k+1}\sum\limits_{x \in X(k): x \subset y}f(x).
\]
The ``Down'' operator lowers functions $f \in C_{k+1}$ to $D_{k+1} f \in C_k$:
\[
\forall x \in X(k): D_{k+1}f(x) = \frac{1}{n-k}\sum\limits_{y \in X(k+1): y \supset x} f(y).
\]
Since it is often useful to compose these operators, we will use the shorthand $U^k_i=U_{k-1}\ldots U_i$ and $D^k_i = D_{i+1}\ldots D_{k}$ to denote the composed averaging operators which raise and lower functions between $C_i$ and $C_k$ by averaging.
\end{definition}
In fact, the averaging operators are crucial even to \textit{defining} higher order random walks. We discuss this definition in greater detail in \Cref{sec:prelims}, and for now settle for noting that a basic class of HD-walks are given simply by composing $U$ and $D$. For instance, the canonical walk $N_k^i$ is the composition $D_k^{k+i}U^{k+i}_k$, and the partial swap walks $S_k^i$ turn out to be an affine combinations of the $N_k^i$ \cite{alev2019approximating}.

The proof of \Cref{thm:new-intro-hdx-spectra} relies on a useful decomposition for functions on simplicial complexes based upon the averaging operators we call the \textit{HD-level-set Decomposition} recently introduced by Dikstein, Dinur, Filmus, and Harsha \cite{dikstein2018boolean}. The idea is to break functions into components coming from each level of the complex, lifted by averaging to the top level.
\begin{theorem}[HD-Level-Set Decomposition, Theorem 8.2 \cite{dikstein2018boolean}]\label{thm:decomp-ddfh}
Let $X$ be a $d$-dimensional two-sided $\gamma$-local-spectral expander, $\gamma < \frac{1}{d}$, $0 \leq k \leq d$, and let: 
\[
H^0=C_0, H^i=\text{Ker}(D_i), V_k^i = U^{k}_iH^i.
\]
Then:
\[
C_k = V^0_k \oplus \ldots \oplus V^k_k.
\]
In other words, every $f \in C_k$ has a unique decomposition $f=f_0+\ldots+f_k$ such that $f_i=U^{k}_ig_i$ for $g_i \in \text{Ker}(D_i)$.
\end{theorem}
The HD-Level-Set Decomposition is particularly useful not only for its rigid combinatorial structure, but also its spectral properties. Namely, one familiar with the Johnson Scheme might notice that when $X$ is the complete complex, this decomposition exactly gives the eigenspaces of the Johnson Scheme. DDFH gave an approximate extension of this result to the basic ``upper'' walk $D_{k+1}U_k$, proving that the $V_k^i$ are approximate eigenspaces.

The proof of \Cref{thm:new-intro-hdx-spectra} follows from combining an extension of this result to all HD-walks (see \Cref{prop:hdx-eig-vals}) with \Cref{thm:intro-approx-ortho}, which states that any approximate eigendecomposition strictly controls the spectrum of the underlying operator. However, since the proof of \Cref{thm:intro-approx-ortho} is purely linear algebraic (and is essentially irrelevant for understanding our core results on expansion and unique games), we defer detailed discussion of the result to \Cref{sec:linear-algebra}.
\subsection{Characterizing Expansion}\label{sec:overview-expansion}
We now take a look at how this combinatorial understanding of the spectra of higher order random walks allows us to characterize their edge-expansion. We start with a standard observation, the expansion of a set can be written as the inner product of its indicator function. In other words, if $X$ is a simplicial complex and $M$ is an HD-walk on $X(k)$, the expansion of any $S \subset X(k)$ with respect to $M$ may be written as:
\[
\phi(M,\mathbbm{1}_S) = 1 - \frac{1}{\mathbb{E}[\mathbbm{1}_S]}\langle \mathbbm{1}_S,M\mathbbm{1}_S \rangle.
\]
The key is then to notice that by the bilinearity of inner products, we can expand the right-hand side in terms of the HD-Level-Set Decomposition. In particular, writing $\mathbbm{1}_S=\sum\limits_{i=0}^k \mathbbm{1}_{S,i}$ for $\mathbbm{1}_{S,i} \in V_k^i$, we have:
\[
\phi(M,\mathbbm{1}_S) = 1 - \frac{1}{\mathbb{E}[\mathbbm{1}_S]}\sum\limits_{i=0}^k \langle \mathbbm{1}_S, M \mathbbm{1}_{S,i}\rangle.
\]
Finally, since we know each $V_k^i$ is approximately an eigenstrip concentrated around some $\lambda_i$, up to additive error in $\gamma$ we can simplify this to:
\begin{equation}\label{eq:overview-expansion}
    \phi(M,\mathbbm{1}_S) \approx 1 - \frac{1}{\mathbb{E}[\mathbbm{1}_S]}\sum\limits_{i=0}^k \lambda_i \langle \mathbbm{1}_S, \mathbbm{1}_{S,i} \rangle.
\end{equation}
Thus we see that understanding the expansion of $S$ comes down to the interplay between its projection onto each level of the complex and their corresponding approximate eigenvalues. Our characterization of expansion combines this observation with the combinatorial and spectral structure of the HD-Level-Set Decomposition. We devote the rest of the section to sketching the proofs of \Cref{thm:new-intro-local-vs-global} (local expansion vs global spectra), \Cref{intro:new-l2-pr-sets-expand} (pseudorandom sets expand), and our reduction from the $\ell_\infty$-variant to $\ell_2$-variant.

\paragraph{Proof sketch of \Cref{thm:new-intro-local-vs-global}:} Recall that \Cref{thm:new-intro-local-vs-global} shows a tight inverse relation between the expansion of links and the spectra of $M$. Given $\tau \in X(j)$, we wish to examine the indicator function $\mathbbm{1}_{X_\tau}$ of the link $X_\tau$. First, notice that since $\mathbbm{1}_{X_\tau} = {k \choose j}U^k_j\mathbbm{1}_\tau$, it's easy to see that $\mathbbm{1}_{X_\tau} \in V_k^0 \oplus \ldots \oplus V_k^j$. We prove something stronger, that $\mathbbm{1}_{X_\tau}$ in fact lies almost entirely in $V_k^j$. In particular, we show that for every $i \neq j$, $\mathbbm{1}_{X_\tau}$ has almost no projection onto $V_k^i$:
\[
\langle \mathbbm{1}_{X_\tau}, \mathbbm{1}_{X_\tau,i}\rangle \leq O_k(\gamma)\mathbb{E}[\mathbbm{1}_{X_\tau}].
\]
Since the HD-Level-Set Decomposition is approximately orthogonal \cite{dikstein2018boolean} (more generally, this is true of any approximate eigendecomposition, see \Cref{lemma:approx-orthog}), this implies that the mass on level $j$ is around $(1 \pm O_k(\gamma)) \mathbb{E}[\mathbbm{1}_{X_\tau}]$, and plugging these observations into \Cref{eq:overview-expansion} we get:
\begin{align*}
        \phi(M,\mathbbm{1}_S) &\approx 1 - \frac{1}{\mathbb{E}[\mathbbm{1}_S]}\sum\limits_{i=0}^k \lambda_i \langle \mathbbm{1}_S, \mathbbm{1}_{S,i} \rangle\\
        &\approx 1 - \frac{\lambda_i}{\mathbb{E}[\mathbbm{1}_S]}\langle \mathbbm{1}_S, \mathbbm{1}_{S,i} \rangle\\
        &\approx 1 - \lambda_i
\end{align*}
where we have ignored some $O_k(\gamma)$ error terms.
\paragraph{Proof sketch of \Cref{intro:new-l2-pr-sets-expand}:} Proving that pseudorandom functions expand is a bit more involved, but still at its core revolves around the analysis of \Cref{eq:overview-expansion}. In particular, since the approximate eigenvalues of $M$ monotonically decrease (i.e.\ $ \forall i \geq j$, $\lambda_i(M) \leq \lambda_j(M)$) then for any $\delta>0$ we can simplify \Cref{eq:overview-expansion} to:
\begin{equation}\label{eq:overview-exp-2}
    \phi(M,\mathbbm{1}_S) \lesssim 1 - \frac{1}{\mathbb{E}[\mathbbm{1}_S]}\sum\limits_{i=0}^{R_\delta(M)-1} \lambda_i \langle \mathbbm{1}_S, \mathbbm{1}_{S,i} \rangle + \delta,
\end{equation}
where we recall $R_\delta(M)$ is the ST-rank, denoting the number of eigenstrips of $M$ with eigenvalues greater than $\delta$. Using \Cref{thm:new-intro-hdx-spectra}, it is possible to show (see \Cref{prop:eig-decrease}) that essentially all HD-walks of interest satisfy this property. With this in hand, proving \Cref{intro:new-l2-pr-sets-expand} boils down to characterizing the projection of a set $S$ onto low levels of the complex by its behavior on links.
\begin{theorem}[Pseudorandomness bounds low-level weight (Informal \Cref{thm:body-local-spec-proj})]\label{thm:overview-pr-projection}
Let $X$ be a sufficiently strong $d$-dimensional $\gamma$-local-spectral expander. For any $i \leq k \leq d$ let $S \subset X(k)$ be a $(\varepsilon,i)$-$\ell_2$-pseudorandom set. Then for all $j \leq i$:
\[
\langle \mathbbm{1}_S, \mathbbm{1}_{S,j} \rangle \lesssim {k \choose i}\varepsilon\mathbb{E}[\mathbbm{1}_S]
\]
\end{theorem}
\Cref{intro:new-l2-pr-sets-expand} follows immediately from plugging this result into \Cref{eq:overview-exp-2}. Perhaps surprisingly, \Cref{thm:overview-pr-projection} itself follows without too much difficulty from combining analysis of the averaging operators and our spectral analysis of the HD-Level-Set Decomposition (namely \Cref{thm:new-intro-hdx-spectra}). In particular, given a set $S \subset X(k)$, the idea is to examine the lowered indicator function $D^k_j\mathbbm{1}_S$, which exactly gives the expectation of $S$ over $j$-links, that is for any $\tau \in X(j)$:
\[
D^k_j\mathbbm{1}_S(\tau) = \underset{X_\tau}{\mathbb{E}}[\mathbbm{1}_S].
\]
Proving \Cref{intro:new-l2-pr-sets-expand} then corresponds to lower bounding $\text{Var}(D^k_j\mathbbm{1}_S)$ by some function of the projection of $S$ onto level $j$. In fact, it is possible to exactly understand this connection. The idea is to reduce the problem to a spectral analysis of HD-walks using the adjointness of $D^k_j$ and $U^k_j$ (see e.g.\ \cite{dikstein2018boolean}):
\begin{align*}
    \text{Var}(D^k_j\mathbbm{1}_S) &= \langle D^k_j\mathbbm{1}_S, D^k_j\mathbbm{1}_S \rangle - \mathbb{E}[D^k_j\mathbbm{1}_S]^2\\
    &= \langle \mathbbm{1}_S, U^k_jD^k_j\mathbbm{1}_S\rangle - \mathbb{E}[\mathbbm{1}_S]^2\\
    &=\sum\limits_{\ell=0}^k \langle \mathbbm{1}_S, U^k_jD^k_j\mathbbm{1}_{S,\ell} \rangle - \mathbb{E}[\mathbbm{1}_S]^2\\
    &=\sum\limits_{\ell=1}^k \langle \mathbbm{1}_S, U^k_jD^k_j\mathbbm{1}_{S,\ell} \rangle.
\end{align*}
Since $U^k_jD^k_j$ is an HD-walk, \Cref{thm:new-intro-hdx-spectra} implies that $U^k_jD^k_j\mathbbm{1}_{S,\ell} \approx \lambda_\ell \mathbbm{1}_{S,\ell}$ for some approximate eigenvalue $\lambda_\ell$. In the formal version of \Cref{thm:new-intro-hdx-spectra} we compute the approximate eigenvalues of such walks, and in this case a simple computation shows that $\lambda_\ell = \frac{{j \choose \ell}}{{k \choose \ell}}$. As a result, we get a tight connection between $\text{Var}(D^k_j\mathbbm{1}_S)$ and the projection of $S$ onto low levels of the complex:
\begin{align*}
    \text{Var}(D^k_j\mathbbm{1}_S)
    &\approx \sum\limits_{\ell=1}^k \frac{{j \choose \ell}}{{k \choose \ell}}\langle \mathbbm{1}_S,\mathbbm{1}_{S,\ell} \rangle
\end{align*}
where we have ignored factors in $\gamma$. The proof then follows from noting that the right-hand side is approximately lower bounded by $\frac{1}{{k \choose j}}\langle \mathbbm{1}_S, \mathbbm{1}_{S,j} \rangle$ by near orthogonality of the HD-Level-Set decomposition \cite{dikstein2018boolean}. 

\paragraph{$\ell_\infty$ to $\ell_2$ reduction:}
We end with a simple observation that generalizes our results to the $\ell_\infty$-variant characterization via reduction. The idea is that the 2-norm may be re-written as a (normalized) expectation over a modified distribution:
\[
\frac{1}{\mathbb{E}[\mathbbm{1}_S]}\text{Var}(D^k_j\mathbbm{1}_S)
= \underset{\Pi_S}{\mathbb{E}}[D^k_j \mathbbm{1}_S] - \mathbb{E}[\mathbbm{1}_S] \leq \max(D^k_j\mathbbm{1}_S - \mathbb{E}[\mathbbm{1}_S])
\]
where $\Pi_S$ is the distribution on $X(j)$ naturally induced by restricting $\Pi$ to the support of $S$ (see \Cref{sec:pseudorandomness} for further details). It immediately follows that any $(\varepsilon,\ell)$-$\ell_\infty$-pseudorandom set is also $(\varepsilon,\ell)$-$\ell_2$-pseudorandom.

\subsection{Playing Unique Games}\label{sec:ug-proof-overview}
We end the section with a discussion of \Cref{thm:informal-intro-unique-games}, an application of our structural theorems to algorithms for unique games. We follow the algorithmic framework of \cite{bafna2020playing} for playing unique games on Johnson graphs, a special case of HD-walks on the complete complex. We generalize their algorithm and its analysis to HD-walks over all (sufficiently strong) two-sided local spectral expanders by abstracting the algorithm into the broader local-to-global paradigm for HDX:
\begin{enumerate}
    \item For some $r$, break the UG instance down into sub-instances over $r$-links of the complex.
    \item Solve the local sub-instances and patch together the solutions into a good global solution.\footnote{This is only informally speaking, as the actual algorithm is iterative and repeats these two steps many times to obtain a good solution for the whole graph}
\end{enumerate}
For their underlying algorithmic machinery, BBKSS rely heavily on a well-studied paradigm known as the Sum-of-Squares (SoS) semidefinite programming hierarchy. We won't go into too much detail about this paradigm here (see \Cref{sec:SoS-background} for details), and for the moment it suffices to say that analysing SoS algorithms relies on porting the proofs of relevant inequalities (specific to the problem at hand) into the Sum-of-Squares proof system, a restricted proof system for proving polynomial inequalities. We leverage the BBKSS framework to solve unique games on HDX by observing that their algorithm's analysis can be generalized to rely on two core structural properties true of the graphs underlying HD-walks:
\begin{enumerate}
    \item A low-degree Sum-of-Squares proof that non-expanding sets have high variance in size across links.
    \item For every small enough $\eps$, the existence of some $r=r(\varepsilon)$ such that:
    \begin{enumerate}
        \item The $(r+1)$-st largest (distinct) stripped-eigenvalue of the constraint graph is small:
        \[ 
        \lambda_{r}\leq 1-\Omega(\varepsilon)
        \]
        \item The expansion of any $r$-link is small as well:
        \[
        \forall \tau \in X(r): \Phi(X_\tau) \leq O(\varepsilon).
        \]
    \end{enumerate}
\end{enumerate}
In essence, properties (1) and (2a) ensure that there are good local solutions at level $r$, and property (2b) ensures that these solutions can be patched together without too much loss. Given the above properties, the proof of \Cref{thm:informal-intro-unique-games} follows the BBKSS analysis framework, but is more technical as properties that hold for Johnson graphs generally only hold approximately for HD-walks. 

The novelty in our analysis lies mostly in proving these two properties. Luckily, we have already done most of the work! Property (1) is simply an SoS variant of \Cref{intro:new-l2-pr-sets-expand}, which we prove by developing SoS versions of the now standard machinery for the HD-Level-Set Decomposition from \cite{dikstein2018boolean}. The proof then follows essentially as discussed in \Cref{sec:overview-expansion}. The second property is slightly more subtle. This parameter, found in \cite{bafna2020playing} by direct computation on the Johnson graphs, determines both the soundness and runtime of their algorithm. In fact, our framework completely demystifies its existence: since $\lambda_r$ and the expansion of $r$-links are inversely correlated by \Cref{thm:new-intro-local-vs-global}, $r$ is exactly the ST-Rank of the underlying constraint graph! As a result, using these properties and generalizing the analysis of BBKSS results in \Cref{thm:informal-intro-unique-games}: an algorithm for unique games on HD-walks with approximation and runtime guarantees dependent on ST-Rank.

\section*{Organization}
In \Cref{sec:related-work} we discuss related work on HD-walks, Unique Games, and CSP approximation. In \Cref{sec:prelims} we give requisite background and notation. In \Cref{sec:linear-algebra} we prove that approximate eigendecompositions tightly control spectral structure. In \Cref{sec:hdx-spectra}, we show that the combinatorial decomposition of \cite{dikstein2018boolean} is an approximate eigendecomposition for any HD-walk and compute its corresponding approximate spectra, thereby determining the spectral structure of HD-walks. In \Cref{sec:pseudorandomness} we prove \Cref{thm:overview-pr-projection}, showing how pseudorandomness controls projections onto any HD-walks eigendecomposition. In \Cref{sec:expansion}, we use this fact to give a tight characterization of edge expansion in HD-walks. Finally, in \Cref{sec:unique-games} we give a sum-of-squares variant of our results and show they imply efficient algorithms for unique games over HD-walks.

\section{Related Work}\label{sec:related-work}

\paragraph{Higher Order Random Walks:}
The spectral structure of higher order random walks has seen significant study in recent years, starting with the work of Kaufman and Oppenheim \cite{kaufman2020high} who proved bounds on the spectra of $N^1_k$ on one-sided local-spectral expanders. Their result lead not only to the resolution of the Mihail-Vazirani conjecture \cite{anari2019log}, but to a number of further breakthroughs in sampling algorithms via a small but consequential improvement on their bound by Alev and Lau \cite{alev2020improved}. The spectral structure of $N^1_k$ on the stronger two-sided local-spectral expanders was further studied by DDFH \cite{dikstein2018boolean} who introduced the HD-Level-Set Decomposition, and Kaufman and Oppenheim \cite{kaufman2020high} who introduced a distinct approximate eigendecomposition with the benefit of orthogonality (though this came at the cost of additional combinatorial complexity). In recent work, Kaufman and Sharakanski \cite{kaufman2020chernoff} claim that these two decompositions are equivalent on sufficiently strong two-sided $\gamma$-local-spectral expanders, but their proof relies on \cite[Theorem 5.10]{kaufman2020high} which has a non-trivial error. Indeed, it is possible to construct arbitrarily strong two-sided local-spectral expanders for which the HD-Level-Set Decomposition is not orthogonal (see \Cref{App:ortho}), so their result cannot hold.\footnote{It is worth noting that the main results of \cite{kaufman2020high,kaufman2020chernoff} are unaffected by this error, as an approximate version of \cite[Theorem 5.10]{kaufman2020high} remains true and is sufficient for their purposes.} Finally, Alev, Jeronimo, and Tulsiani \cite{alev2019approximating} showed that the HD-Level-Set Decomposition is an approximate eigendecomposition (in a weaker sense than we require) for general HD-walks, a result we strengthen in \Cref{sec:hdx-spectra}. For further information on these prior works and their applications, the interested reader should see \cite{AlevThesis}.

\paragraph{Unique Games:}
The study of unique games has played a central role in hardness-of-approximation since Khot's \cite{khot2002power} introduction of the Unique Games Conjecture. One line of work towards \textit{refuting} the UGC focuses on building efficient algorithms for unique games for restricted classes of constraint graphs based off of spectral or spectrally-related properties; these include works employing spectral expansion \cite{arora2008unique,makarychev2010play}, threshold rank \cite{arora2015subexponential,barak2011rounding,guruswami2011lasserre,kolla2011spectral}, hypercontractivity \cite{barak2012hypercontractivity}, and certified small-set-expansion or characterized non-expansion \cite{bafna2020playing}. Our work continues to expand this direction with polynomial-time algorithms for (affine) unique games over HD-walks and the introduction of ST-Rank. On the other hand, recent work towards \textit{proving} the UGC has focused on characterizing non-expanding sets in structures such as the Grassmann \cite{subhash2018pseudorandom,khot2017independent,dinur2018towards,dinur2018non,barak2018small,khot2018small, subhash2018pseudorandom} and Shortcode \cite{barak2018small,subhash2018pseudorandom} graphs. Our spectral framework based on HD-walks and the HD-Level-Set Decomposition provides a more general method to approach this direction than previous Fourier analytic machinery.

\paragraph{CSPs on HDXs:} Finally, it is worth noting a related, recent vein of work connecting high dimensional expansion, Sum of Squares, and CSP-approximation. In particular, Alev, Jeronimo, and Tulsiani \cite{alev2019approximating} recently showed that for $k>2$, certain natural $k$-CSP's on two-sided local-spectral expanders can be efficiently approximated by Sum of Squares. Conversely, Dinur, Filmus, Harsha, and Tulsiani \cite{dinur2020explicit} later used cosystolic expanders (a stronger variant) to build explicit instances of 3-XOR that are hard for SoS. Both these works focus on k-CSPs for $k \geq 3$, where the variables of the CSP are level 1 of the complex and the constraints are defined using the higher levels of the complex; therefore these works do not encompass unique games. At a high level, the techniques of \cite{alev2019approximating} are based on~\cite{barak2011rounding}, and they generalize the BRS algorithm for 2-CSPs to k-CSPs. While these works are not directly related to ours since their definition do not encompass unique games, we see a similar pattern where high dimensional expanding structure is useful both for hardness of and algorithms for CSP-approximation.

\section{Preliminaries and Notation}\label{sec:prelims}

\subsection{Local-Spectral Expanders and Higher Order Random Walks}
We now overview the theory of local-spectral expanders and higher order random walks in more formality than our brief treatment in \Cref{sec:results}.
\subsubsection{Two-Sided Local-Spectral Expanders}
Two-sided local-spectral expanders are a generalization of spectral expander graphs to weighted, uniform hypergraphs, which we will think of as simplicial complexes. 

\begin{definition}[Weighted, Pure Simplicial Complex]
A $d$-dimensional, pure simplicial complex $X$ on $n$ vertices is a subset of ${[n] \choose d}$. We will think of $X$ as the downward closure of these sets, and in particular define the level $X(i)$ as:
\[
X(i) = \left \{ s \in {[n] \choose i}~\bigg\rvert~\exists t \in X, s \subseteq t \right \}.
\]
We call the elements of $X(i)$ $i$-faces.\footnote{We differ here from much of the HDX literature where an $i$-face is often defined to have $i+1$ elements. Since our work is mostly combinatorial rather than topological or geometric, defining an $i$-face to have $i$ elements ends up being the more natural choice.} A simplicial complex is \textit{weighted} if its top level faces are endowed with a distribution $\Pi$. This induces a distribution over each $X(i)$ by downward closure:
\begin{equation}\label{eq:distr-down-clos}
\Pi_i(x) = \frac{1}{i+1}\sum\limits_{y \in X(i+1):y \supset x}  \Pi_{i+1}(y),
\end{equation}
where $\Pi_d = \Pi$.
\end{definition}
Two-sided local-spectral expanders are based upon a phenomenon called \textit{local-to-global} structure, which looks to propogate information on local neighborhoods of a simplicial complex called \textit{links} to the entire complex.
\begin{definition}[Link]
Given a weighted, pure simplicial complex $(X,\Pi)$, the \textbf{link} of an $i$-face $s \in X(i)$ is the sub-complex containing $s$, i.e.
\[
X_s = \{t \setminus s \in X~|~ t \supseteq s\}.
\]
$\Pi$ induces a distribution over $X_s$ by normalizing over top-level faces which we denote by $\Pi_s$. When considering a function on the $k$-th level of a complex, we also use $X_s$ to denote the $k$-faces which contain $s$ as long as it is clear from context, and refer to $X_s$ as an $i$-link if $s \in X(i)$.
\end{definition}
Two-sided local-spectral expansion simply posits that the graph underlying every link\footnote{The underlying graph of a simplicial complex $X$ is its $1$-skeleton $(X(1),X(2))$.} must be a two-sided spectral expander.
\begin{definition}[Local-spectral expansion]
A weighted, pure simplicial complex $(X,\Pi)$ is a two-sided $\gamma$-local-spectral expander if for every $i \le d-2$ 
and every face $s \in X(i)$, the underlying graph of $X_s$
is a two-sided $\gamma$-spectral expander.\footnote{A weighted graph $G(V,E)$ with edge weights $\Pi_E$ is a two-sided $\gamma$-spectral expander if the vertex-edge-vertex random walk with transition probabilities proportional to $\Pi_E$ has second largest eigenvalue in absolute value at most $\gamma$.}
\end{definition}

\subsubsection{Higher Order Random Walks}
Weighted simplicial complexes admit a natural generalization of the standard vertex-edge-vertex walk on graphs known as \textit{higher order random walks} (HD-walks). The basic idea is simple: starting at some $k$-set $S \subset X(k)$, pick at random a set $T \in X(k+1)$ such that $T \supset S$, and then return to $X(k)$ by selecting some $S' \subset T$. Let the space of functions $f: X(k) \to \mathbb{R}$ be denoted by $C_k$. Formally, higher order random walks are a composition of two averaging operators: the ``Up'' operator which lifts a function $f \in C_k$ to $U_kf \in C_{k+1}$:
\[
\forall y \in X(k+1): U_kf(y) = \frac{1}{k+1}\sum\limits_{x \in X(k): x \subset y}f(x),
\]
and the ``Down'' operator which lowers a function $f \in C_{k+1}$ to $D_{k+1} f \in C_k$:
\[
\forall x \in X(k): D_{k+1}f(x) = \frac{1}{k+1}\sum\limits_{y \in X(k+1): y \supset x} \frac{\Pi_{k+1}(y)}{\Pi_{k}(x)}f(y).
\]
These operators exist for each level of the complex, and composing them gives a basic set of higher order random walks we call \textit{pure} (following \cite{alev2019approximating}).

\begin{definition}[$k$-Dimensional Pure Walk]
Given a weighted, simplicial complex $(X,\Pi)$, a $k$-dimensional pure walk $Y: C_k \to C_k$ on $(X,\Pi)$ (of height $h(Y)$) is a composition:
\[
Y = Z_{2h(Y)} \circ \cdots \circ Z_{1},
\]
where each $Z_i$ is a copy of $D$ or $U$.
\end{definition}
We call an affine combination\footnote{An affine combination is a linear combination whose coefficients sum to $1$.} of pure walks which start and end on $X(k)$ a \textit{$k$-dimensional HD-walk}.
\begin{definition}[HD-walk]
Let $(X,\Pi)$ be a pure, weighted simplicial complex. Let $\mathcal Y$ be a family of pure walks $Y: C_k \to C_k$ on $(X,\Pi)$. We call an affine combination 
\[
M = \sum\limits_{Y \in \mathcal Y} \alpha_Y Y
\]
a $k$-dimensional HD-walk on $(X,\Pi)$ as long as it is self-adjoint and remains a valid walk (i.e.\ has non-negative transition probabilities).
\end{definition}
Previous work on HD-walks mainly focuses on two natural classes: canonical walks (introduced in \cite{kaufman2016high,dinur2017high}), and partial-swap walks (introduced in \cite{alev2019approximating,dikstein2019agreement}).
\begin{definition}[Canonical Walk]
Given a $d$-dimensional weighted, pure simplicial complex $(X,\Pi)$, and parameters $k+j \leq d$, the canonical walk $N^j_k$ is:
\[
N_k^j = D^{k+j}_{k}U^{k+j}_{k},
\]
where $U^k_\ell = U_{k-1}\ldots U_{\ell}$, and $D^k_\ell = D_{\ell+1}\ldots D_{k}$.
\end{definition}
In other words, the canonical walk $N^j_k$ takes $j$ steps up and down the complex via the averaging operators. Partial-swap walks are a similar process, but after ascending the complex, we restrict to returning to faces with a given intersection from the starting point.
\begin{definition}[Partial-Swap walk]\label{def:intro-HD-walk}
The partial-swap walk $S^j_k$ is the restriction of $N_k^j$ to faces with intersection $k-j$. In other words, if $|s \cap s'| \neq k-j, S_k^j(s,s')=0$, and otherwise $S_k^j(s,s')=\alpha_sN_k^j(s,s')$, where
\[
\alpha_s = \left (\sum\limits_{s':|s \cap s'|=k-j}N_k(s,s') \right )^{-1}
\]
is the appropriate normalization factor.
\end{definition}
It is not hard to see that partial-swap walk $S^t_k$ on the complete complex $J(n,d)$ (all $d$-subsets of $[n]$ endowed with the uniform distribution) is exactly the Johnson graph $J(n,k,k-t)$. While it is not immediately obvious that the partial-swap walks are HD-walks, Alev, Jeronimo, and Tulsiani \cite{alev2019approximating} showed this is the case by expressing them as an alternating hypergeometric sum of canonical walks.
\subsubsection{Expansion of HD-Walks}
In this work, we study the combinatorial edge expansion of HD-walks, a fundamental property of graphs with strong connections to many areas of theoretical computer science, including both hardness and algorithms for unique games. Given a weighted graph $G=((V,E),(\Pi_V,\Pi_E))$ where $\Pi_V$ is a distribution over vertices, and $\Pi_E$ is a set of non-negative edge weights, the expansion of a subset $S \subset V$ is the average edge-weight leaving $S$.
\begin{definition}[Weighted Edge Expansion]
Given a weighted, directed graph $G=((V,E),(\Pi_V,\Pi_E))$, the weighted edge expansion of a subset $S \subset V$ is:
\[
\phi(G,S) = \underset{v \sim \Pi_V|_S}{\mathbb{E}}\left[ E(v,V \setminus S)\right],
\]
where 
\[
E(v, V \setminus S) = \sum\limits_{(v,y) \in E: y \in V \setminus S} \Pi_E((v,y))
\]
is the total weight of edges between vertex $v$ and the subset $V \setminus S$, and $\Pi_V|_S$ is the re-normalized restriction of $\Pi_V$ to $S$. In the context of a $k$-dimensional HD-Walk $M$ on a weighted simplicial complex $(X,\Pi)$, we will always have $V=X(k)$, $\Pi_V=\Pi_k$, and $E, \Pi_E$ given by $M$. Thus when clear from context, we will simply write $\phi(S)$. 
\end{definition}
Edge expansion in a weighted graph is closely related to the spectral structure of its adjacency matrix. Given a set $S \subset V$ of density $\alpha = \mathbb{E}[\mathbbm{1}_S]$, we may write 
\[
\phi(G,S) = 1 - \frac{1}{\alpha}\langle \mathbbm{1}_S, A_G\mathbbm{1}_S \rangle_{\Pi_V},
\]
where $A_G$ is the adjacency matrix with weights given by $\Pi_E$, and $\langle f, g \rangle_{\Pi_V}$ is the expectation of $fg$ over $\Pi_V$. When considering such an inner product over a weighted simplicial complex $(X,\Pi)$, the associated distribution will always be $\Pi_k$, so we will drop it from the corresponding notation. Notice that the right-hand side of this equivalence may be further broken down via a spectral decomposition of $\mathbbm{1}_S$ with respect to $A_G$. Thus to understand the edge-expansion of HD-walks, it is crucial to understand the structure of their spectra.

\section{Approximate Eigendecompositions and Eigenstripping}\label{sec:linear-algebra}
In this section we prove a general linear algebraic result concerning the spectra of operators that admit an approximate eigendecomposition: their spectra lies tightly concentrated around the decomposition's approximate eigenvalues. Before giving the formal result, we formalize the concept of approximate eigendecompositions.
\begin{definition}
Let $M$ be an operator over an inner product space $V$. We call $V=V^1 \oplus \ldots \oplus V^k$ a $(\{\lambda_i\}_{i=1}^k,\{c_i\}_{i=1}^k)$-approximate eigendecomposition if for all $i$ and $v_i \in V^i$, the following holds:
\[
\norm{Mv_i - \lambda_i v_i} \leq c_i \norm{v_i}.
\]
\end{definition}
As long as the $c_i$ are sufficiently small, we prove each $V^i$ (loosely) corresponds to an eigenstrip, the span of eigenvectors with eigenvalue closely concentrated around $\lambda_i$.
\begin{theorem}[Eigenstripping]\label{thm:approx-ortho}
Let $M$ be a self-adjoint operator over an inner product space $V$, and $V=V^1 \oplus \ldots \oplus V^k$ a $(\{\lambda_i\}_{i=1}^k,\{c_i\}_{i=1}^k)$-approximate eigendecomposition. Let $c_{\max} = \max_i\{ c_i\}$, $\lambda_{\text{dif}}$ = $\min_{i,j}\{|\lambda_i - \lambda_j|\}$, and $\lambda_{\text{ratio}} = \frac{\max_i\{|\lambda_i|\}}{\lambda_{\text{dif}}^{1/2}}$. Then as long as $c_{\max}$ is sufficiently small:
\[
c_{\max} \leq \frac{\lambda_{\text{dif}}}{4k},
\]
the spectra of $M$ is concentrated around each $\lambda_i$:
\[
\text{Spec}(M) \subseteq \bigcup_{i=1}^k \left [ \lambda_i - e, \lambda_i + e\right ] =  I_{\lambda_i},
\]
where $e = O\left(k \cdot \lambda_{\text{ratio}} \cdot c_{\text{max}}^{1/2} \right)$.
\end{theorem}
This result was recently improved by Zhang \cite{Zhang2020} to have dependence $e \leq O(\sqrt{k}c_{\text{max}})$, removing the dependence on the approximate eigenvalues altogether. This follows from a simple modification to our proof which we will note below. Finally, it is also worth mentioning that a version of \Cref{thm:approx-ortho} holds with no assumption on $c_{\text{max}}$, but the assumption substantially simplifies the bounds and is sufficient for our purposes. 

Before proving \Cref{thm:approx-ortho}, we note a useful property of approximate eigendecompositions of self-adjoint operators: they are approximately orthogonal.
\begin{lemma}\label{lemma:approx-orthog}
Let $M$ be a self-adjoint operator over an inner product space $V$. Further, let $V=V^1 \oplus \ldots \oplus V^k$ be a $(\{\lambda_i\}_{i=1}^k,\{c_i\}_{i=1}^k)$-approximate eigen-decomposition. Then for $i\neq j$, $V^i$ and $V^j$ are nearly orthogonal. That is, for any $v_i \in V^i$ and $v_j \in V^j$:
\[
|\langle v_i, v_j \rangle| \leq \frac{c_i+c_j}{|\lambda_i - \lambda_j|}\norm{v_i}\norm{v_j}.
\]
\end{lemma}
\begin{proof}
This follows from the fact that $M$ is self-adjoint, and $V^i$ and $V^j$ are approximate eigenspaces. In particular, notice that for any $v_i \in V^i$ and $v_j \in V^j$ we can bound the interval in which $\langle Mv_i, v_j \rangle=\langle v_i, Mv_j \rangle$ lies by Cauchy-Schwarz:
\begin{align*}
    \langle Mv_i, v_j \rangle \in \lambda_i \langle v_i, v_j \rangle \pm c_i\norm{v_i}\norm{v_j}
\end{align*}
and
\begin{align*}
    \langle v_i, Mv_j \rangle \in \lambda_j \langle v_i, v_j \rangle \pm c_j\norm{v_i}\norm{v_j}.
\end{align*}
Since these terms are equal, the right-hand intervals must overlap. As a result we get:
\begin{align*}
    |(\lambda_i-\lambda_j) \langle v_i, v_j \rangle | \leq (c_i+c_j)\norm{v_i}\norm{v_j},
\end{align*}
as desired.
\end{proof}
Using \Cref{lemma:approx-orthog}, we can modify \cite[Theorem 5.9]{kaufman2020high} to prove \Cref{thm:approx-ortho}. Given an eigenvalue $\mu$ of $M$, the idea is to find a probability distribution over $[k]$ for which the expectation of $|\mu-\lambda_i|$ is small, where $i \in [k]$ is sampled from the aforementioned distribution.
\begin{proof}
The proof follows mostly along the lines of \cite[Theorem 5.9]{kaufman2020high}, modifying where necessary due to lack of orthogonality. Let $\phi$ be an eigenvector of $M$ with eigenvalue $\mu$. Our goal is to prove the existence of some $\lambda_i$ such that $|\mu-\lambda_i|$ is small. To do this, we appeal to an averaging argument. In particular, denoting the component of $\phi$ in $V^i$ by $\phi_i$, we bound the expectation of $|\mu - \lambda_i|^2$ over a distribution $P_\phi$ given by the (normalized) squared norms $\norm{\phi_i}^2$:
\begin{equation}\label{eq:m-l-expectation}
\underset{i \sim P_\phi}{\mathbb{E}}\left[|\mu - \lambda_i|^2\right] = \frac{1}{\sum\limits_{j=1}^k \norm{\phi_j}^2}\sum\limits_{i=1}^k |\mu - \lambda_i|^2 \norm{\phi_i}^2.
\end{equation}
If we can upper bound this expectation by some value $c$, then by averaging there must exist $\lambda_i$ such that $|\mu-\lambda_i| \leq \sqrt{c}$, and thus the spectra of $M$ must lie in strips $\lambda_i \pm \sqrt{c}$. To upper bound \Cref{eq:m-l-expectation}, consider the result of pushing the outer summation inside the norm:
\begin{equation}\label{eq:eigen-analysis}
\sum\limits_{i=1}^k \left | \mu - \lambda_i \right |^2 \norm{\phi_i}^2 = \norm{\sum\limits_{i=1}^k \left ( \mu - \lambda_i \right )\phi_i}^2 - \sum\limits_{1 \leq i \neq j \leq k} (\mu - \lambda_i)(\mu - \lambda_j) \left \langle\phi_i, \phi_j \right \rangle.
\end{equation}
We will separately bound the two resulting terms, the former by the fact that the $\phi_i$ are approximate eigenvectors, and the latter by their approximate orthogonality. We start with the former, which follows by a simple application of Cauchy-Schwarz:
\begin{align*}
\norm{\sum\limits_{i=1}^k \left ( \mu - \lambda_i \right )\phi_i}^2 &= \norm{\mu \phi  - \sum\limits_{i=1}^k \lambda_i \phi_i}^2\\
&= \norm{M\phi  - \sum\limits_{i=1}^k \lambda_i \phi_i}^2\\
&= \norm{ \sum\limits_{i=1}^k \left (M \phi_i - \lambda_i \phi_i\right )}^2\\
&\leq k\sum\limits_{i=1}^k \norm{\left (M \phi_i - \lambda_i \phi_i\right )}^2\\
&\leq kc^2_{\text{max}}\sum\limits_{i=1}^k \norm{\phi_i}^2.
\end{align*}
The latter takes a bit more effort. Let $\lambda_{\text{max}}$ be $\max_i\{|\lambda_i|\}$, then by \Cref{lemma:approx-orthog} we have: 
\begin{align*}
 \left |\sum\limits_{1 \leq i \neq j \leq k} (\mu - \lambda_i)(\mu - \lambda_j) \left \langle \phi_i, \phi_j \right \rangle\right | &\leq \sum\limits_{1 \leq i \neq j \leq k} |\mu - \lambda_i||\mu - \lambda_j| \frac{c_i+c_j}{|\lambda_i-\lambda_j|}\norm{\phi_i}\norm{\phi_j}\\
&\leq  2c_{\text{max}}\lambda_{\text{dif}}^{-1}(\lambda_{\text{max}}+\norm{M})^2\left (\sum\limits_{i=1}^k\norm{\phi_i}\right )^2\\
&\leq  2kc_{\text{max}}\lambda_{\text{dif}}^{-1}(\lambda_{\text{max}}+\norm{M})^2\sum\limits_{i=1}^k\norm{\phi_i}^2\\
\end{align*}
Since we'd like our bound to depend only on $\lambda_i$ and $c_i$, we must further bound $\norm{M}$ which will follow similarly from approximate orthogonality. Let $v$ be a unit eigenvector with eigenvalue $\norm{M}$ and $v_i$ be $v$'s component on $V^i$, then we have:
\begin{align*}
    \norm{M} &= \norm{Mv}\\
    & = \norm{\sum\limits_{i=1}^k Mv_i - \lambda_i v_i + \lambda_i v_i}\\
    &\leq \sum\limits_{i=1}^k (\lambda_i + c_i)\norm{v_i}\\
    &\leq \left(\lambda_{\text{max}} + c_{\text{max}}\right)\sum\limits_{i=1}^k \norm{v_i}\\
    &\leq \left(\lambda_{\text{max}} + c_{\text{max}}\right)\sqrt{k\sum\limits_{i=1}^k \norm{v_i}^2}\\
    &\leq \left(\lambda_{\text{max}} + c_{\text{max}}\right)\sqrt{2k}.
\end{align*}
where the last step follows from \Cref{lemma:approx-orthog} and our assumption on $c_{\text{max}}$:
\begin{align*}
    \sum\limits_{i=1}^k \norm{v_i}^2 &= \norm{v}^2 + \sum\limits_{1 \leq i \neq j \leq k} \langle v_i,v_j \rangle\\
    &\leq 1 +\frac{2c_{\text{max}}}{\lambda_{\text{dif}}}\sum\limits_{1 \leq i \neq j \leq k} \norm{v_i}\norm{v_j} \\
&\leq 1 +\frac{2c_{\text{max}}}{\lambda_{\text{dif}}} \left (\sum\limits_{i=1}^k \norm{v_i} \right )^2\\
&\leq 1+\frac{2kc_{\text{max}}}{\lambda_{\text{dif}}}\sum\limits_{i=1}^k \norm{v_i}^2\\
&\leq 1+\frac{1}{2}\sum\limits_{i=1}^k \norm{v_i}^2.
\end{align*}

Together, these bounds imply the existence of some
$\lambda_{i'}$ such that:
\[
|\mu - \lambda_{i'}| \leq \sqrt{kc_{\text{max}}\left(c_{\text{max}}+2\lambda_{\text{dif}}^{-1}\left(\lambda_{\text{max}} + \left(\lambda_{\text{max}} + c_{\text{max}}\right)\sqrt{2k}\right)^2\right)},
\]
which implies the desired result when accounting for our assumption on $c_{\text{max}}$.
\end{proof}
Zhang's \cite{Zhang2020} improvement to our proof came from the observation that the analysis of the latter term can be simplified by a recursive strategy. In particular, this term can instead be upper bounded by $\frac{2c_{\text{max}}k}{\lambda_{\text{dif}}}\sum\limits_{i=1}^k|\mu-\lambda_i|^2 \norm{\phi_i}^2$. With the appropriate assumption on $c_{\text{max}}$, plugging this back into \Cref{eq:eigen-analysis} gives the desired result.

In either case, notice that if $c_{\text{max}}$ is sufficiently small, the intervals $I_{\lambda_i}$ are disjoint. As a result, each $V^i$ corresponds to an \textit{eigenstrip} $W^i$:
\[
W^i = \text{Span}\left\{ \phi:~M\phi=\mu\phi, \mu \in I_{\lambda_i} \right\}.
\]
The approximate eigenspaces $V^i$ are closely related to the resulting eigenstrips. Indeed, it is possible to show that most of the weight of a function in $V^i$ must lie on $W^i$, though we will not need this result in what follows. Previous works \cite{kaufman2020high,kaufman2020chernoff} make stronger claims for the specific case of the HD-Level-Set Decomposition, most notably that $V^i$ and $W^i$ are in fact equivalent on sufficiently strong two-sided local-spectral expanders. Unfortunately, these results are based off of \cite[Theorem 5.10]{kaufman2020high}, whose proof has a non-trivial error we discuss further in \Cref{App:ortho}. Indeed, were their proof correct, it would imply (due to the generality of their argument) that $V^i=W^i$ for any approximate eigendecomposition. However, it is easy to see this cannot be the case by considering a diagonal $2 \times 2$ matrix with an approximate eigendecomposition given by a slight rotation of the standard basis vectors in $\mathbb{R}^2$.

\section{The Spectra of HD-walks}\label{sec:hdx-spectra}
We now show that the HD-Level-Set Decomposition is an approximate eigendecomposition for any HD-Walk. Combined with \Cref{thm:approx-ortho}, this proves \Cref{thm:new-intro-hdx-spectra}, that the spectrum of any $k$-dimensional HD walk is tightly concentrated in $k+1$ eigenstrips. As a result, we give explicit bounds on the spectra of HD-walks, paying special attention to the canonical and partial-swap walks. Finally, we show that the approximate eigenvalues (and thus the values in their corresponding eigenstrips) of the HD-Level-Set Decomposition decrease monotonically for a broad class of HD-Walks we call \textit{complete} walks which, to our knowledge, encompass all walks used in the literature. As we will see in the following section, such decay is crucial for understanding edge expansion.

To start, we recall the definition of pure and HD-walks along with introducing some useful notation.
\begin{definition}[$k$-Dimensional Pure Walk]
Given a weighted, simplicial complex $(X,\Pi)$, a $k$-dimensional pure walk $Y: C_k \to C_k$ on $(X,\Pi)$ is a composition:
\[
Y = Z_{2h(Y)} \circ \cdots \circ Z_{1},
\]
where each $Z_i$ is a copy of $D$ or $U$, and $h(Y)$ is the \textit{height} of the walk, measuring the total number of down (or up) operators. 
\end{definition}
\begin{definition}[$k$-Dimensional HD-Walk]
Given a weighted, simplicial complex $(X,\Pi)$, a $k$-dimensional HD-walk on $(X,\Pi)$ is an affine combination of pure walks
\[
M = \sum_{Y \in \mathcal Y} \alpha_Y Y
\]
which is self-adjoint and gives a valid walk on $(X,\Pi)$ (i.e.\ has non-negative transition probabilities). We say the height of $M$, $h(M)$, is the maximal height of any $Y$ with a non-zero coefficient, and say the weight of $M$, $w(M)$, is the one norm of the $\alpha_Y$ (namely, $w(M)=\sum |\alpha_Y|$).
\end{definition}
Our proofs in this section rely mainly on a useful observation of \cite{dikstein2018boolean}, who show that the up and down operators on two-sided $\gamma$-local-spectral expanders satisfy the following relation:
\begin{align}\label{eq:hdx-eposet}
\norm{D_{i+1}U_i - \frac{1}{i+1}I - \frac{i}{i+1}U_{i-1}D_i} \leq \gamma.
\end{align}
This fact leads to a particularly useful structural lemma showing the effect of flipping $D$ through multiple $U$ operators.
\begin{lemma}[Claim 8.8 \cite{dikstein2018boolean}]\label{lemma:body-DU-UD}
Let $(X,\Pi)$ be a $d$-dimensional $\gamma$-local-spectral expander. Then for all $j< k < d$:
\[
\norm{D_{k+1}U^{k+1}_{k-j} - \frac{j+1}{k+1}U^{k}_{k-j} - \frac{k-j}{k+1} U^{k}_{k-j-1}D_{k-j}} \leq \frac{(j+1)(2k-j+2)}{2(k+1)}\gamma.
\]
\end{lemma}
One crucial application of \Cref{lemma:body-DU-UD} lies in understanding the relation between $\norm{f_i}$, and $\norm{g_i}$, where $f_i=U^k_ig_i$. 
\begin{lemma}[Lemmas 8.10, 8.13, Theorem 4.6 \cite{dikstein2018boolean}]\label{lemma:hdx-fvsg-body}
Let $(X,\Pi)$ be a $d$-dimensional $\gamma$-local-spectral expander with $\gamma < 1/d$, $f \in C_k$ a function with HD-Level-Set Decomposition $f_0 + \ldots + f_k$. Then for all $0 \leq \ell \leq k \leq d$:
\[
\norm{f_\ell}^2 = \frac{1}{{k \choose \ell}}\left (1 \pm c_1(k,\ell)\gamma \right )\norm{g_{\ell}}^2,
\]
where $c_1(k,\ell)=O(k^2{k \choose \ell})$. 
\end{lemma}
In \Cref{app:decomp}, we prove a stronger version of both \Cref{lemma:body-DU-UD} and \Cref{lemma:hdx-fvsg-body} for $\gamma \leq 2^{-\Omega(k)}$ where the dependence on the first order term $\gamma$ is polynomial rather than exponential in $k$. However, since this only provides a substantial improvement for a small range of $\gamma$, we use the simpler versions from \cite{dikstein2018boolean} throughout the body of the paper. Using \Cref{lemma:body-DU-UD} and \Cref{lemma:hdx-fvsg-body}, an inductive argument shows that the HD-Level-Set Decomposition is an approximate eigendecomposition. We show this first for the basic case of a pure walk, and then note that the general result follows immediately from the triangle inequality.
\begin{proposition}\label{prop:hdx-pure-eig-vals}
Let $(X,\Pi)$ be a two-sided $\gamma$-local-spectral expander with $\gamma \leq 2^{-\Omega(k)}$ and $Y:C_k \to C_k$ a pure walk:
\[
Y = Z_{2h(Y)} \circ \cdots \circ Z_1.
\]
Let $i_1 \leq \ldots \leq i_{h(Y)}$ denote the $h(Y)$ indices at which $Z_i$ is a down operator. Then for all $0 \leq \ell \leq k, f \in V_k^\ell$:
\[
\norm{Y f-\prod\limits_{s=1}^{h(Y)} \left(1-\frac{\ell}{\max\{\ell,i_s - 2s+k + 1\}}\right)f} \leq O\left (\gamma h(Y)(k+h(Y)){k \choose \ell}\norm{f}\right).
\]
\end{proposition}
\begin{proof}
We prove a slightly stronger statement to simplify the induction. For $b>0$, let $Y_{j}^b: C_\ell \to C_{\ell+b}$ denote an unbalanced walk with $j$ down operators, and $j+b$ up operators. If $Y_{j}^b$ has down operators in positions $i_1 \leq \ldots \leq i_j$ and $g_\ell \in H^\ell$, we claim:
\begin{equation}\label{eq:pure-b}
\norm{Y^b_j g_\ell - \prod\limits_{s=1}^{j}\left ( 1-\frac{\ell}{\max\{\ell,i_s-2s+\ell + 1\}} \right )U^{b+\ell}_\ell g_\ell} \leq \gamma j(b+j)\norm{g_\ell}.
\end{equation}
Notice that since $f \in V_k^\ell$ may be written as $U^k_\ell g_\ell$ for $g_\ell \in H^\ell$, then we may write $Yf$ as $Y^{k-\ell}_{h(Y)} g_\ell$ where $Y^{k-\ell}_{h(Y)}$ has down operators in positions $i_1 + k-\ell \leq \ldots \leq i_j + k - \ell$. Combining \Cref{eq:pure-b} with \Cref{lemma:hdx-fvsg-body} then implies the result.

We prove \Cref{eq:pure-b} by induction. The base case $j=0$ is trivial. Assume the inductive hypothesis holds for all $Y^b_i, i<j$. Notice first that if $i_1=1$, we are done since $g_\ell \in H^\ell$, and
\[
\prod\limits_{s=1}^{j}\left ( 1-\frac{\ell}{\max\{\ell,i_s-2s+\ell+1\}} \right )Y_0^bg_\ell = 0,
\]
as $i_s-2s+\ell+1 = \ell$ for $s=1$. Otherwise, it must be the case that one or more copies of the up operator appear before the first down operator, and we may therefore apply \Cref{lemma:body-DU-UD} to get:
\begin{align*}
    Y^b_j g_\ell = \left(\frac{i_1-1}{i_1+\ell-1}\right)Y^{b}_{j-1}g_\ell+\Gamma g_\ell,
\end{align*}
where we can (loosely) bound the spectral norm of $\Gamma$ by
\[
\norm{\Gamma} \leq (b+j)\gamma
\]
since at worst the first down operator $D$ passes through $b+j$ up operators. By the form of \Cref{lemma:body-DU-UD}, $Y^{b}_{j-1}$ has down operators at indices $i_2-2\leq \ldots \leq i_j-2$. Then by the fact that $i_1 + \ell - 1 > \ell$ and the inductive hypothesis:
\begin{align*}
        Y^b_j g_\ell &= \left(\frac{i_1-1}{\max\{\ell,i_1+\ell-1\}}\prod\limits_{s=1}^{j-1}\frac{i_{s+1}-2s-1}{\max\{\ell,i_{s+1}-2s+\ell-1\}}\right)Y^b_0g_\ell + \frac{i_1-1}{i_1+\ell-1}h + \Gamma g_\ell\\
         &= \left(\frac{i_1-1}{\max\{\ell,i_1+\ell-1\}}\prod\limits_{s=2}^{j}\frac{i_s-2s+1}{\max\{\ell,i_s-2s+\ell+1\}}\right)Y^b_0g_\ell + \frac{i_1-1}{i_1+\ell-1}h + \Gamma g_\ell\\
        &= \left(\prod\limits_{s=1}^{j}\frac{i_s-2s+1}{\max\{\ell,i_s-2s+\ell+1\}}\right)Y_0^b g_\ell + \frac{i_1-1}{i_1+\ell-1}h + \Gamma g_\ell,
\end{align*}
where $\norm{h} \leq \gamma (j-1)(b+j-1)\norm{g_\ell}$ and we have used the (vacuous) fact that $\max\{\ell,i_1+\ell-1\} = i_1+\ell-1$. Finally, we can bound the norm of the right-hand error term by:
\begin{align*}
\norm{\frac{i_1-1}{i_1+\ell-1}h+\Gamma g_\ell}
& \leq \norm{h} + \norm{\Gamma}\norm{g_\ell}\\
&\leq (j-1)(b+j-1)\norm{g_\ell} + (b+j)\norm{g_\ell}\\
&\leq j(b+j)\norm{g_\ell}
\end{align*}
as desired.
\end{proof}
Since HD-walks are simply affine combinations of pure walks, the triangle inequality immediately implies the result carries over to this more general setting.
\begin{corollary}\label{prop:hdx-eig-vals}
Let $(X,\Pi)$ be a two-sided $\gamma$-local-spectral expander with $\gamma \leq 2^{-\Omega(k)}$ and $M = \sum\limits_i \alpha_i Y_i$ a $k$-dimensional HD-walk on $(X,\Pi)$. Then for all $0 \leq \ell \leq k, f \in V_k^\ell$:
\[
\norm{Mf - \lambda_\ell(M) f} \leq O\left (\gamma w(M)h(M)(k+h(M)){k \choose \ell}\norm{f}\right),
\]
where
\[
\lambda_\ell(M) = \sum\limits \alpha_i \lambda_\ell(Y_i)
\]
and $\lambda_\ell(Y_i)$ is the approximate eigenvalue of $Y_i$ given in \Cref{prop:hdx-pure-eig-vals}.
\end{corollary}

It is worth noting that the resulting approximate eigenvalues in \Cref{prop:hdx-eig-vals} are exactly the eigenvalues of $M$ when considered on a sequentially differential poset with $\vec{\delta}_i = i/(i+1)$. We discuss this generalization in more depth and give tighter bounds on the approximate spectra in our upcoming companion paper. It should be noted that this result is similar to one appearing in \cite{alev2019approximating}, where a weaker notion of approximate eigenspaces based on the quadratic form $\langle f,Mf \rangle$ is analyzed. Plugging \Cref{prop:hdx-eig-vals} into \Cref{thm:approx-ortho}, we immediately get that for small enough $\gamma$ the true spectra of HD-walks lie in strips around each $\lambda_i(M)$, and thus that that the approximate eigenvalues of the HD-Level-Set Decomposition and the spectra of HD-walks are essentially interchangeable.

For concreteness, we now turn our attention to computing the approximate eigenvalues (and thereby the true spectra) of the canonical and swap walks.
\begin{corollary}[Spectrum of Canonical Walks]\label{prop:canon-spectra-HDX}
Let $(X,\Pi)$ be a $d$-dimensional $\gamma$-local-spectral expander with $\gamma$ satisfying $\gamma \leq 2^{-\Omega(k+j)}$, $k+j \leq d$, and $f_\ell \in V_k^\ell$. Then:
\[
\norm{N^j_kf_\ell - \frac{{k \choose \ell}}{{k+j \choose \ell}}f_\ell}  \leq c(k,\ell,j)\norm{f_\ell},
\]
where $c(k,\ell,j) = O\left (\gamma j(j+k){k \choose \ell}\right )$.
Moreover:
\[
\text{Spec}(N^j_k) = \{1\} \cup \bigcup_{j=1}^{k} \left [\frac{\binom{k}{\ell}}{\binom{k+j}{\ell}} \pm 2^{O(j+k)}\sqrt{\gamma} \right ].
\]
\end{corollary}
\begin{proof}
By \Cref{prop:hdx-pure-eig-vals}, $N^j_k$ is an $\left(\{\lambda_\ell\}_{\ell=0}^k,\{c(k,\ell,j)\}_{\ell=0}^k\right)$-approximate eigendecomposition for
\[
\lambda_\ell = \prod\limits_{s=1}^{j} \left(1-\frac{\ell}{\max\{k-2s+i_s+1,\ell\}}\right),
\]
where $i_1 \leq \ldots \leq i_s$ denote the indices of down operators. By the definition of $N^j_k$ we have $i_s=j+s$, and therefore
\[
\lambda_\ell = \prod\limits_{s=1}^{j} \left(1-\frac{\ell}{k-s+j+1}\right) = \frac{{k \choose \ell}}{{k + j \choose \ell}}
\]
as desired. The bounds on Spec$(N^j_k)$ follow immediately from plugging the above into \Cref{thm:approx-ortho}.
\end{proof}
A priori, it is not obvious how to bound the spectra of the partial-swap walks, or indeed even that they are HD-walks. However, Alev, Jeronimo, and Tulsiani \cite{alev2019approximating} proved that partial-swap walks may be written as an alternating hypergeometric sum of canonical walks. 
\begin{proposition}[Corollary 4.13 \cite{alev2019approximating}]\label{prop:AJT}
Let $(X,\Pi)$ be a two-sided $\gamma$-local-spectral expander with $\gamma < 1/k$. Then for $0\leq j \leq k$:
\begin{align*}
 S_k^j = \frac{1}{{k \choose k-j}}\sum\limits_{i=0}^{j} (-1)^{j-i}{j \choose i}{k+i \choose i}
 N_k^{i}.
\end{align*}
\end{proposition}
As a result, we can use \Cref{prop:canon-spectra-HDX} to bound their approximate eigenvalues and true spectrum.
\begin{corollary}\label{cor:swap-HDX-spectra}
Let $X$ be $d$-dimensional two-sided $\gamma$-local-spectral expander, $\gamma < 2^{-\Omega(k)}$, $k+j \leq d$, and $f_\ell \in V_k^\ell$. Then:
\[
\norm{S_k^j f_\ell - \frac{{k-j \choose \ell}}{{k \choose \ell}}f_\ell} \leq c(k)\norm{f_\ell},
\]
where $c(k) = \gamma 2^{O(k)}$. Moreover,
\[
\text{Spec}(S^j_k) = \{1\} \cup \bigcup_{j=1}^{k} \left [\frac{\binom{k-j}{\ell}}{\binom{k}{\ell}} \pm 2^{O(k)}\sqrt{\gamma} \right ].
\]
\end{corollary}
\begin{proof}
By \Cref{prop:hdx-eig-vals}, $\bigoplus_{\ell=0}^k V_k^\ell$ is a $\left(\{\lambda_\ell\}_{\ell=0}^k,\{c'(k,\ell,j)\}_{\ell=0}^k\right)$-approximate eigendecomposition for $S_k^j$ with
\begin{align*}
    \lambda_\ell &= \frac{1}{{k \choose k-j}}\sum\limits_{i=0}^{j} (-1)^{j-i}{j \choose i}{k+i \choose i}\lambda_\ell(N_k^{i})\\
    &= \frac{1}{{k \choose k-j}}\sum\limits_{i=0}^{j} (-1)^{j-i}{j \choose i}{k+i \choose i}\frac{{k \choose \ell}}{{k+i \choose \ell}}\\
    & = \frac{1}{{k \choose k-j}}\sum\limits_{i=0}^j(-1)^{j-i}{j \choose i}{k-\ell+i \choose i}\\
    &= \frac{{k - \ell \choose j}}{{k \choose k-j}}\\
    &= \frac{\binom{k-j}{\ell}}{\binom{k}{\ell}},
\end{align*}
and
\begin{align*}
    c'(k,\ell,j) &= \gamma 2^{O(k)}.
\end{align*}
This latter fact follows from noting that
\begin{align*}
\norm{\vec{\alpha}}_1 = \sum\limits_{i=0}^j {j \choose i}{{k+i \choose i}} \leq 2^{2j+k}
\end{align*}
where $\vec{\alpha}$ consists of the hypergeometric coefficients of \Cref{prop:AJT}. The bounds on Spec($S^j_k$) then follow from \Cref{thm:approx-ortho}.
\end{proof}
Together, \Cref{prop:canon-spectra-HDX} and \Cref{cor:swap-HDX-spectra} prove \Cref{thm:intro-hdx-swap-canon} (assuming $\gamma$ is sufficiently small).

In \Cref{thm:new-intro-hdx-spectra}, we mentioned that approximate eigenvalues $\lambda_i(M)$ are monotonically decreasing for any HD-walk. In fact, to prove this we will need to restrict our original definition of HD-walks slightly, requiring that our walks are always well-defined on the complete complex.
\begin{definition}[Complete HD-Walk]\label{def:complete-walk}
Let $(X,\Pi)$ be a weighted, pure simplicial complex and $M=\sum\limits_{Y \in \mathcal Y} \alpha_Y Y$ an HD-walk on $(X,\Pi)$. We call $M$ \textit{complete} if for all $n \in \mathbb{N}$ there exist $n_0 > n$ and $d$ such that $\sum\limits_{Y \in \mathcal Y} \alpha_Y Y$ is also an HD-walk when taken to be over $J(n_0,d)$.
\end{definition}
To our knowledge, all walks considered in the literature (pure, canonical, partial-swap) are complete. We can prove that the eigenstrips of complete HD-walks corresponding to the HD-Level-Set Decomposition exhibit eigenvalue decay by noting that the approximate eigenvalues of \Cref{prop:hdx-eig-vals} are independent of the underlying complex.
\begin{proposition}\label{prop:eig-decrease}
Let $(X,\Pi)$ be a two-sided $\gamma$-local-spectral expander, $M = \sum\limits_{Y \in \mathcal Y} \alpha_Y Y$ a complete HD-walk over $(X,\Pi)$, and $\gamma$ small enough to apply the conditions of \Cref{thm:approx-ortho}. Then for all $0 \leq i < j \leq k$,
\[
\lambda_i(M) \geq \lambda_j(M)
\]
\end{proposition}
\begin{proof}
The proof follows from two observations. First, recall from \Cref{prop:hdx-eig-vals} that $\lambda_i(M)$ is independent of the underlying complex. Second, any HD-walk on the complete complex can be written as a non-negative sum of partial-swap walks, which satisfy the monotonic decrease property. Let $n \in \mathbb{N}$ be any parameter such that applying $\sum\limits_{Y \in \mathcal Y} \alpha_Y Y$ to $J(n,d)$ results in a valid walk (i.e.\ a non-negative matrix). By the symmetry of $J(n,d)$, the transition probabilities of this walk depends only on size of intersection, and it may thus be written as some convex combination of partial-swap walks:
\[
M=\sum\limits_{Y \in \mathcal Y} \alpha_Y Y =\sum\limits_{i=0}^k \beta_i S_k^i.
\]
Since these walks are equivalent over $J(n,d)$, their spectra must match. Then by \Cref{thm:approx-ortho}, it must be the case that for every $1 \leq \ell \leq k$ and $n$ sufficiently large, the intersection of $\sum\limits_{Y \in \mathcal Y} \alpha_Y \lambda_\ell(Y) \pm O(1/n)$ and $\sum\limits_{Y \in \mathcal Y} \beta_i \lambda_\ell(S_k^i) \pm O(1/n)$ is non-empty. Since we may take $n$ arbitrarily large, this implies the two quantities are in fact equivalent. Finally, by \Cref{cor:swap-HDX-spectra} $\lambda_\ell(S^i_k)$ decreases monotonically in $\ell$ for all $i$, which implies that the $\lambda_i(M) = \sum\limits \alpha_Y\lambda_i(Y) = \sum\limits \beta \lambda_i(S^j_k)$ decrease monotonically as desired.
\end{proof}

\section{Pseudorandomness and the HD-Level-Set Decomposition}\label{sec:pseudorandomness}
Now that we have examined the spectral structure of the HD-Level-Set Decomposition, we turn to understanding its combinatorial characteristics. In this section, we give a combinatorial characterization how arbitrary functions project onto the HD-Level-Set decomposition, proving in particular a generalization of \Cref{thm:overview-pr-projection}: pseudorandom sets have bounded projection onto corresponding levels of the complex. We discuss two variants of pseudorandomness, an $\ell_2$-variant stating that the variance across links is small, and an $\ell_\infty$-variant stating that the maximum (or $\ell_\infty$-norm for arbitrary functions) across links is small. We focus mainly on giving an exact analysis for the $\ell_2$-case, as this forms the structural core of our algorithm for unique games in \Cref{sec:unique-games}. We further discuss the impliations of our $\ell_2$ analysis to the $\ell_\infty$-variant by reduction.

As such, we'll start by analyzing the $\ell_2$-variant. First, let's extend our definition of $\ell_2$-pseudorandomness from sets (boolean functions) to arbitrary functions. 
\begin{definition}[$\ell_2$-Pseudorandom functions]\label{def:pseudorandom}
A function $f \in C_k$ is $(\varepsilon_1,\ldots,\varepsilon_\ell)$-$\ell_2$-pseudorandom if its variance across $i$-links is small for all $1 \leq i \leq \ell$:
\[
Var(D^k_if) \leq \varepsilon_i |\mathbb{E}[f]|
\]
\end{definition}
As mentioned in \Cref{sec:proof-overview}, the key to connecting $\text{Var}(D^k_if)$ and the HD-Level-Set Decomposition is to notice that by the adjointness of $D$ and $U$, we can reduce the problem to analyzing the spectral structure of HD-walks. The proof then follows immediately from arguments in the previous section.
\begin{theorem}\label{lem:low-level-weight}
Let $(X,\Pi)$ be a $\gamma$-local-spectral expander with $\gamma \leq 2^{-\Omega(k)}$ and let $f \in C_k$ have HD-Level-Set Decomposition $f=f_0+\ldots+f_k$. Then for any $\ell \leq k$, the $\text{Var}(D^k_\ell f)$ is controlled by its projection onto $V_k^0 \oplus \ldots \oplus V_k^\ell$ in the following sense: 
\begin{align*}
     \text{Var}(D^k_\ell f)
     &\in \sum\limits_{j=1}^\ell \frac{{\ell \choose j}}{{k \choose j}}\langle f, f_j \rangle \pm c\gamma \langle f, f \rangle
\end{align*}
where $c \leq 2^{O(k)}$.
\end{theorem}
\begin{proof}
To start, notice that since $\langle D_k^\ell f, D_k^\ell f \rangle = \langle U^k_\ell D^k_\ell f, f \rangle$ it is enough to analyze the application of the HD-walk $U^k_\ell D^k_\ell$ to $f$. By \Cref{prop:hdx-pure-eig-vals}, we know that each $f_j$ is an approximate eigenvector satisfying:
\[
\norm{U^k_\ell D^k_\ell f_j - \frac{{\ell \choose j}}{{k \choose j}}f_j} \leq c_1\gamma\norm{f},
\]
where $c_1 \leq 2^{O(k)}$. Combining these observations gives:
\begin{align*}
     \left \langle D^k_\ell f,D^k_\ell f \right \rangle &= \left \langle  f,U^k_\ell D^k_\ell f \right \rangle\\
     &= \sum\limits_{j=0}^k \langle f, U^k_\ell D^k_\ell f_j \rangle\\
     &\in \sum\limits_{j=0}^\ell \frac{{\ell \choose j}}{{k \choose j}}\langle f, f_j \rangle \pm c\gamma \langle f, f \rangle
\end{align*}
where all constants $c \leq 2^{O(k)}$, and noting that $\langle f, f_0 \rangle = \mathbb{E}[f]^2$ completes the result.
\end{proof}
\Cref{thm:overview-pr-projection} follows immediately from combining this result with approximate orthogonality of the HD-Level-Set Decomposition.
\begin{lemma}[DDFH Theorem 4.6]\label{lemma:HDX-approx-ortho}
Let $(X,\Pi)$ be a $\gamma$-local-spectral expander with $\gamma \leq 2^{-\Omega(k)}$ and let $f \in C_k$ have HD-Level-Set Decomposition $f=f_0 + \ldots + f_k$. Then for all $i\neq j$:
\[
|\langle f_i,f_j \rangle | \leq c_1\gamma \langle f,f \rangle
\]
where $c_1 \leq 2^{O(k)}$.
\end{lemma}
\begin{corollary}\label{cor:ell2-proj}
Let $(X,\Pi)$ be a $\gamma$-local-spectral expander with $\gamma \leq 2^{-\Omega(k)}$ and let $f \in C_k$ be an $(\varepsilon_1,\ldots,\varepsilon_\ell)$-$\ell_2$-pseudorandom function. Then for any $1 \leq i \leq \ell$:
\[
|\langle f, f_i \rangle| \leq {k \choose i}\varepsilon_i|\mathbb{E}[f]| + c\gamma\langle f,f \rangle
\]
\end{corollary}
where $c \leq 2^{O(k)}$.
\begin{proof}
By \Cref{lemma:HDX-approx-ortho}, for all $j$ we have $\langle f, f_j \rangle \geq -c\gamma\langle f,f\rangle$ for some $c \leq 2^{O(k)}$. Then by \Cref{lem:low-level-weight}, for all $0 \leq i \leq k$ the variance of $D_i^k f$ is lower bounded by the projection onto $f_i$:
\begin{align*}
    \text{Var}(D_i^k f) \geq \frac{1}{{k \choose i}}\langle f,f_i \rangle - c_2\gamma\langle f,f \rangle,
\end{align*}
where $c_2 \leq 2^{O(k)}$. Finally, since $f$ is $(\varepsilon_1,\ldots,\varepsilon_\ell)$-$\ell_2$-pseudorandom, we have by definition that for all $1 \leq i \leq \ell$, $\text{Var}(D_i^k f) \leq \varepsilon_i|\mathbb{E}[f]|$. Therefore isolating $\langle f,f_i \rangle$ gives the desired upper bound:
\begin{align*}
    \langle f,f_i \rangle &\leq {k \choose i}\text{Var}(D_i^k f) + c_3\gamma\langle f,f \rangle\\
    &\leq {k \choose i}\varepsilon_i|\mathbb{E}[f]| + c_3\gamma\langle f,f \rangle
\end{align*}
where $c_3 \leq 2^{O(k)}$. Finally, we already noted that $\langle f,f_i\rangle \geq -c\gamma\langle f,f\rangle$ for some $c \leq 2^{O(k)}$, which gives the desired lower bound and completes the proof.
\end{proof}
It is worth noting that both \Cref{lem:low-level-weight} and \Cref{cor:ell2-proj} are tight. This is obvious for the former which is a near-equality, and the latter is clearly tight for subsets like $i$-links which project (almost) entirely onto level $i$ (we'll prove this formally in the next section).

We'll also see in \Cref{sec:expansion} that \Cref{cor:ell2-proj} leads to a tight $\ell_2$-characterization of edge-expansion stating that any non-expanding set must have high variance across links. Since an $\ell_\infty$-variant of this result for the Grassmann graphs was recently crucial for the resolution of the 2-2 Games Conjecture \cite{subhash2018pseudorandom}, it is natural to discuss what our results imply for this regime. First, let's formalize $\ell_\infty$-pseudorandomness for arbitrary functions.
\begin{definition}[$\ell_\infty$-Pseudorandom functions]\label{def:pseudorandom-infty}
A function $f \in C_k$ is $(\varepsilon_1,\ldots,\varepsilon_\ell)$-$\ell_\infty$-pseudorandom if for all $1 \leq i \leq \ell$ its local expectation is close to its global expectation:
\[
\left \| D^k_if - \mathbb{E}[f] \right \|_\infty \leq \varepsilon_i.
\]
\end{definition}
We will prove via reduction to the $\ell_2$-variant that a version of \Cref{cor:ell2-proj} holds in this regime as well. We proceed in two steps. First, we show that $\ell_\infty$-pseudorandom functions are also $\ell_2$-pseudorandom assuming a weak local-consistency property.

\begin{definition}
Let $(X,\Pi)$ be a weighted, pure simplicial complex. We say a function $f \in C_k$ has $\ell$-local constant sign if:
\begin{enumerate}
    \item $\mathbb{E}[f] \neq 0$,
    \item $\forall s \in X(\ell)$ s.t.\ $~\underset{X_s}{\mathbb{E}}[f] \neq 0: \text{sign}\left (\underset{X_s}{\mathbb{E}}[f]\right ) = \text{sign} \left (\mathbb{E}[f] \right )$.
\end{enumerate}
\end{definition}
Second, we'll reduce to the case of locally constant sign by noting that we can always shift a function to satisfy this property. With these definitions in hand, we can now state the $\ell_\infty$-variant of \Cref{thm:overview-pr-projection}:
\begin{theorem}\label{thm:body-local-spec-proj}
Let $(X,\Pi)$ be a $\gamma$-local-spectral expander with $\gamma \leq 2^{-\Omega(k)}$ and let $f \in C_k$ have HD-Level-Set Decomposition $f=f_0+\ldots+f_k$. If $f$ is $(\varepsilon_1,\ldots,\varepsilon_\ell)$-$\ell_\infty$-pseudorandom, then for all $1 \leq i \leq \ell$:
\[
|\langle f, f_i \rangle| \leq  \left({k \choose i} + c(k)\gamma\right)\varepsilon_i^2 + c(k)\gamma\norm{f}^2,
\]
where $c(k) \leq 2^{O(k)}$, and if $f$ has $i$-local constant sign:
\[
|\langle f, f_i \rangle| \leq  {k \choose i}\varepsilon_i|\mathbb{E}[f]| + c(k)\gamma\norm{f}^2.
\]
\end{theorem}
In dealing with expansion, we will mainly be interested in boolean-valued functions, which always have locally-constant sign and satisfy $\langle f, f \rangle = \mathbb{E}[f]$. Thus in the boolean case we have:
\[
\langle f, f_i \rangle \leq \left({k \choose i}\varepsilon_i+2^{O(k)}\gamma\right) \mathbb{E}[f],
\]
which is particularly useful since the expansion of $f$ may be written as $1 - \frac{1}{\mathbb{E}[f]}\langle f, Mf \rangle$. In the case $f$ is non-negative, we can also replace the $\ell_\infty$ norm with maximum in our definition of pseudorandomness, which is important to show non-expanding sets are locally \textit{denser} than expected. Finally, we note that one can improve the dependence on $\gamma$ by pushing exponential dependence on $k$ to the second order $\gamma^2$ term via more careful analysis of error propagation. However since the analysis is complicated and only gives a substantial improvement for a small range of relevant $\gamma$, we relegate such discussion to \Cref{app:decomp}.

We now move to the proof of \Cref{thm:body-local-spec-proj}, starting with our generic $\ell_\infty$ to $\ell_2$ reduction for functions with locally-constant sign.
\begin{lemma}\label{lem:reduction}
Let $(X,\Pi)$ be a weighted, pure simplicial complex, and $f \in C_k$ a $(\varepsilon_1,\ldots,\varepsilon_\ell)$-$\ell_\infty$-pseudorandom function with $i$-local constant sign for any $i \leq \ell$. Then f is also $(\varepsilon_1,\ldots,\varepsilon_\ell)$-$\ell_2$-pseudorandom
\end{lemma}
\begin{proof}
For ease of notation, let $\Pi_k(X_s)$ be shorthand for $\sum\limits_{t\in X_s}\Pi_k(t)$ (i.e. the normalization factor for the above restricted expectation). The trick is to notice that since $f$ has locally constant sign, we may rewrite $\norm{D^k_i f}_2^2$ as an expectation over a related distribution $P_{i}$:
\begin{align*}
    \frac{1}{\mathbb{E}[f]}\langle D^k_i f, D^k_i f \rangle  &= \sum\limits_{s \in X(i)} \Pi_i(s)\left (\frac{1}{\mathbb{E}[f]}\sum\limits_{t \in X_s}\frac{\Pi_k(t)f(t)}{\Pi_k(X_s)}\right )D^k_i f(s)\\
      &= \sum\limits_{s \in X(i)} \left (\frac{1}{\mathbb{E}[f]}\sum\limits_{t \in X_s}\frac{\Pi_k(t)f(t)}{{k \choose i}}\right )D^k_i f(s)\\
    &= \underset{P_{i}}{\mathbb{E}}[D^k_i f], \label{eq:exp-to-norm}
\end{align*}
where we have used the fact that $\Pi_k(X_s) = {k \choose i}\Pi_i(s)$ by \Cref{eq:distr-down-clos}. To understand $P_{i}(s)$ more intuitively, consider the special case when $f$ is non-negative. Here, $\Pi$ and $f$ induce a distribution $P_k$ over $X(k)$, where
\[
P_k(t) = \frac{\Pi_k(t)f(t)}{\mathbb{E}[f]}.
\]
$P_k$ then induces the distribution $P_i$ on $X(i)$ via the following process: draw a face $t \in X(k)$ from $P_k$, and then choose a $i$-face $s \subset t$ uniformly at random. Replacing the non-negativity of $f$ with the conditions in the theorem statement still leaves $P_{i}(s)$ a valid distribution, albeit one with a less intuitive description.

The result then follows from an averaging argument:
\[
\left |\frac{1}{\mathbb{E}[f]}\text{Var}(D^k_if)\right |
= \left |\underset{P_i}{\mathbb{E}}[D^k_i f] - \mathbb{E}[f] \right | \leq \|D^k_if - \mathbb{E}[f]\|_{\infty}
\]
We note that when $\mathbb{E}[f]>0$, the $\ell_\infty$-norm may be replaced with maximum in the above.
\end{proof}
To complete the proof of \Cref{thm:body-local-spec-proj}, we reduce to \Cref{lem:low-level-weight} by noting that any function may be shifted to have locally constant sign and applying our reduction.
\begin{proof}[Proof of \Cref{thm:body-local-spec-proj}]
Note that the latter bound for functions with locally constant sign is immediate since, by \Cref{lem:reduction}, $f$ is also $(\varepsilon_1,\ldots,\varepsilon_\ell)$-$\ell_2$-pseudorandom and therefore satisfies the guarantees of \Cref{cor:ell2-proj}.

For the former, assume for simplicity that $\mathbb{E}[f]\geq 0$ (the negative case follows from a similar argument) and consider the shifted function $f'=f + (\varepsilon_i - \mathbb{E}[f])\mathbbm{1}$. Notice that as long as $\varepsilon_i>0$, $f'$ has positive expectation over all $i$-links and non-zero expectation, and further that
\[
f' = f_0' + f_1 + \ldots + f_k,
\]
where $f_0' = f_0+(\varepsilon_i - \mathbb{E}[f])\mathbbm{1}$ and $f=\sum f_i$ is the original HD-Level-Set Decomposition of $f$. Since adding a constant has no effect on the $\ell_\infty$-pseudorandomness, $f'$ remains $(\varepsilon_1,\ldots, \varepsilon_\ell)$-$\ell_\infty$-pseudorandom and, by our $\ell_\infty$ to $\ell_2$ reduction, $(\varepsilon_1,\ldots, \varepsilon_\ell)$-$\ell_2$-pseudorandom as well. Applying \Cref{cor:ell2-proj} then gives:
\begin{align*}
\langle f +(\varepsilon_i - \mathbb{E}[f])\mathbbm{1}, f_i \rangle &\leq {k \choose  i}\varepsilon_i\mathbb{E} [f+(\varepsilon_i - \mathbb{E}[f])\mathbbm{1}] + c\gamma \langle f + (\varepsilon_i - \mathbb{E}[f])\mathbbm{1}, f + (\varepsilon_i - \mathbb{E}[f])\mathbbm{1} \rangle \\
&\leq \left ({k \choose  i}+c\gamma\right)\varepsilon_i^2 + c\gamma\langle f, f \rangle
\end{align*}
where we have used the fact that $f_i$ is orthogonal to $\mathbbm{1}$ for all $i>0$. We are left to deal with the case that $\varepsilon_i = 0$, which follows from a limiting argument applying the above to any $\varepsilon > 0$.
\end{proof}

\section{Expansion of HD-walks}\label{sec:expansion}
In this section we characterize the edge expansion of HD-walks, proving \Cref{thm:new-intro-local-vs-global}, \Cref{intro:new-l2-pr-sets-expand}, and \Cref{cor:new-non-expansion}. As a reminder, these results focus on two key aspects of edge-expansion: the expansion of links, and the structure of non-expanding sets. We'll start with the former, but first let's recall the definition of edge-expansion (specified to HD-walks for simplicity).
\begin{definition}[Weighted Edge Expansion]
Given a weighted simplicial complex $(X,\Pi)$, a $k$-dimensional HD-Walk $M$ over $(X,\Pi)$, and a subset $S \subset X(k)$, the weighted edge expansion of $S$ is
\[
\phi(S) = \underset{v \sim \Pi_k|_S}{\mathbb{E}}\left[ M(v,X(k) \setminus S)\right],
\]
where 
\[
M(v, X(k) \setminus S) = \sum\limits_{y \in X(k) \setminus S} M(v,y)
\]
and $M(v,y)$ is the transition probability from $v$ to $y$.
\end{definition}
We start by proving \Cref{thm:new-intro-local-vs-global}: that the expansion of $i$-links is (up to $O(\gamma)$ error) exactly controlled by the eigenvalue of the $i$th eigenstrip.
\begin{theorem}[Local Expansion vs Global Spectra]\label{thm:local-vs-global}
Let $(X,\Pi)$ be a $d$-dimensional two-sided $\gamma$-local-spectral expander with $\gamma \leq 2^{-\Omega(k)}$, and $M$ a $k$-dimensional, complete HD-walk with $k<d$. Then for all $0 \leq i \leq k$ and $\tau \in X(i)$:
\[
\phi(X_\tau) \in 1 - \lambda_{i}(M) \pm c\gamma,
\]
where $c \leq w(M) h(M)^2 2^{O(k)}$.
\end{theorem}
The main idea behind \Cref{thm:local-vs-global} is simply to show that
the indicator function of any $i$-link always lies almost entirely in $V_k^i$ (or equivalently, almost entirely in the $i$th eigenstrip $W_k^i$).
\begin{lemma}\label{lemma:link-projection}
Let $(X,\Pi)$ be a $d$-dimensional two-sided $\gamma$-local-spectral expander with $\gamma \leq 2^{-\Omega(k)}$. Then for all $0 \leq i \leq k<d$ and $\tau \in X(i)$, $1_{X_\tau}$ lies almost entirely in $V_k^i$. That is for all $j \neq i$:
\[
\langle 1_{X_\tau}, 1_{X_\tau,j} \rangle \leq c\gamma \langle 1_{X_\tau}, 1_{X_\tau} \rangle
\]
where $c \leq 2^{O(k)}$.
\end{lemma}
\begin{proof}
It is enough to analyze the expansion of $X_\tau$ with respect to $N^1_k$ since the quantity can both be analyzed directly, and expressed in terms of $X_\tau$'s HD-Level-Set decomposition. For the direct analysis, recall that $N^1_k$ describes the process of moving from a $k$-face $\sigma$ to a $(k+1)$-face $T = \sigma \cup \{v\}$, then back to a $k$-face $\sigma' \subset T$. Crucially, the latter step is performed uniformly at random. Applying $N^1_k$ to any element in $X_\tau$, the probability of returning to $X_\tau$ is exactly the probability that we remove an element in $T \setminus \{\tau\}$ in the final step, which gives:
\[
\bar{\phi}(1_{X_\tau}) = \frac{k+1-i}{k+1}.
\]
On the other hand, we may also expand out $\bar{\phi}(1_{X_\tau})$ in terms of $1_{X_\tau}$'s HD-Level-Set decomposition. Using the fact that $1_{X_\tau} = {k \choose i}U^k_i1_{\tau} \in V_k^0 \oplus \ldots \oplus V_k^i$, we may write:
\begin{align*}
\bar{\phi}(1_{X_\tau}) &= \frac{1}{\langle 1_{X_\tau}, 1_{X_\tau} \rangle }\sum\limits_{j=0}^{i} \langle 1_{X_\tau}, N^1_k 1_{X_\tau,j} \rangle\\
&= \frac{1}{\langle 1_{X_\tau}, 1_{X_\tau} \rangle }\sum\limits_{j=0}^{i} \frac{k+1-j}{k}\langle 1_{X_\tau}, 1_{X_\tau,j} \rangle + \frac{1}{\langle 1_{X_\tau}, 1_{X_\tau} \rangle }\sum\limits_{s=0}^i \langle \mathbbm{1}_{X_\tau}, \Gamma_s\rangle,
\end{align*}
where $\norm{\Gamma_s} \leq 2^{O(k)}\gamma \norm{\mathbbm{1}_{X_\tau}}$ by \Cref{prop:hdx-eig-vals} and the fact that $\norm{\mathbbm{1}_{X_\tau,s}} \leq (1+2^{O(k)}\gamma)\norm{\mathbbm{1}_{X_\tau}}$. Finally, applying Cauchy-Schwarz to the error term gives:
\begin{equation}\label{eq:link-proj}
\bar{\phi}(1_{X_\tau}) \in \frac{1}{\langle 1_{X_\tau}, 1_{X_\tau} \rangle }\sum\limits_{j=0}^{i} \frac{k+1-j}{k}\langle 1_{X_\tau}, 1_{X_\tau,j} \rangle \pm c\gamma
\end{equation}
for $c \leq 2^{O(k)}\gamma$. Recall by approximate orthogonality (\Cref{lemma:HDX-approx-ortho}), the projections $\langle 1_{X_\tau}, 1_{X_\tau,j} \rangle$ cannot be too negative, that is $\langle 1_{X_\tau}, 1_{X_\tau,j} \rangle \geq -c\gamma \langle  1_{X_\tau},  1_{X_\tau} \rangle$ for some $c \leq 2^{O(k)}$. Then if there exists some $j \neq i$ such that $\langle 1_{X_\tau}, 1_{X_\tau,j} \rangle >c_2\gamma \langle 1_{X_\tau}, 1_{X_\tau} \rangle$ for large enough $c_2 \leq 2^{O(k)}$, the LHS of \Cref{eq:link-proj} lies strictly above $\frac{k+1-i}{k+1}$, giving the desired contradiction.
\end{proof}
Since the HD-Level-Set decomposition of $1_{X_\tau}$ has no dependence on the walk in question, \Cref{thm:local-vs-global} follows almost immediately.
\begin{proof}[Proof of \Cref{thm:local-vs-global}]
The expansion of $X_\tau$ may be written as:
\begin{align*}
\phi(X_\tau) &= 1 - \frac{1}{\alpha}\langle \mathbbm{1}_{X_\tau}, M \mathbbm{1}_{X_\tau} \rangle\\
&= 1 - \frac{1}{\alpha}\sum\limits_{s=0}^i \langle \mathbbm{1}_{X_{\tau}}, M \mathbbm{1}_{X_{\tau},s} \rangle,
\end{align*}
where $\alpha$ is the density of $X_\tau$ and $\mathbbm{1}_{X_\tau,s} \in V_k^s$. By \Cref{prop:hdx-eig-vals} we can simplify the sum up to an error term:
\begin{align*}
    \phi(X_\tau) &= 1 - \frac{1}{\alpha}\sum\limits_{s=0}^i \langle \mathbbm{1}_{X_{\tau}}, M \mathbbm{1}_{X_{\tau},s} \rangle\\
    & = 1 - \frac{1}{\alpha}\sum\limits_{s=0}^i \lambda_j(M)\langle \mathbbm{1}_{X_{\tau}}, \mathbbm{1}_{X_{\tau},s} \rangle + \frac{1}{\alpha}\sum\limits_{s=0}^i \langle \mathbbm{1}_{X_\tau}, \Gamma_s\rangle,
\end{align*}
where by the fact that $\norm{\mathbbm{1}_{X_\tau,s}} \leq (1+2^{O(k)}\gamma)\norm{\mathbbm{1}_{X_\tau}}$ we have $\norm{\Gamma_s} \leq O\left(w(M) h(M)(h(M)+k){k \choose s}\gamma \norm{\mathbbm{1}_{X_\tau}}\right)$. By \Cref{lemma:link-projection} for all $s \neq i$ we can absorb $\langle \mathbbm{1}_{X_{\tau}}, \mathbbm{1}_{X_{\tau},s} \rangle$ into the error term which gives:
\[
\phi(X_\tau) \in 1 - \frac{1}{\alpha}\lambda_i(M)\langle \mathbbm{1}_{X_\tau}, \mathbbm{1}_{X_\tau,i} \rangle \pm c\gamma
\]
for $c \leq w(M)h(M)^22^{O(k)}$. Finally, by \Cref{lemma:link-projection} and approximate orthogonality, we know that $\langle \mathbbm{1}_{X_\tau}, \mathbbm{1}_{X_\tau,i} \rangle$ is very close to $\alpha$:
\[
(1-c_1\gamma)\alpha \leq \langle \mathbbm{1}_{X_\tau}, \mathbbm{1}_{X_\tau,i} \rangle \leq (1+c_2\gamma)\alpha
\] 
where $c_1,c_2 \leq 2^{O(k)}$. Combining this with the above completes the proof. Note that the proof falls through for $d=k$ since we cannot analyze $N_k^1$ on such a complex. The upper bound, however, still holds in this case simply by expanding out the inner product and bounding the inner summation using $\lambda_i(M)$. 
\end{proof}
\Cref{thm:local-vs-global} will form one of two core pieces of our algorithm with unique games. The fact that links corresponding to bad eigenvalues have poor expansion will allow us to patch together good local solutions into a global solution without seeing too much interference. Moreover, close connection between local expansion and stripped eignenvalues will result in the performance of our algorithm being tied directly to ST-rank (which we recall here for convenience).
\begin{definition}[Stripped Threshold Rank]
Let $(X,\Pi)$ be a two-sided $\gamma$-local-spectral expander and $M$ a $k$-dimensional HD-walk with $\gamma$ small enough that the HD-Level-Set Decomposition has a corresponding decomposition of disjoint eigenstrips $C_k=\bigoplus W_k^i$.\footnote{It should be noted that when the HD-Level-Set Decomposition has spaces with the same approximate eigenvalue, their corresponding eigenstrips technically must be merged. However, since this detail has no effect on our arguments, we ignore it in what follows. We now show how to express the expansion of a set $S \subset X(k)$ with respect to an HD-walk $M$ in terms of the pseudorandomness of $S$ and the ST-rank of $M$.} The ST-Rank of $M$ with respect to $\delta$ is the number of strips containing an eigenvector with eigenvalue at least $\delta$:
\[
R_\delta(M) = |\{ W_k^i : \exists f \in V^i, Mf = \lambda f, \lambda > \delta\}|.
\]
We often write just $R_\delta$ when $M$ is clear from context.
\end{definition}
A basic corollary of \Cref{thm:local-vs-global} is that, like the Johnson graphs, links are small, non-expanding sets (at lesat when $|X(1)| \gg k$ or $\gamma$ is small). The second core piece of our algorithm for unique games relies on a certain converse to this result: that \textit{all} non-expanding sets are explained by links. As discussed in \Cref{sec:intro} and \Cref{sec:results}, we consider two regimes for this problem. The first, which we call the $\ell_2$-variant, claims that any non-expanding set must have high \textit{variance} over links---this regime is useful for constructing algorithms for unique games, as we'll show in the next section. The second is the $\ell_\infty$-variant, which claims that any non-expanding set must have a high \textit{maximum} over links---this regime is useful in hardness of approximation. We examine both regimes through their converse: that both $\ell_2/\ell_\infty$-pseudorandom sets expand near-perfectly.
\begin{theorem}\label{thm:hdx-expansion}
Let $(X,\Pi)$ be a two-sided $\gamma$-local-spectral expander, $M$ a $k$-dimensional, complete HD-walk, and let $\gamma$ be small enough that the eigenstrip intervals of \Cref{thm:approx-ortho} are disjoint. For any $\delta>0$, let $r=R_\delta(M)-1$. Then the expansion of a set $S \subset X(k)$ of density $\alpha$ is at least:
\[
\phi(S) \geq 1 - \alpha - (1-\alpha)\delta - c\gamma - \sum\limits_{i=1}^{r} (\lambda_i(M)-\delta){k \choose i}\varepsilon_i,
\]
where $\lambda_i(M)$ is the approximate eigenvalue given by \Cref{prop:hdx-eig-vals}, $S$ is either $(\varepsilon_1,\ldots,\varepsilon_r)$-$\ell_2$-pseudorandom or $(\varepsilon_1,\ldots,\varepsilon_r)$-$\ell_\infty$-pseudorandom, and $c \leq w(M)h(M)^22^{O(k)}$.
\end{theorem}
\begin{proof}
Recall that the expansion of $S$ may be written as:
\[
\phi(S) = 1-\frac{1}{\mathbb{E}[\mathbbm{1}_S]}\langle \mathbbm{1}_S, M \mathbbm{1}_S \rangle.
\]
Decomposing $\mathbbm{1}_S=\mathbbm{1}_{S,0} + \ldots + \mathbbm{1}_{S,k}$ by the HD-Level-Set Decomposition, we have:
\begin{align*}
\phi(S) &= 1- \frac{1}{\mathbb{E}[\mathbbm{1}_S]}\sum\limits_{i=0}^k \langle \mathbbm{1}_{S},M \mathbbm{1}_{S,i} \rangle\\
&= 1- \frac{1}{\mathbb{E}[\mathbbm{1}_S]}\sum\limits_{i=0}^k \lambda_i(M)\langle \mathbbm{1}_{S}, \mathbbm{1}_{S,i} \rangle + \frac{1}{\mathbb{E}[\mathbbm{1}_S]}\sum\limits_{i=1}^k\langle \mathbbm{1}_{S},\Gamma_i \rangle\\
\end{align*}
where by \Cref{prop:hdx-eig-vals} $\norm{\Gamma_i} \leq O\left(w(M) h(M)(h(M)+k){k \choose i}\gamma \norm{\mathbbm{1}_{S,i}}\right)$. Using Cauchy-Schwarz and the fact that $\norm{\mathbbm{1}_{S,i}} \leq (1+2^{O(k)}\gamma)\norm{\mathbbm{1}_{S}}$ (this follows from approximate orthogonality, see \cite[Corollary 8.13]{dikstein2018boolean}) we can simplify this to 
\begin{align*}
\phi(S) &\geq 1- \frac{1}{\mathbb{E}[\mathbbm{1}_S]}\sum\limits_{i=0}^{k} \lambda_i(M)\langle \mathbbm{1}_{S}, \mathbbm{1}_{S,i} \rangle - e\gamma,
\end{align*}
where $e \leq w(M) h(M)^2 2^{O(k)}$. Since $M$ is a complete walk, we know the $\lambda_i(M)$ decrease monotonically and as long as $\gamma$ is sufficiently small, correspond to the eigenvalues in strip $W^i$ as well. Thus we may write:
\begin{align*}
\phi(S) &\geq 1-e\gamma - \frac{1}{\mathbb{E}[\mathbbm{1}_S]}\sum\limits_{i=0}^{r} \lambda_i(M)\langle \mathbbm{1}_{S}, \mathbbm{1}_{S,i} \rangle - \frac{\delta}{\mathbb{E}[\mathbbm{1}_S]} \sum\limits_{i=r+1}^k \langle \mathbbm{1}_{S}, \mathbbm{1}_{S,i} \rangle\\
&= 1-e_2\gamma - \frac{1}{\mathbb{E}[\mathbbm{1}_S]}\sum\limits_{i=1}^{r} \lambda_i(M)\langle \mathbbm{1}_{S}, \mathbbm{1}_{S,i} \rangle - \delta\left(1 - \frac{1}{\mathbb{E}[\mathbbm{1}_S]}\sum\limits_{i=0}^{r} \langle \mathbbm{1}_{S}, \mathbbm{1}_{S,i} \rangle \right)\\
&= 1-e_2\gamma - \delta - \frac{1}{\mathbb{E}[\mathbbm{1}_S]}\sum\limits_{i=1}^{r} (\lambda_i(M)-\delta)\langle \mathbbm{1}_{S}, \mathbbm{1}_{S,i} \rangle\\
\end{align*}
Finally, recalling that $\mathbbm{1}_{S,i}=\mathbb{E}[\mathbbm{1}_S]\vec{1}$ and applying \Cref{thm:body-local-spec-proj} gives the $\ell_2$-variant result, and the $\ell_\infty$-variant follows immediately from our $\ell_\infty$ to $\ell_2$ reduction (\Cref{lem:reduction}).
\end{proof}
We now turn to the discussion of a surprisingly subtle point: the tightness of \Cref{thm:hdx-expansion}. There are two main parameters of interest: the pseudorandomness parameter $\varepsilon$, and the level of the HD-walk $k$. We'll first prove that in both the $\ell_2$ and $\ell_\infty$ regimes, if we fix the dependence on $\varepsilon$ to be linear, our bound is exactly tight. 
\begin{proposition}\label{prop:HDX-tight}
Let $X=J(n,d)$ be the complete complex, $2k-t \leq d$, $m | n$,  and $B_m$ be the the set of all $k$-faces $\binom{[n/m]}{k}$. Then for any $t$, $B_m$ witnesses the tightness of \Cref{thm:hdx-expansion} with respect to $S_k^{k-t}$ as $n,m \to \infty$.
\end{proposition}
\begin{proof}
First, note that it is enough to examine only the $\ell_\infty$ bound, since the function in question has locally-constant sign the tightness of the $\ell_2$ bound follows from \Cref{lem:reduction}.

By direct computation, it is not hard to show that the expansion of $B_m$ with respect to $S_k^{k-t}$ is:
\[
\phi(B_m) = 1 - \frac{{\frac{n}{m}-k \choose k-t}}{{n-k \choose k-t}} = 1 - \frac{m^t}{m^k} + O_{k,m}(1/n)
\]
On the other hand, we can directly compute that $B_m$ is $(\varepsilon_1,\ldots,\varepsilon_k)$-pseudorandom, where:
\[
\varepsilon_i \leq \frac{\binom{\frac{n}{m}-i}{k-i}}{\binom{n-i}{k-i}}
\]
Since the complete complex is a two-sided $O(1/n)$-local-spectral expander \cite{dikstein2018boolean}, for large enough $n$ \Cref{thm:hdx-expansion} gives the bound:
\begin{align*}
\phi(B_m) &\geq 1 - \sum\limits_{i=0}^t \binom{t}{i}\frac{\binom{\frac{n}{m}-i}{k-i}}{\binom{n-i}{k-i}} - O_{k,m}(1/n)\\
&= 1 - \frac{\binom{n/m}{k}}{\binom{n}{k}}\sum\limits_{i=0}^t \binom{t}{i}\frac{\binom{n}{i}}{\binom{n/m}{i}} - O_{k,m}(1/n)\\
&\geq 1 - m^{-k}\sum\limits_{i=0}^t {t \choose i}m^i - O_{k,m}(1/n)\\
&= 1 - \frac{(m+1)^t}{m^{k}} - O_{k,m}(1/n)\\
\end{align*}
Thus we see that for large $n$, the bound is tight up to the leading term in $m$.
\end{proof}
However, \Cref{prop:HDX-tight} does not preclude a bound with better (or even no) dependence on $k$. Indeed, proving that such a bound holds in the $\ell_\infty$-regime for the Johnson graphs was an important stepping stone in the proof of the 2-2 Games Conjecture \cite{khot2018small}. On the other hand, we can actually prove that a $k$-independent bound is \textit{impossible} in the $\ell_2$-regime. The intuition behind the difference can be summarized by examining the behavior of these two variants on a link. In the $\ell_\infty$ regime, links are of course $\Omega(1)$-pseudorandom. On the other hand, somewhat counter-intuitively, links are actually $O({k \choose i}^{-1})$-pseudorandom in the $\ell_2$-regime. Since links tend to have poor expansion, the dependence on $k$ in our bound has to make up for the fact that links are $O({k \choose i}^{-1})$-$\ell_2$-pseudorandom (this also explains the particular dependence on ${k \choose i}$).

\begin{proposition}\label{prop:ell2-k-dependence}
For every $\ell \in \mathbb{N}$ and $c_1(\ell),c_2(\ell) > 0$, there exist $n \gg k \gg \ell$ and an $(\varepsilon_1,\ldots,\varepsilon_\ell)$-$\ell_2$-pseudorandom subset $S$ of $S_k^{k/2}$ on $J(n,2k)$\footnote{Note that this is exactly the Johnson Graph $J(n,k,k/2)$ mentioned in \Cref{sec:results}.} of density $\alpha$ satisfying:
\[
\phi(S) < 1 - \alpha - \lambda_{\ell+1}(S_k^{k/2}) - c_1\varepsilon^{c_2}
\]
\end{proposition}
\begin{proof}
The result follows from letting $S$ be any $\ell$-link of the complete complex $J(n,2k)$ for sufficiently large $n$,$k$. In particular, notice that the expansion of $S$ with respect to $S_k^{k/2}$ is then:
\[
\phi(S) \leq 1 - 2^{-\ell} + O_k\left(\frac{1}{n} \right).
\]
On the other hand, note that $\lambda_{\ell+1}(S_k^{k/2}) \in 2^{-\ell-1} \pm O_k(\frac{1}{n})$. Since one can also check that $S$ is $(\varepsilon_1,\ldots,\varepsilon_\ell)$-$\ell_2$-pseudorandom for $\varepsilon_i \leq {k \choose i}^{-1} + O_k\left(\frac{1}{n}\right)$, we for large enough $n \gg k$ that:
\begin{align*}
1 - \alpha - \lambda_{\ell+1}(S_k^{k/2}) - c_1\varepsilon^{c_2} &\geq 1 - 2^{-\ell-1} - O_k\left(\frac{1}{n}\right) - o_k(1) > \phi(S)
\end{align*}
as desired.
\end{proof}
In other words, no $k$-independent version of \Cref{thm:hdx-expansion} exists for the $\ell_2$-variant, as such a result would violate the upper bound in \Cref{prop:ell2-k-dependence}. While this doesn't directly rule out a reduction from $\ell_\infty$ to $\ell_2$ that proves a $k$-independent bound for the former (the reduction itself would need to be \textit{$k$-dependent} in this case), it's known to be impossible in our framework which encompasses the Grassmann where such a result is known to be false \cite{dinur2018non}.

The difference between a $k$-independent bound and the regime we consider is most stark when examining the contrapositive of \Cref{thm:hdx-expansion}, which states that non-expanding sets must be concentrated inside links.
\begin{corollary}\label{cor:hdx-non-expansion}
Let $(X,\Pi)$ be a two-sided $\gamma$-local-spectral expander, $M$ a $k$-dimensional, complete HD-walk, and let $\gamma$ be small enough to satisfy the requirements of \Cref{thm:approx-ortho}. Then for any $\delta> 0$, if $S \subset X(k)$ is a set of density $\alpha$ and expansion:
\[
\phi(S) < 1 - \alpha - (1-\alpha)\delta - c\gamma
\]
for $c \leq  w(M)h(M)^22^{O(k)}$, then $S$ is non-trivially correlated with an $i$-link for $1 \le i \le R_{\delta/2}$:
\[
\exists 1 \leq i \leq R_{\delta/2}, \tau \in X(i): \underset{X_\tau}{\mathbb{E}}[\mathbbm{1}_S] \geq \alpha + \frac{\delta}{c_2R_{\delta/2}{k \choose i}\lambda_i(M)},
\]
where $c_2>2$ is some small absolute constant.
\end{corollary}
Notice that the excess correlation implied by \Cref{cor:hdx-non-expansion} decays as $k$ grows large (except in the case of very deep walks like $S_k^{k-O(1)}$); this is one of the main obstructions to using results like \Cref{cor:hdx-non-expansion} for hardness of unique games. On the other hand, we now turn our attention to \textit{algorithms} for unique games, where the $\ell_2$-regime gets its chance to shine.

\section{Playing Unique Games on HD-Walks}\label{sec:unique-games}
Following the recent algorithmic framework of \cite{bafna2020playing}, we show how to translate our spectral machinery and combinatorial characterization of non-expansion into a polynomial time algorithm for unique games over HD-walks whose finer-grained runtime and approximation guarantees depend on ST-rank. Before we dive into the theorems let's start with some notation for this section.

\paragraph{Notation:} We briefly comment on the use of $\mathcal{I}=(M,\mathcal{S})$ to denote an affine UG instance over a random-walk $M$, since unique games are formally defined over graphs, not random walks. In particular, it is well known that every HD-walk $M$ over $X(k)$ uniquely corresponds to an undirected weighted graph $G_M = (X(k), E)$ (see Section~\ref{sec:ug-rwalk}). Thus by $\mathcal{I} = (M,\mathcal{S})$, we really mean that the constraints $\mathcal{S}$ are over the edges of $G$ and the value is calculated according to the distribution over edges $E$. We will use the $\widetilde{O}(\cdot)$ notation in this section to hide log factors, that is, $\tO(f)$ denotes $O(f\log^c f)$ for any constant $c > 0$.

With notation out of the way, let us describe the main theorem:

\begin{theorem}\label{thm:unique-games}
For any $\varepsilon \in (0,.01)$, there exists an algorithm $\mathcal{A}$ with the following guarantee. Let $(X,\Pi)$ be a $d$-dimensional two-sided $\gamma$-local spectral expander and $M$ be a $k$-dimensional complete HD-walk over $X$ such that $\gamma \leq w(M)^{-1}2^{-\Omega(k+h(M))}$ and $d>k$. Let $\mathcal{I} = (M,\mathcal{S})$ be an instance of affine unique games over $M$ with value at least $1-\varepsilon$. Then $\mathcal{A}$ outputs an $\Omega\left(\frac{\varepsilon^3}{{k \choose r(2\varepsilon)}^2}\right)$-satisfying assignment in time $|X(k)|^{\tO\left({k \choose r(2\varepsilon)}\frac{1}{\varepsilon}\right)}$, where $r(\varepsilon)=R_{1-16\varepsilon}(M)$ is the ST-rank of $M$.
\end{theorem}

Since the HDX literature focuses mostly on canonical and partial-swap walks, we also give the specification of \Cref{thm:unique-games} to this class for concreteness. Here we see that the finer-grained performance of our algorithm depends on the \textit{depth} of the walk.

\begin{corollary}\label{cor:ug-depth}
For any $\varepsilon \in (0,.01)$, there exists an algorithm $\mathcal{A}$ with the following guarantee. Let $(X,\Pi)$ be a $d$-dimensional two-sided $\gamma$-local spectral expander and $M$ be a $k$-dimensional canonical or partial-swap walk of depth $\beta \in [0,1]$ over $X$ such that $\gamma \leq w(M)^{-1}2^{-\Omega(k+h(M))}$ and $d>k$. Let $\mathcal{I} = (M,\mathcal{S})$ be an instance of affine unique games over $M$ with value at least $1-\varepsilon$. Then $\mathcal{A}$ outputs an $\Omega\left(\frac{\varepsilon^3}{{k \choose c\frac{\varepsilon}{\beta}}^2}\right)$-satisfying assignment in time $|X(k)|^{\tO\left({k \choose c\varepsilon/\beta}\frac{1}{\varepsilon}\right)}$ for some absolute constant $c>0$.
\end{corollary}

To prove \Cref{thm:unique-games}, we follow the general SoS algorithmic paradigm introduced by BBKSS for Johnson graphs (partial-swap walks on the complete complex), described in \Cref{sec:algorithm} for completeness. As discussed in Section~\ref{sec:ug-proof-overview} we abstract BBKSS' analysis to rely on two core structural properties true of the graphs underlying HD-walks:
\begin{enumerate}
    \item There exists a low-degree Sum-of-Squares proof that non-expanding sets have high variance in size across links.
    \item For every small enough $\eps$, there exists a parameter $r=r(\varepsilon)$ such that:
    \begin{enumerate}
        \item The $(r+1)$-st largest (distinct) stripped-eigenvalue of $G$ is small:
        \[ 
        \lambda_{r}\leq 1-\Omega(\varepsilon)
        \]
        \item The expansion of any $r$-link is small as well:
        \[
        \forall \tau \in X(s): \Phi(X_\tau) \leq O(\varepsilon).
        \]
    \end{enumerate}
\end{enumerate}
Recall that we have already proved property (1) in \Cref{thm:hdx-expansion}, albeit not in the low-degree SoS proof system. In \Cref{sec:restricted-con-round} we modify the proof of \Cref{thm:hdx-expansion} to show that it has a degree 2 SoS proof. For property (2), it is clear that the existence of $r(\eps)$ is an inherent consequence of \Cref{thm:local-vs-global}, and furthermore it is exactly the $(1 - O(\eps))$ ST-Rank of the HD-walk. Analyzing the generalized BBKSS algorithm using these properties combined with a few more technicalities (described in Section~\ref{sec:restricted-con-round}) then gives an efficient algorithm for unique games over HD-walks on two-sided local spectral expanders.

\subsection{Background for Unique Games and SoS}\label{sec:SoS-background}

Proving \Cref{thm:unique-games} from the ground up requires substantial background in the SoS framework. However, since we mostly rely on a number of higher level results from \cite{bafna2020playing} for the SoS side of our work, we cover here only background necessary to understand our methods, and refer the reader to the surveys of~\cite{BS14,TCS-086,RSS18} and additionally Sections 1, 2, and A of \cite{bafna2020playing} for more information.
\\
\\
\noindent\textbf{The Sum of Squares framework} is a method for approximating polynomial optimization problems through semi-definite programming relaxations. In particular, given the problem:
\[
\text{Maximize } p \in \R[x_1,\ldots,x_n]~\text{constraint to } \{q_i = 0\}_{i=1}^m,
\]
for $q_i \in \mathbb{R}[x_1,\ldots,x_n]$, the Degree-$D$ Sum of Squares semidefinite programming relaxation outputs in time $n^{O(D)}$ a \textit{pseudoexpectation operator}  $\pE: \text{poly}_{\mathbb{R}}(n,D) \to \mathbb{R}$ over polynomials in $\mathbb{R}[x_1,\ldots,x_n]$ of degree at most $D$ satisfying:
\begin{enumerate}
    \item Scaling: $\pE[1] = 1$
    \item Linearity: $\pE[af(x)+bg(x)] = a\pE[f(x)] + b\pE[g(x)]$
    \item Non-negativity (for squares): $\pE[f(x)^2] \geq 0$
    \item Program constraints: $\pE[f(x)q_i(x)] = 0$
    \item Optimality: $\pE[p(x)] \geq \max_x \left\{p(x): \{q_i = 0\}_{i=1}^m\right\}$
\end{enumerate}

Note that the first four properties give the definition of a pseudoexpectation (under constraints $\{q_i=0\}_{i=1}^m$), whereas the fifth is promised by the SoS relaxation. The pseudoexpectation operator can equivalently be defined as a weighted expectation over a ``pseudodistribution'' $\mu$ satisfying similar properties (and analogous to an actual distribution). We will sometimes use the pseudodistribution view below when convenient and refer to $\pE[\cdot]$ (or $\pE_{\mu}[\cdot]$ in this case) as pseudomoments of the pseudodistribution (see \cite{barak2014sum} for a more detailed exposition).
\\
\\
\noindent\textbf{A Degree-$D$ Sum of Squares proof} of a polynomial inequality $f(x) \leq g(x)$ (where $f,g$ are polynomials of degree at most $D$) is a method for ensuring the inequality continues to hold over any degree-$D$ pseudoexpectation. In particular, given constraints $\{q_i = 0\}_{i=1}^m$, a degree-$D$ sum of squares proof of $f \leq g$, denoted by:
\[
\{q_i = 0\}_{i=1}^m \vdash_D f \leq g,
\]
is a certificate of the form $g(x) = f(x) + \sum s(x)^2 + \sum_i t(x)q_i(x)$ where all terms have degree at most $D$. Notice that properties 2, 3, and 4 then immediately imply $\pE[f(x)] \leq \pE[g(x)]$.
\\
\\

\noindent\textbf{Conditioning} is a standard algorithmic technique in the SoS paradigm used to improve the value of independently sampling a solution from the output of an SoS semidefinite relaxation (see e.g.\ \cite{barak2011rounding, bafna2020playing}). Given a degree $D$ pseudodistribution $\mu$, and a degree $< D$ sum of squares polynomial $s(x)$, we can define a new pseudodistribution $\mu'$ by \textit{conditioning} on $s(x)$ as follows:
\[
\pE_{\mu'}[f(x) | s(x)] = \frac{\pE_\mu[f(x)s(x)]}{\pE_\mu[s(x)]},
\]
for all polynomials $f(x)$ of degree $\leq D - \deg(s)$. We have that $\mu'$ is a valid pseudodistribution of degree $D - \deg(s)$ that satisfies the axioms satisfied by $\mu$\footnote{Technically only the axioms that were of degree $\leq D - \deg(s)$}. In an algorithmic context, this is often used to restrict the pseudoexpectation to some partial solution.
\\

\subsection{The Algorithm}\label{sec:algorithm}
We now present our algorithm for solving Unique Games on HD-walks. We follow the overall framework of BBKSS, which is based on the Sum-of-Squares semidefinite programming (SDP) relaxation paradigm and its view as optimizing over pseudoexpectation operators. The unique games problem (\Cref{def:ug}) can be written as a polynomial optimization problem. In particular, given an instance $I$ of unique games with alphabet $\Sigma$ and constraints $\mathcal{S}$ over $G(V,E)$ (that are of the form $X_u - X_v = s_{uv} ~ (\text{mod } k)$), consider the following quadratic optimization problem $\mathcal{A}_I$ over variables $\{X_{v,s}\}_{V \times \Sigma}$ that computes $\text{val}(I)$:

\begin{align*}
    \text{Maximize:} &\quad\quad\quad\quad \underset{(u,v) \sim E}{\mathbb{E}}\left[\sum\limits_{s \in \Sigma} X_{u,s}X_{v,\pi_{uv}(s)}\right]\\
    \text{Constraint to:} & \quad\quad\quad\quad\quad\quad\quad \ \ X^2_{v,s} = X_{v,s} & \quad\quad\quad\quad\quad\quad\quad \ \ \forall v \in V, s \in \Sigma\\
    & \quad\quad\quad\quad\quad\quad\quad \ \ X_{v,a}X_{v,b} = 0 & \forall v \in V, a\neq b \in \Sigma \\
    & \quad\quad\quad\quad\quad\quad\quad~~ \sum\limits_{s \in \Sigma}X_{v,s} = 1 & \forall v \in V
\end{align*}

The variables $X_{v,a}$ are $0/1$ indicators that vertex $v \in V$ takes the value $a \in \Sigma$, $E$ is a distribution over the edges of the weighted graph $G$, the constraints are $\pi_{uv}(s) = s - s_{uv} \ (\text{mod }k)$, and the objective function maximizes the fraction of constraints satisfied. We will work with the Degree-$D$ Sum of Squares relaxation of this program, which outputs a degree-$D$ pseudodistribution $\mu$ and corresponding pseudoexpectation operator  $\pE_\mu:X^{\le D} \to \R$, where $X^{\le D}$ is the set of all monomials in the $X$ variables up to degree $D$, and $\pE_\mu$ satisfies the above equality constraints as axioms. The \emph{value} of $\pE_\mu$, with respect to the instance $I$ is denoted by $\val_\mu(I) = \pE_\mu[\val_I(X)] = \pE_\mu[\underset{(u,v) \sim E}{\mathbb{E}}[\sum_{s \in \Sigma} X_{u,s} X_{v,\pi_{uv}(s)}]]$. This operator is obtained in time $|V|^{O(D)}$ such that $\pE_\mu[\text{val}_I(X)] \geq \text{val}(I)$.

We now describe the rounding algorithm that takes a pseudodistribution $\mu$ over an almost-satisfiable UG instance $I$ and produces a high value assignment.

\paragraph{Condition\&Round:} We start with a basic sub-routine which will then be used in the final algorithm for HD-walks. This subroutine takes a pseudodistribution $\mu$ for an affine unique games instance $(G(V,E),\Pi)$ on alphabet $\Sigma$ and outputs an assignment $x \in \Sigma^{V}$ via the following process:
\begin{enumerate}
    \item Sample a vertex $v \in V$ uniformly at random, and condition $\mu$ on event $X_{v,0}=1$ to get the conditioned pseudodistribution $\mu | (X_{v,0}=1)$. 
    \item Sample a solution $x \in \Sigma^V$ by independently sampling a label for every vertex $w \neq v$ from its marginal distribution in $\mu|(X_{v,0}=1)$: $\Pr[x_w = s] = \pE_\mu[X_{w,s} | X_{v,0}=1]$.
\end{enumerate}
Note that $\pE_\mu[\cdot | X_{v,0}=1]$ is by definition conditioning $\mu$ on the polynomial $X_{v,0}$, hence $\pE_\mu[X_{w,s} | X_{v,0}=1] = \pE_\mu[X_{w,s}X_{v,0}]/\pE_\mu[X_{v,0}]$. Following \cite{bafna2020playing}, we define the term Condition\&Round-value (abbreviated to CR-val) of an instance $I$ with respect to a pseudodistribution $\mu$:

\begin{definition}[Condition\&Round value]
The \textit{CR-Value} of the instance $I$ with respect to a pseudodistribution $\mu$ is the expected value of the solution output by Condition\&Round on instance $I$ and pseudodistribution $\mu$, denoted $\text{CR-Val}_\mu(I)$. We drop the subscript $\mu$ when clear from context.
\end{definition}

Before we describe the main algorithm, an iterative framework for applying Condition\&Round, we need to introduce a simple operation on pseudodistributions for affine unique games that allows for ease of analysis:

\paragraph{Symmetrization:} Symmetrization is an operation on pseudodistributions introduced in \cite{bafna2020playing} to take advantage of the symmetric structure of affine unique games. The idea is to average the pseudoexpectation operator over shifts $s \in \Sigma$. Formally, given a degree $D$ pseudodistribution $\mu$, define the symmetrized pseudodistribution $\mu_{sym}$ via its pseudoexpectation operator as follows. For all degree $\leq D$ monomials:
\[
\pE_{\mu_{sym}}[X_{u_1,a_1} \ldots X_{u_t,a_t}] =  \frac{1}{|\Sigma|}\sum\limits_{s \in \Sigma} \pE_\mu[X_{u_1,a_1+s} \ldots X_{u_t,a_t+s}].
\]
We will call a pseudodistribution shift-symmetric if it is invariant under this operation. If $\mu$ is a degree $D$ pseudodistribution that satisfies the unique games axioms $\mathcal{A}_I$ it is easy to verify that the symmetrized pseudodistribution $\mu_{sym}$ is also a valid degree $D$ pseudodistribution satisfying $\mathcal{A}_I$ with $\val(\mu) = \val(\mu_{sym})$. Furthermore symmetrization can be performed in time subquadratic in the description of $\pE_\mu$. As a result, we can perform this operation essentially for free inside our algorithm and therefore assume throughout that we are working with a shift-symmetric pseudodistribution.

We are now ready to describe the main algorithm which is called Iterated Condition\&Round. The algorithm follows the strategy presented in \cite[Algorithm 6.1]{bafna2020playing}, differing mainly in that the parameter $r(\varepsilon)$ satisfying their second condition has been replaced with the $(1-O(\eps))$-ST-Rank of the underlying constraint graph. 

\paragraph{Iterated Condition\&Round:} The full algorithm builds a solution by iteratively applying Condition\&Round to links. Let $M$ be a complete $k$-dimensional HD-walk over a $d$-dimensional two-sided $\gamma$-local spectral expander $(X,\Pi)$ and $G_M = (X(k),E)$ be the corresponding undirected weighted graph on vertex set $X(k)$. Let $I=(M,\mathcal{S})$ be an instance of affine unique games over alphabet $\Sigma$ with $\text{val}(I) \geq 1-\varepsilon$ and $r=R_{1-16\varepsilon}(M)$. Further, given a subset $H \subset X(k)$, let $I_H$ denote the restriction of the instance $I$ to the subgraph vertex-induced by $H$. Given a subroutine for finding a link with high CR-value (see \Cref{prop:restricted-con-round}), the following process returns an $\Omega_{\varepsilon,r,k}(1)$ satisfying assignment.
\begin{enumerate}\label{AlgIT}
    \item Let $\delta(\varepsilon) := \Omega\left(\frac{\varepsilon}{{k \choose r}}\right)$. Solve the Degree-$D=\tO\left(1/\delta(\varepsilon)\right)$ SoS SDP relaxation of unique games, and symmetrize the resulting pseudodistribution to get $\mu_0$. Set $j=1$.
    \item Let $\text{Dif}(j)=\pE_{\mu_0}[val_I(x)] - \pE_{\mu_{j-1}}[val_I(x)]$. While $\text{Dif}(j) \leq \varepsilon$:
    \begin{enumerate}
        \item Find an $r$-link $X_\tau$ such that the CR-Value of  $I|_{X_\tau}$ is at least $\delta(\varepsilon + \text{Dif}(j))$\footnote{When $\gamma,k,d$ satisfy certain inequalities then such a link is guaranteed to exist by \Cref{prop:restricted-con-round}. Therefore we can find it by enumerating over all links $\tau \in X(r)$ and computing CR-$\val(X_\tau)$}.
        \item Let $S_j$ be the subgraph of $X_\tau$ induced by the vertices in $X(k)$ which have not yet been assigned a value in any partial assignment $f_i$, $i\leq j$, and perform Condition\&Round on $S_j$ to get partial assignment $f_j$.
        \item Create a new pseudodistribution $\mu_j$ by making the marginal distribution over assigned vertices uniform and independent of others, i.e. for all degree $\leq D$ monomials let $\pE_{\mu_j}$ be:
        \[
        \pE_{\mu_j}[X_{h_1,a_1}\ldots X_{h_t,a_t}X_{u_1,b_1}\ldots X_{u_m,b_m}] = \frac{1}{|\Sigma|^t} \pE_{\mu_{j-1}}[X_{u_1,b_1}\ldots X_{u_m,b_m}],
        \]
        where $h_i \in S_j$ and $u_i \in X(k) \setminus S_j$. Increment $j \gets j+1$.
    \end{enumerate}
\end{enumerate}

It is worth noting that the Condition\&Round subroutine, and thus the entire Iterated Condition\&Round algorithm, can be derandomized by standard techniques like the method of conditional expectations \cite{bafna2020playing}.

\subsection{Analysis of Algorithm~\ref{AlgIT}}\label{sec:alg-analysis}
BBKSS' analysis of Iterated Condition\&Round algorithm relies heavily on analysing a quantity called the Approximate shift-partition potential defined therein. For completeness we include the definitions of the potential functions used in~\Cref{app:potential}. For the purpose of this section, the potential $\Phi_{\beta,\nu}^I(X)$ can be thought of as a low-degree polynomial (with degree = $\tO(1/\nu)$) in the variables of $\mathcal{A}_I$ ($X = (\ldots,X_{v,a},\ldots)$ where $v \in V$ and $a \in \Sigma$). Given a pseudodistribution $\mu$ for instance $I$, the pseudoexpectation of the potential will be denoted by $\Phi^I_{\beta,\nu}(\mu) = \pE_{\mu}[\Phi_{\beta,\nu}^I(X)]$. $\beta,\nu$ are some parameters in $[0,1]$ that control the 
degree of $\Phi_{\beta,\nu}^I(X)$ and the value of the assignment we finally obtain via Condition\&Round.

The framework developed by BBKSS for analyzing Iterated Condition\&Round can be described as follows. First, they prove that the Condition\&Round subroutine, when run on a pseudodistribution $\mu$ for a unique games instance $I$, returns a high value assignment whenever $\Phi^I_{\beta,\nu}(\mu)$ (Definition~\ref{def:sp-apx}) is high:

\begin{theorem}[BBKSS Theorem 3.3]\label{thm:round}
Let $I = (G, \mathcal{S})$ be an affine instance of Unique Games over graph $G = (V,E)$ and the alphabet $\Sigma$. 
Let $\beta,\nu \in [0,1]$ and $\mu$ be a  degree-$\tO(1/\nu)$ shift-symmetric pseudodistribution satisfying the unique games axioms $\A_I$ (Section~\ref{sec:algorithm}). If $\Phi^I_{\beta,\nu}(\mu) \geq \delta$, then on input $\mu$, the Condition\&Round Algorithm runs in time $\poly(|V(G)|)$ and returns an assignment of expected value at least $(\delta - \nu)(\beta-\nu)$ for $I$.
\end{theorem}

Using this rounding algorithm as a subroutine the analysis of Iterated Condition\&Round
in BBKSS proceeds in the following way:

\begin{enumerate}
    \item For every $I=(G,\Pi)$, where $G$ is a Johnson graph, using the structure theorem for non-expanding sets of $G$ and the properties of the potential function, there exists a link $X_\tau$ of $G$ such that the potential induced on $C$ is large: $\Phi^{I|_{X_\tau}}_{\beta,\nu}(\mu) \geq \poly(\delta)$, where $I|_{X_\tau}$ denotes the UG instance induced on the subgraph $X_\tau$. Using the rounding theorem~\ref{thm:round} we get that the Condition\&Round value on this link is high.
    
    \item Iterative Condition\&Round: Using iterations one can patch together solutions on different links, and since each link is a non-expanding set we get a good solution for the whole graph.
\end{enumerate}

We show that this framework extends to all UG instances on HD-walks. We use the analysis of Condition\&Round (Theorem~\ref{thm:round}) in a blackbox way. We start with the analysis of Point 1, which relies on the technical machinery developed in the previous sections for analysing the non-expanding sets of HD-walks. Using this, we show that there always exists a link with high potential and therefore with high Condition\&Round value. 

\begin{proposition}\label{prop:restricted-con-round}
Let $M$ be a $k$-dimensional complete HD-walk on a $d$-dimensional, two-sided $\gamma$-spectral expander with $\gamma \leq w(M)^{-1}2^{-\Omega(h(M)+k)}$ and $d>k$, and $I$ be an affine unique games instance over $M$ with value at least $1-\eta$ where $1/2^k \leq\eta < .02$. Let $r=r(\eta)=R_{1-16\eta}(M)-1$, and assume $r \leq k/2$. Then given a degree-$\tO\left (\frac{1}{\eta}{k \choose r} \right )$ pseudodistribution satisfying the axioms $\mathcal{A}_I$, we can find in time $|X(k)|^4$ an $r$-link $X_\tau$ with CR-Value$(X_\tau) \geq \Omega\left(\frac{\eta}{{k \choose r}} \right)$.
\end{proposition}

The proof of this lemma is similar in nature to the proof of the analogous lemma in BBKSS, albeit with all of our technical machinery is plugged in. In particular we use a sum-of-squares version of~\Cref{lem:low-level-weight}, the relation between eigenvalues and expansion of links (\Cref{thm:local-vs-global}), and a new property analyzing the expansion of vertices within a link (\Cref{lem:vert-edge-exp}). 
We elaborate on this in Section~\ref{sec:restricted-con-round}. 

Given the ability to find a link with high CR-value, we can use BBKSS analysis of Iterated Condition\&Round (\cite[Lemma 6.12]{bafna2020playing}) as a blackbox to conclude it produces an assignment satisfying a large fraction of the constraints for the \emph{whole graph}.

\begin{lemma}[Lemma 6.12 \cite{bafna2020playing}]\label{lem:iterated-con-round}
Let $\varepsilon \in (0,.01)$, $\delta: [0,1] \to [0,1]$ be any function, and $\delta_{\text{min}}=\min_{\delta(\eta) \in [\varepsilon,2\varepsilon]}$. Let $G$ be a weighted graph\footnote{BBKSS only state the result for regular graphs, but their proof works for weighted graphs too, where the value, expansion etc are measured appropriately according to the edge-weights.} and $I$ be any affine unique games instance on $G$ with alphabet size $|\Sigma| \geq \Omega\left(\frac{1}{\delta_{\text{min}}} \right)$ and value at least $1-\varepsilon$.

Suppose we have a subroutine $\mathcal A$ which, for any $\varepsilon \leq \eta \leq 2\varepsilon$, given as input a shift-symmetric degree-$D$ pseudodistribution $\mu$ satisfying $\mathcal A_I$ with $\pE_\mu[val_I(x)] \geq 1 - \eta$ returns a vertex-induced subgraph $H$ such that:
\begin{enumerate}
    \item The CR-Value of $I_H$ is at least $\delta(\eta)$.
    \item The expansion of $H$ is $O(\eta)$.
\end{enumerate}
Then if $\mathcal A$ runs in time $T(\mathcal A)$, Iterated Condition\&Round\footnote{We note that the version of Iterated Condition\&Round analyzed in \cite{bafna2020playing} differs slightly from the one we present. Our version simplifies their algorithm a bit but the analysis is exactly the same.} outputs a solution for $I$ satisfying an $\Omega(\delta_{\text{min}}^2\varepsilon)$-fraction of edges of $G$ in time $|V(G)|(T(\mathcal A) + |V(G)|^{O(D)})$.
\end{lemma}

With these results in hand, the final observation to prove \Cref{thm:unique-games} lies in the relation between the local expansion of links and global spectra of eigenstrips we proved in \Cref{sec:expansion}. Namely, since the $r$th strip is by definition the last one with eigenvalue worse than $1-O(\varepsilon)$, the expansion of $r$-links is at most $O(\varepsilon)$ by \Cref{thm:local-vs-global}. We therefore meet the conditions of \Cref{lem:iterated-con-round}, and can use Iterated Condition\&Round to output a good global assignment. We now prove \Cref{thm:unique-games} assuming \Cref{prop:restricted-con-round} and \Cref{lem:iterated-con-round} hold.

\begin{proof}[Proof of \Cref{thm:unique-games}]
To start, we first prove that we may assume without loss of generality both that $\varepsilon \geq 1/2^k$, and that $r(2\varepsilon) \leq k/2$ (we will need these properties to meet the conditions of \Cref{prop:restricted-con-round} and \Cref{lem:iterated-con-round}). This relies on the following claim, the proof of which we defer to \Cref{sec:app-alg-proofs}.
\begin{claim}\label{claim:non-lazy-spectra}
Let $M$ be a $k$-dimensional complete HD-walk on a $d$-dimensional, two-sided $\gamma$-local-spectral expander with $d>k$ and $\gamma \leq w(M)^{-1}2^{-\Omega(h(M)+k)}$. If the expected laziness of $M$ (i.e.\ $\mathbb{E}_{\Pi_k}[\id{v}^TM\id{v}]$) is at least $.1$, then the spectral gap of $M$ is at least $\Omega(1/k)$ and $\lambda_{k/2}(M) \leq .68$.
\end{claim}
With this in mind, note that we can always assume the expected laziness $\mathbb{E}[\id{v}^TM\id{v}]$ is at most $1/10$. This follows from the fact that the expected laziness of $M$ exactly corresponds to the probability of drawing a self-edge on the corresponding graph $G_M$ (see \Cref{sec:ug-rwalk}), and in an affine unique game, every self edge is either always or never satisfiable. Since our game has value at least $.99$, in a walk with more than $.1$ laziness, at least $.09$ of the weight on self-edges must be satisfiable, and therefore \textit{every} assignment will satisfy our approximation guarantee. Assuming then that $\mathbb{E}[\id{v}^TM\id{v}] \leq .1$, \Cref{claim:non-lazy-spectra} implies that the spectral gap of $M$ is at least $\Omega(1/k)$, in which case standard algorithms for unique games on expander graphs (e.g.\ \cite{makarychev2010play}) give the desired result. Similarly, we have $\lambda_{k/2} \leq .68 \leq 1-32\varepsilon$ by assumption, so we may further assume $r(2\varepsilon) \leq k/2$ as desired.

Now that we have $1/2^k \leq\varepsilon < .02$ and $r \leq k/2$, we are in position to solve the degree-$D$ SoS relaxation of the unique games integer program $\mathcal{A}_I$ for $D=\tO(\frac{1}{\varepsilon} {k \choose r(2\varepsilon)})$, and apply \Cref{prop:restricted-con-round} to build a sub-route that satisfies \Cref{lem:iterated-con-round}. Namely, we have that for any $\varepsilon \leq \eta \leq 2\varepsilon$ and pseudodistribution of value at least $1 - \eta$, \Cref{prop:restricted-con-round} finds a link $X_\tau$ with high Condition\&Round value: 
\[
\text{CR-val}(X_\tau) \geq \Omega\left(\frac{\eta}{{k \choose r(\eta)}}\right).
\]
Further, note that by our assumptions on $\gamma$ and the fact that $M$ is complete, \Cref{prop:eig-decrease} gives that the eigenvalues corresponding to each eigenstrip strictly decrease, and therefore by definition of ST-Rank that the approximate eigenvalue corresponding to the $r$th eigenstrip is at least $1-O(\eta)$. By \Cref{thm:local-vs-global}, this implies that $X_\tau$ has poor expansion:
\[
\Phi(X_\tau) \leq O(\eta) + w(M)h(M)^22^{O(k)}\gamma \leq O(\eta)
\]
since $\eta \geq \varepsilon \geq w(M)2^{O(h(M)+k)}\gamma$ by assumption. As a result, this sub-routine satisfies the conditions of \Cref{lem:iterated-con-round} with $\delta$ set to $\delta(\eta) = \Omega\left(\frac{\eta}{{k \choose r(\eta)}}\right)$. The only catch is that we need the alphabet $\Sigma$ to satisfy $|\Sigma| \leq \Omega\left (\frac{1}{\delta_{\min}} \right)$. This can be assumed without loss of generality, since a random solution satisfies a $1/|\Sigma|$ fraction of constraints in expectation (and can be easily derandomized), which satisfies our approximation guarantee if $|\Sigma|\leq O\left (\frac{1}{\delta_{\min}} \right)$. As a result, applying \Cref{lem:iterated-con-round} outputs a $\Omega(\delta^2_{\text{min}}\varepsilon)$-satisfying solution for $\delta_{\text{min}} = \min_{\eta \in [\varepsilon,2\varepsilon]}\left\{ \frac{\varepsilon}{{k \choose r(\eta)}}\right\} \geq \Omega\left (\frac{\varepsilon}{{k \choose r(2\varepsilon)}} \right)$. The running time guarantee follows from noting that the sub-routine promised by \Cref{prop:restricted-con-round} runs in time $|V(G)|^4$, so \Cref{lem:iterated-con-round} then runs in time $|V(G)|^{O(D)}$ with $|V(G)| = |X(k)|$. Since solving the original Degree-$D$ SoS relaxation also takes time only $|V(G)|^{O(D)}$, we get the desired result.
\end{proof}

\subsection{Proof of Proposition~\ref{prop:restricted-con-round}}\label{sec:restricted-con-round}
Now that we have given the algorithmic background, we prove the main technical lemma behind the analysis. The broad outline of the proof is as follows:

\begin{enumerate}
\item We first prove a structure theorem (\Cref{thm:structure-SoS}) for non-expanding sets of HD-walks. We show an SoS proof of the fact that every non-expanding set must have large variance of size when restricted to links of the complex.
\item Using the structure theorem, in~\Cref{prop:high-global} we show that given a pseudodistribution $\mu$ with objective value $1-\eta$ for unique games over an HD-walk $M$, one can find a link $X_\tau$ with high \emph{global} shift-partition potential (\Cref{def:sp-global}).
\item In the final step (\Cref{lem:ug-local-to-global}), we relate the global shift-partition potential to the shift-partition potential on the subgraph induced by $X_\tau$\footnote{This is just the potential measured on the sub-instance $I$ when restricted to the induced subgraph given by the link $X_\tau$.}: we show that $\Phi_{\beta,\nu}^{I_\tau}(\mu)$ (\Cref{def:sp-induced}) is large. By the rounding theorem (Theorem~\ref{thm:round}), we then conclude that the expected value of the Condition\&Round algorithm, when performed on $X_\tau$ (CR-$val_\mu(X_\tau)$) must be high. 
\end{enumerate}

This proof structure is similar to the analogous Lemma 6.9 of~\cite{bafna2020playing}. Most of our technical work goes into proving points 1 and 3 above, and point 2 turns out to be straightforward given BBKSS. 

We now turn to proving our structure theorem in the low-degree SoS proof system. We will reprove the lemmas for expansion for pseudorandom sets in HD-walks proved in \Cref{sec:pseudorandomness} and \Cref{sec:expansion}, but this time carefully making sure that they are in the low-degree SoS proof system.

\begin{theorem}[SoS Structure Theorem for HD-Walks]\label{thm:structure-SoS}
Let $(X,\Pi)$ be a two-sided $\gamma$-local-spectral expander and $M$ be a $k$-dimensional, complete HD-walk on $X$ with $\gamma \leq w(M)^{-1}h(M)^{-2}2^{-\Omega(k)}$. Then for any  $f \in C_k$ and any $0 \leq \ell \leq k/2$, the expansion\footnote{This notion of expansion varies slightly from our previous definition, coinciding (up to normalization) for Boolean functions.} of $f$ with respect to $M$ is large:
\begin{align*}
\vdash_2\,\,
\langle f, (I-M)f \rangle \geq (1-\lambda_{\ell+1})\left ((1-c_1\gamma)(\E{}{f} + B(f)) - {k \choose \ell}\langle D^\ell_kf, D^\ell_kf \rangle \right ) ,
\end{align*}
where $B(f) = \mathbb{E}[f^2-f]$ represents the Booleanity constraints and $c_1 \leq 2^{O(k)}$.
\end{theorem}
The proof of \Cref{thm:structure-SoS} relies on a number of SoS variants of properties of the HD-Level-Set Decomposition used in the previous sections, namely approximate orthogonality and the relation between $\norm{f_i}$ and $\norm{g_i}$.

\begin{lemma}\label{lemma:f-SoS}
Let $(X,\Pi)$ be a two-sided $\gamma$-local-spectral expander with $\gamma \leq 2^{-\Omega(k)}$ and $f \in C_k$. Then for $f_i=U^k_ig_i \in V_k^i$:
\[
\vdash_2 \langle f_i, f_i \rangle \in \left(\frac{1}{{k \choose i}} \pm c_1\gamma\right) \langle g_i,g_i \rangle,
\]
where $c_1 \leq O(k^2)$. Further, weak variants of approximate orthogonality also have degree 2 proofs:
\[
\vdash_2 \langle f_i, f_i \rangle \leq (1+c_2\gamma) \langle f,f \rangle,
\]
and
\[
\vdash_2 \langle f, f_i \rangle \geq -c_3\gamma \langle f,f \rangle,
\]
where $c_2,c_3 \leq 2^{O(k)}$
\end{lemma}
The proof of \Cref{lemma:f-SoS} follows from combining arguments in \cite{dikstein2018boolean} with standard SoS tricks. For completeness, we give the proof in \Cref{app:SoS}. With \Cref{lemma:f-SoS} in hand, \Cref{thm:structure-SoS} follows similarly to \Cref{thm:body-local-spec-proj} and \Cref{thm:hdx-expansion}.
\begin{proof}[Proof of \Cref{thm:structure-SoS}]
As in the proof of Theorem~\ref{thm:hdx-expansion}, it is sufficient to bound the weight of $f$ onto low levels of the HD-Level-Set decomposition to prove that the expansion of $f$ is small. First, note that given the function $f$, the low-level decomposition functions $f_i$'s are explicit linear functions of the coordinates of $f$ \cite{dikstein2018boolean}. We will thus show
that Lemma~\ref{lem:low-level-weight} has an SoS proof, that is
the following relation between $\norm{D^k_\ell f}^2$ and $f$'s weight holds:
\begin{equation}\label{eq:D-f-SoS}
\vdash_2 ~~ \sum\limits_{j=0}^\ell \langle f, f_j \rangle \leq 
{k \choose \ell}\langle D^\ell_k f, D^\ell_k f \rangle + 2^{O(k)}\gamma \langle f, f \rangle.
\end{equation}

Before proving \Cref{eq:D-f-SoS}, we check that if the above inequality holds then the theorem statement follows. To see this, first recall from \Cref{sec:hdx-spectra} that $Mf_i = \lambda_i f_i + \Gamma g_i$ where $\Gamma$ is a matrix with spectral norm $\leq w(M)h(M)^22^{O(k)}\gamma$ which we call $C$ for convenience. It will be useful to have an SoS upper bound on $\uinner{f}{\Gamma g_i}$, which we will use multiple times throughout this proof. First by an SoS version of Cauchy-Schwarz we get that for all real constants $\zeta > 0$:
\[\vdash_2 \uinner{f}{\Gamma g_i} \leq \frac{\zeta}{2}\uinner{f}{f} + \frac{1}{2\zeta}\uinner{\Gamma g_i}{\Gamma g_i}. \]

Now substituting $\zeta = C$ and simplifying using the spectral norm bounds on $\Gamma$ and \Cref{lemma:f-SoS} we get:
\begin{align*}
\uinner{f}{\Gamma g_i} &\leq \frac{C}{2}\uinner{f}{f} + \frac{1}{2C}\uinner{\Gamma g_i}{\Gamma g_i} \\
&\leq \frac{C}{2}\uinner{f}{f} + \frac{C}{2}\uinner{g_i}{g_i} \\
&\leq \frac{C}{2}\uinner{f}{f} + \frac{C}{2}\left({k \choose i} + O(\gamma)\right)\uinner{f_i}{f_i}, \\
&\leq c_1 \gamma \uinner{f}{f},
\end{align*}
where all inequalities are degree $2$ SoS inequalities and $c_1 \leq w(M) h(M)^2 2^{O(k)}$. With this in hand and assuming \Cref{eq:D-f-SoS}, the result follows from expanding out $\langle f, (I-M)f \rangle$ and applying \Cref{lemma:f-SoS}: 
\begin{align*}
\langle f,(I-M)f \rangle &= \langle f, f \rangle - \sum\limits_{i=0}^k \lambda_i \langle f, f_i \rangle - \sum\limits_{i=0}^k \langle f, \Gamma g_i \rangle \\
&\geq (1-c_1\gamma)\langle f, f \rangle - \sum\limits_{i=0}^{\ell} \lambda_i \langle f, f_i \rangle - \sum\limits_{i=\ell+1}^{k} \lambda_i \langle f, f_i \rangle\\
&\geq (1-c_2\gamma)\langle f, f \rangle - \sum\limits_{i=0}^{\ell} \langle f, f_i \rangle - \lambda_{\ell+1}\sum\limits_{i=\ell+1}^{k} \langle f, f_i \rangle\\
&= (1-c_2\gamma)\langle f, f \rangle - \sum\limits_{i=0}^{\ell} \langle f, f_i \rangle - \lambda_{\ell+1}\langle f, f \rangle + \lambda_{\ell+1}\sum\limits_{i=0}^\ell \langle f, f_i \rangle\\
&= (1-\lambda_{\ell+1})\left ( \left(1-c_3\gamma \right)\langle f,f \rangle - \sum\limits_{i=0}^{\ell} \langle f, f_i \rangle \right )\\
&\geq (1-\lambda_{\ell+1})\left ( \left(1-c_4\gamma\right)\langle f,f \rangle - {k \choose  \ell}\langle D^\ell_k, D^\ell_k \rangle \right )\\
&= (1-\lambda_{\ell+1})\left(\left(1-c_4\gamma \right)\mathbb{E}[f] + (1-c_4\gamma)B(f) - {k \choose  \ell}\langle D^\ell_k, D^\ell_k \rangle \right )
\end{align*}
where $c_1,c_2,c_3,c_4 \leq w(M)h(M)^22^{O(k)}$, and $B(f) = \mathbb{E}[f^2-f]$.

It remains to prove \Cref{eq:D-f-SoS}. This follows from a similar modification of Lemma~\ref{lem:low-level-weight}. Notice that by the adjointness of $D$ and $U$, it is enough to analyze the walk $U_\ell^kD_\ell^k$:
\[
\langle D_\ell^kf, D_\ell^k f \rangle = \langle f, U_\ell^kD_\ell^k f \rangle 
\]
Using \Cref{prop:hdx-pure-eig-vals} and the assumption $\ell \leq k/2$, we can decompose the righthand side as:
\[
{k \choose \ell}\langle f, U_\ell^kD_\ell^k f \rangle =  \sum\limits_{j=0}^\ell {k - j \choose \ell-j} \langle f, f_j \rangle + \sum\limits_{j=0}^k \langle f, \Gamma g_j \rangle
\]
where $\norm{\Gamma} \leq w(M)h(M)^22^{O(k)}\gamma := C$. Noting that ${k-j \choose \ell-j}$ is at least $1$ for $0 \leq j \leq \ell$, we can apply the upper bound on $\uinner{f}{\Gamma g_j}$ proved above and \Cref{lemma:f-SoS} to get:
\begin{align*}
    {k \choose \ell}\langle D_\ell^kf, D_\ell^k f \rangle &\geq \sum\limits_{j=0}^\ell \langle f,f_j \rangle - c_5\gamma \langle f,f\rangle - \sum\limits_{j=0}^k \langle f, \Gamma g_j \rangle\\
    &\geq \sum\limits_{j=0}^\ell \langle f,f_j \rangle - c_6\gamma\langle f,f \rangle \\
\end{align*}
where all the constants $c_i$'s are less than $w(M)h(M)^22^{O(k)}$ and all inequalities are degree-2 SoS.
\end{proof}
It is worth giving a brief qualitative comparison of this result to a similar version for the Johnson graphs in \cite{bafna2020playing}. In particular, \Cref{thm:structure-SoS} not only gives a tighter bound (by a factor of exp($r$)), but perhaps more importantly shows how viewing the problem from the framework of high dimensional expansion demystifies the original Fourier analytic proof. This understanding allows us to extend the structural result well beyond the Johnson graphs to all HD-walks. In fact, this result also holds for the more general class of expanding posets \cite{dikstein2018boolean} (albeit with different parameters). We leave their discussion to future work.

Given this structure theorem, we use the BBKSS' framework to show that the global potential (Definition~\ref{def:sp-global}) on a link is high. This follows because of the properties of the potential function and the pseudodistribution over $(1-\eta)$-satisfying assignments. In general, the potential function $\Phi_{\beta,\nu}^I(\mu)$ turns out to be an average taken over non-expanding partitions of the graph. Using our structure theorem, we get that every non-expanding partition must have a large variance across links, and therefore there must exist a link $X_\tau$ such that the partition is large even when restricted to $X_\tau$. As a result, there must be
be a link $X_\tau$ where the partitions are large on average too, which corresponds to a quantity called the global potential restricted to $X_\tau$, denoted by $\Phi_{\beta,\nu}^I(\mu|_\tau)$, being large. The proof is essentially the same as \cite[Lemma 6.9]{bafna2020playing}, but we include it in \Cref{sec:app-ug-proofs} for completeness. 

\begin{proposition}\label{prop:high-global}
Let $M$ be a $k$-dimensional complete HD-walk on a $d$-dimensional, two-sided $\gamma$-local-spectral expander satisfying $\gamma w(M)2^{O(h(M)+k)} \leq 1/2^k < \eta < 0.02$ and $d>k$, and $I$ be an affine unique games instance over $M$ with value at least $1-\eta$. Let $r(\eta)=R_{1-16\eta}(M)-1$, and assume $r \leq k/2$. Let $\beta = 19\eta$ and $\nu = \frac{\eta}{56{k \choose r}}$. Then given a degree-$\tO(1/\nu)$ pseudodistribution $\mu$ satisfying the axioms $\mathcal{A}_I$, there exists an $r$-link $X_\tau$ such that the global potential restricted to $X_\tau$ is large: 
\[
\Phi_{\beta,\nu}^{I}(\mu|_\tau) \geq \frac{1}{4{k \choose r}}.
\]
\end{proposition}

In the next lemma, we will relate the global potential to the potential induced on $X_\tau$. In particular, we'll show that since the global potential on $X_\tau$ is high (\Cref{prop:high-global}), the potential induced on $X_\tau$, $\Phi_{\beta,\nu}^{I_\tau}(\mu)$, must also be high. Since this is just the usual potential function when applied to the sub instance $I|_\tau$ corresponding to the subgraph induced by $X_\tau$ (Definition~\ref{def:sp-induced}), we can then apply the rounding theorem of BBKSS~\Cref{thm:round} to surmise that the CR-value of $X_\tau$ is high.

\begin{lemma}\label{lem:ug-local-to-global}
Assume the conditions of~\Cref{prop:high-global} hold. Let $X_\tau$ be an $r$-link for $r =R_{1-16\eta}(M)-1$ with high global potential: $\Phi_{\beta,\nu}^{I}(\mu|_\tau) \geq \frac{1}{4{k \choose r}}$, for $\beta = 19\eta$ and $\nu = \frac{\eta}{56{k \choose r}}$. Then the potential induced on $X_\tau$ is also high:
    \[\Phi_{\eta,\nu}^{I|_\tau}(\mu) \geq \frac{1}{8{k \choose r}}.\]
\end{lemma}
We prove \Cref{lem:ug-local-to-global} in \Cref{sec:app-ug-proofs}, but it's worth pausing to discuss the proof, especially in how it differs from the analogous result in BBKSS \cite[Claim 6.11]{bafna2020playing}. Recall that the main idea behind \Cref{lem:ug-local-to-global} is to relate the global potential on $X_\tau$ (bounded in \Cref{prop:high-global}) to the potential induced on $X_\tau$ itself (which implies high CR-Value by \Cref{thm:round}). In a bit more detail, the global potential on a link depends on the ``value'' of the vertices in the link (which measures the value of an assignment at the vertex), but is measured with respect to its neighbors across the \emph{entire graph}. On the other hand, the induced potential on the link depends on the value of vertices when measured only with respect to the edges \emph{inside} the link (see Definitions~\ref{def:sp-global},\ref{def:sp-induced} for exact details). It is possible to relate the two potentials by observing that the internal-value of a vertex (value with respect to only the neighbors inside the link) can decrease by at most an additive factor equal to its edge-expansion inside the link (i.e. the fraction of edges incident on the vertex that leave the set). By leveraging machinery developed in \Cref{sec:hdx-spectra} and some additional properties, we show not only that the expansion of links is small (by \Cref{thm:local-vs-global}), but in fact that this holds approximately vertex by vertex: for most vertices in the link, only an $O(\eta)$-fraction of their edges are outgoing. We prove the following claim formally in \Cref{sec:app-ug-proofs}: 

\begin{lemma}\label{lem:vert-edge-exp}
Let $M$ be a $k$-dimensional HD-walk on $d$-dimensional two-sided $\gamma$-local-spectral expander satisfying $\gamma \leq w(M)^{-1}2^{-\Omega(h(M)+k)}$ and $d>k$. Then for every $i$-link $X_\tau$, the deviation of the random variable $\phi_{X_\tau}(v)$ ($v \sim X_\tau$) is small:
\[\e_{v \sim X_\tau}[|\phi_{X_\tau}(v) - \phi(X_\tau)]|] \leq \frac{1}{2^{11k}},\]
where $\phi_{X_\tau}(v)$ denotes the fraction of edges incident on $v \in X_\tau$ that leave $X_\tau$.
\end{lemma}

BBKSS use an exact version of this statement for the Johnson (that expansion holds vertex-by-vertex) to prove their analogous version of \Cref{lem:ug-local-to-global}. We relax the conditions required for their proof and show that the approximate statement state above is sufficient in \Cref{sec:app-ug-proofs}.

Finally, we complete the section by using \Cref{prop:high-global} and \Cref{lem:ug-local-to-global} to prove \Cref{prop:restricted-con-round} (and thereby \Cref{thm:unique-games} as well). We restate the proposition here for convenience:

\begin{proposition}[Restatement of \Cref{prop:restricted-con-round}]
Let $M$ be a $k$-dimensional complete HD-walk on a $d$-dimensional, two-sided $\gamma$-spectral expander with $\gamma \leq w(M)^{-1}2^{-\Omega(h(M)+k)}$ and $d>k$, and $I$ be an affine unique games instance over $M$ with value at least $1-\eta$ where $1/2^k \leq\eta < .02$. Let $r=r(\eta)=R_{1-16\eta}(M)-1$, and assume $r \leq k/2$. Then given a degree-$\tO\left (\frac{1}{\eta}{k \choose r} \right )$ pseudodistribution $\mu$ satisfying the axioms $\mathcal{A}_I$, we can find in time $|X(k)|^4$ an $r$-link $X_\tau$ with CR-$\val_\mu(X_\tau) \geq \Omega\left(\frac{\eta}{{k \choose r}} \right)$.
\end{proposition}
\begin{proof}
Let $\beta = 19\eta$ and $\nu = \frac{\eta}{56{k \choose r}}$ as in \Cref{prop:high-global}.
By \Cref{prop:high-global}, there exists an $r$-link $X_\tau$ with high global potential:
\[
\Phi_{\beta,\nu}^{I}(\mu|_\tau) \geq \frac{1}{4{k \choose r}}.
\]
By \Cref{lem:ug-local-to-global}, this implies that $X_\tau$ also has high induced potential:
\[
\Phi_{\eta,\nu}^{I|_\tau}(\mu) \geq \frac{1}{8{k \choose r}}.
\]
As mentioned previously, this is just the usual potential function on the instance induced by $X_\tau$ (Definition~\ref{def:sp-induced}), so we may apply the rounding theorem of BBKSS~\Cref{thm:round} to bound $X_\tau$'s CR-Value. In particular, since we have set $\nu$ such that $(\eta-\nu) \geq \Omega(\eta)$ and $(\Phi_{\eta,\nu}^{I|_\tau}(\mu)-\nu) \geq \Omega
\left(\frac{1}{{k \choose r}}\right)$, \Cref{thm:round} implies that rounding $\mu$, which is a degree $\tO(1/\nu) = \tO\left(\frac{1}{\eta}{k \choose r}\right)$ pseudodistribution as required, would give an assignment with large -value on $X_\tau$:
\[
\text{CR-}\val_\mu(X_\tau) \geq (\eta - \nu)(\Phi_{\eta,\nu}^{I|_\tau}(\mu)-\nu) \geq \Omega
\left(\frac{\eta}{{k \choose r}}\right).
\]
Now that we have shown that there exists an $r$-link $X_\tau$ with large CR-value, it remains to show that we can efficiently find such a link. We can do a brute-force enumeration over all $r$-links in time $|X(r)|$, compute every CR-value in time  $|X(r)|^3$, and therefore in time $|X(r)|^4$ find a link with large CR-value. Since $|X(k)| \geq |X(r)|$ is a standard consequence of $X$ being a local-spectral expander \cite{dikstein2018boolean} the proposition follows. 
\end{proof}

\section{Acknowledgements}
This work stemmed in part from collaboration at the Simons Institute of Theory of Computing, Berkeley during the 2019 summer cluster: ``Error-Correcting Codes and High-Dimensional Expansion''. We thank the Simons institute for their hospitality and the organizers of the cluster for creating such an opportunity. Further, the authors would like to thank Sankeerth Rao and Yotam Dikstein for useful discussions in the early stages of this work, and Ella Sharakanski for discussion regarding HD-walk decompositions. We thank Madhu Sudan and Sam Hopkins for helpful comments on initial drafts of this work. Finally, we additionally owe many thanks to Sam Hopkins for discussions on the Sum of Squares framework, and to Vedat Alev for his insights on the spectral structure of higher order random walks.

\bibliographystyle{amsalpha}  
\bibliography{references} 
\appendix
\section{Proof of \Cref{lemma:body-DU-UD}}\label{app:decomp}
In this section, we prove a strengthening of the main technical lemma of DDFH Section 8 \cite[Claim 8.8]{dikstein2018boolean}, which allows for better control of error propagation.
\begin{lemma}[Strengthened Claim 8.8  \cite{dikstein2018boolean}]\label{lemma:DU-UD}
Let $(X,\Pi)$ be a $d$-dimensional two-sided $\gamma$-local-spectral expander. Then for all $j< k < d$:
\[
D_{k+1}U^{k+1}_{k-j} - \frac{j+1}{k+1}U^{k}_{k-j} - \frac{k-j}{k+1} U^{k}_{k-j-1}D_{k-j} =  \sum \limits_{i=-1}^{j-1} \frac{k-i}{k+1}U^{k}_{k-1-i}\Gamma_iU^{k-1-i}_{k-j}
\]
where $\norm{\Gamma_i} \leq \gamma$.
\end{lemma}
\begin{proof}
The proof follows by a simple induction. The base cases, $j=0$ and $k<d$, follow immediately from \Cref{eq:hdx-eposet}. For the inductive step, consider:
\begin{align*}
D^{k+1}U_{k-(j+1)}^{k+1} = &\left(D^{k+1}U_{k-j}^{k+1} - \frac{j+1}{k+1}U^{k}_{k-j} - \frac{k-j}{k+1}U^{k}_{k-j}D_{k-j}\right) U_{k-j-1}\\
& + \frac{j+1}{k+1}U^{k}_{k-j-1} + \frac{k-j}{k+1}U^{k}_{k-j-1}D_{k-j}U_{k-j-1}
\end{align*}
By the inductive hypothesis, the first term on the RHS may be written as:
\[
\left(D^{k+1}U_{k-j}^{k+1} - \frac{j+1}{k+1}U^{k}_{k-j} - \frac{k-j}{k+1}U^{k}_{k-j}D_{k-j}\right)U_{k-j-1} = \sum\limits_{i=-1}^{j-1}\frac{k-i}{k+1} U^{k}_{k-i-1}\Gamma_i U^{k-1-i}_{k-j-1},
\]
where $\norm{\Gamma_i} \leq 
\gamma$. For the latter term, consider flipping $DU$ and $UD$. By \Cref{eq:hdx-eposet} we have:
\begin{align*}
    \frac{k-j}{k+1}U^{k}_{k-j-1}D_{k-j}U_{k-j-1} = U^{k}_{k-j-1}\left(\frac{1}{k+1}I + \frac{k-j-1}{k+1}U_{k-j-2}D_{k-j-1} + \frac{k-j}{k+1}\Gamma_j\right),
\end{align*}
for some $\Gamma_j$ satisfying $\norm{\Gamma_j}\leq\gamma$.
Combining these observations yields the desired result:
\begin{align*}
    &D^{k+1}U_{k-(j+1)}^{k+1} - \frac{(j+1)+1}{k+1}U^{k}_{k-{j+1}} + \frac{k-(j+1)}{k+1}U^{k}_{k-(j+1)-1}D_{k-{j+1}}\\
    &=D^{k+1}U_{k-(j+1)}^{k+1} - \frac{j+1}{k+1}U^{k}_{k-j-1} -U^{k}_{k-j-1}\left(\frac{1}{k+1}I - \frac{k-j-1}{k+1}U_{k-j-2}D_{k-j-1}\right)\\
    = & \left (\sum\limits_{i=-1}^{j-1}\frac{k-i}{k+1} U^{k}_{k-i-1}\Gamma_i U^{k-1-i}_{k-j-1}\right)+ \frac{k-j}{k+1}U^{k}_{k-(j+1)}\Gamma_j\\
    =&  \sum \limits_{i=-1}^{j}\frac{k-i}{k+1} U^{k}_{k-1-i}\Gamma_iU^{k-1-i}_{k-(j+1)}.
\end{align*}
\end{proof}
We now show how to use this strengthened result to prove tighter bounds on the quadratic form $\langle f, N^j_k f \rangle$ which implies a stronger version of \Cref{lemma:hdx-fvsg-body} as an immediate corollary. This improvement mainly matters in the regime where $\gamma \leq 2^{-ck}$ for $c$ a small constant.

\begin{proposition}\label{prop:canon-spectra-HDX-better}
Let $(X,\Pi)$ be a $d$-dimensional $\gamma$-local-spectral expander with $\gamma$ satisfying $\gamma \leq 2^{-\Omega(k+j)}$, $k+j \leq d$, and $f_\ell \in V_k^\ell$. Then:
\[
\langle f_\ell, N^j_k f_\ell \rangle  = \frac{\binom{k}{\ell}}{\binom{k+j}{\ell}}\left (1  \pm \frac{j(j+2k+2\ell+3)}{4}\gamma \pm c_3(k,j,\ell)\gamma^2 \right)\langle f_\ell,f_\ell \rangle
\]
where $c_3(k,j,\ell)=O((k+j)^3{k+j \choose \ell})$.
\end{proposition}
\begin{proof}
We proceed by induction on $j$. We will prove a slightly stronger statement for the base-case $j=1$:
\[
\langle f_\ell, D_{k+1}U_k f_\ell \rangle = \left (\frac{k+1-\ell}{k+1} \pm \frac{(k-\ell+1)(k+\ell+2)}{2(k+1)}\gamma \pm c_2(k,\ell)\gamma^2 \right )\langle f_\ell, f_\ell \rangle,
\]
where $c_2(k,\ell) = O(k^3{k \choose \ell})$. Recall that $f_\ell$ may be expressed as $U^{k}_\ell g_\ell$, for $g_\ell \in H^\ell$. For notational convenience, we write $f^i_\ell = U^{i}_{\ell}g_\ell$. Then we may expand the inner product based on \Cref{lemma:body-DU-UD}, and simplify based on applying the naive bounds on $N^i_k$ given by \Cref{prop:hdx-eig-vals}:
\begin{align*}
    \langle f_\ell, D_{k+1}U_k f_\ell \rangle &= \langle f_\ell, D_{k+1}U^{k+1}_\ell g_\ell \rangle\\
    &= \frac{k-\ell+1}{k+1}\langle f_\ell, f_\ell \rangle + \sum \limits_{i=-1}^{k-\ell-1} \langle f_\ell, \frac{k-i}{k+1}U^{k}_{k-1-i}\Gamma_iU^{k-1-i}_{\ell}g_\ell \rangle\\
    &= \frac{k-\ell+1}{k+1}\langle f_\ell, f_\ell \rangle + \sum \limits_{i=-1}^{k-\ell-1} \frac{k-i}{k+1}\langle N_{k-i-1}^{i+1}f_\ell^{k-1-i}, \Gamma_if_\ell^{k-1-i} \rangle\\
    &= \frac{k-\ell+1}{k+1}\langle f_\ell, f_\ell \rangle + \sum \limits_{i=-1}^{k-\ell-1}\frac{k-i}{k+1} \frac{{k-i-1 \choose \ell}}{{k \choose \ell}}\langle f_\ell^{k-1-i}, \Gamma_if_\ell^{k-1-i} \rangle + \sum \limits_{i=-1}^{k-\ell-1} \frac{k-i}{k+1}\langle h_{i}, \Gamma_if_\ell^{k-1-i} \rangle
\end{align*}
where $\norm{h_i} \leq \gamma(k-\ell)(i+1)\norm{g_\ell}$. We now apply Cauchy-Schwarz, and \Cref{lemma:hdx-fvsg-body} to collect terms in $\langle f_\ell,f_\ell \rangle$:
\begin{align*}
    \langle f_\ell, D_{k+1}U_k f_\ell \rangle &=\frac{k-\ell+1}{k+1}\langle f_\ell, f_\ell \rangle \pm \gamma\sum \limits_{i=-1}^{k-\ell-1} \frac{k-i}{k+1}\frac{{k-i-1 \choose \ell}}{{k \choose \ell}}\langle f_\ell^{k-1-i}, f_\ell^{k-1-i} \rangle \pm a_1(k,\ell)\gamma^2 \langle g_\ell, g_\ell \rangle\\
    &=\frac{k-\ell+1}{k+1}\langle f_\ell, f_\ell \rangle \pm \gamma\sum \limits_{i=-1}^{k-j-1} \frac{k-i}{k+1}\frac{1}{{k \choose j}}\langle g_\ell, g_\ell \rangle
    \pm a_2(k,\ell)\gamma^2 \langle g_\ell, g_\ell \rangle\\
    &=\frac{k-\ell+1}{k+1}\langle f_\ell, f_\ell \rangle \pm \gamma\sum \limits_{i=-1}^{k-j-1} \frac{k-i}{k+1}\frac{\langle f_\ell, f_\ell \rangle}{(1 - c_1(k,\ell)\gamma)}
    \pm a_2(k,\ell)\gamma^2 \frac{\langle f_\ell, f_\ell \rangle}{(1 - c_1(k,\ell)\gamma)}\\
    &=\frac{k-\ell+1}{k+1}\left (1 \pm \frac{(k+\ell+2)}{2}\gamma
    \pm a_3(k,\ell)\gamma^2 \right )\langle f_\ell, f_\ell \rangle
\end{align*}
where the final step comes from a Taylor expansion assuming $\gamma$ sufficiently small, and $a_3(k,\ell)=O(k^3{k \choose \ell})$.

The inductive step follows from noting that the canonical walk essentially acts like a product of upper walks from lower levels in the following sense:
\begin{align*}
\langle f_\ell, N^j_k f_\ell \rangle &= \langle U^{k+j-1}_{k} f_\ell, D_{k+j}U^{k+j}_{k} f_\ell \rangle \\
&= \langle U^{k+j-1}_{k} f_\ell,N^1_{k+j-1}(U^{k+j-1}_{k} f_\ell) \rangle.
\end{align*}
Thus by the base-case and inductive hypothesis we get:
\begin{align*}
\langle f_\ell, N^j_k f_\ell \rangle &= \langle U^{k+j-1}_{k} f_\ell,N^1_{k+j-1}(U^{k+j-1}_{k} f_\ell) \rangle \\
&= \left (\frac{k+j-\ell}{k+j} \pm \frac{(k+j-\ell)(k+j+\ell+1)}{2(k+j)}\gamma \pm c_2(k+j-1,\ell)\gamma^2 \right)\langle f_\ell,N_k^{j-1} f_\ell \rangle\\
&= \frac{\binom{k}{\ell}}{\binom{k+j}{\ell}} \left (1 \pm \frac{(k+j+\ell+1)}{2}\gamma \pm \frac{k+j}{k+j-\ell}c_2(k+j-1,\ell)\gamma^2 \right)\\
&\quad\quad\quad~\cdot \left (1  \pm \frac{(j-1)(j+2k+2\ell+2)}{4}\gamma \pm c_3(k,j-1,\ell)\gamma^2 \right)\langle f_\ell,f_\ell \rangle\\
&= \frac{\binom{k}{\ell}}{\binom{k+j}{\ell}}\left (1  \pm \frac{j(j+2k+2\ell+3)}{4}\gamma \pm c_3(k,j,\ell)\gamma^2 \right)\langle f_\ell,f_\ell \rangle,
\end{align*}
\end{proof}
Notice that this immediately implies a stronger version of \Cref{lemma:hdx-fvsg-body}, since $\langle U^k_\ell g _\ell, U^k_\ell g_\ell \rangle = \langle N_{\ell}^{k-\ell}g_\ell, g_\ell \rangle$. Finally, we conjecture that a stronger result is true, and the error dependence on $\gamma$ should in fact be $\text{exp}(-\text{poly}(k)\gamma)$. Proving this would require a more careful and involved analysis of how the error term propogates.
\section{Orthogonality and the HD-Level-Set Decomposition}\label{App:ortho}
In this section we discuss in a bit more depth the error in \cite[Theorem 5.10]{kaufman2020high}, and further show by direct counter-example that its implication \cite{kaufman2020chernoff} that the HD-Level-Set is orthogonal does not hold. In \cite{kaufman2020high}, Kaufman and Oppenheim analyze an approximate eidgendecomposition of the upper walk $N^1_k$ for two-sided local-spectral expanders. They prove a specialized version of \Cref{thm:approx-ortho} for this case, and in particular that for sufficiently strong two-sided local-spectral expanders, the spectra of $N^1_k$ is divided into strips concentrated around the approximate eigenvalues of their decomposition. They call the span of each strip $W^i$, and note that the $W^i$ form an orthogonal decomposition of the space. Let $V^i$ be the space in the original approximate eigendecomposition corresonding to strip $W^i$. Kaufman and Oppenheim claim in \cite[Theorem 5.10]{kaufman2020high} that the $W^i$ are closely related to the original approximate decomposition in the following sense:
\[
\forall \phi \in C_k: \norm{P_{W^i}\phi} \leq c\norm{P_{V^i}\phi}
\]
for some constant $c>0$, where $P_{W^i}$ and $P_{V^i}$ are projection operators. Unfortunately, this relation cannot hold, as it implies \cite{kaufman2020chernoff} that the HD-Level-Set Decomposition is orthogonal for sufficiently strong two-sided local-spectral expanders, which we will show below is false by direct example. In slightly greater detail, the issue in the argument is the following. The authors show that for any $j \neq i$:
\[
\norm{P_{W^j}P_{V^i}} \leq c',
\]
for some small constant $c'$, and then claim that this fact implies for any $\phi \in C_k$:
\[
\norm{P_{W^j}P_{V^i}\phi} \leq c'\norm{P_{V^j}\phi}.
\]
Unfortunately, this is not true---the righthand side should read $P_{V^i}$ rather than $P_{V^j}$ for the relation to hold, but this makes it impossible to compare $P_{W^i}\phi$ solely to $P_{V^i}\phi$. 

We now move to showing that for any $\gamma > 0$, there exists a two-sided $\gamma$-local-spectral expander such that the HD-Level-Set Decomposition is not orthogonal, which implies \cite[Theorem 5.10]{kaufman2020high} cannot hold by arguments of \cite{kaufman2020chernoff}.
\begin{proposition}
For any $\gamma > 0$, there exists a two-sided $\gamma$-local-spectral expander such that the HD-Level-Set Decomposition is not orthogonal.
\end{proposition}
\begin{proof}
Our construction is based off of a slight modification of the complete complex $J(n,3)$. In particular, we consider the uniform distribution $\Pi$ over triangles $X={[n] \choose 3} \setminus (123)$. It is not hard to see through direct computation that $(X,\Pi)$ is a two-sided $O(1/n)$-local-spectral expander. Recall that the link of $1$, $U^3_1\mathbbm{1}_{1}$, lies in $V_3^0 \oplus V_3^1$. Our goal is to prove the existence of a function $f=Ug \in V_3^2$ such that the inner product:
\begin{align}\label{eq:inner-prod}
\langle U^3_1\mathbbm{1}_{1}, f \rangle \propto \sum\limits_{(1xy) \in X} g(1x) + g(1y) + g(xy)
\end{align}
is non-zero. To do this, we first simplify the above expression assuming $g \in \text{Ker}(D_2),$ which we recall implies the following relations:
\[
    \forall y \in [n]: \sum\limits_{(xy) \in X(2)} \Pi_2(xy)g(xy) = 0.
\]
In particular, summing over all $y \in [n]$ gives
\[
    \sum\limits_{(xy) \in X(2)} \Pi_{2}(xy)g(xy) = 0.
\]
Notice further that by definition of $\Pi_2$, we have $\Pi_2(12)=\Pi_2(13)=\Pi_2(23)=\frac{n-3}{3{n \choose 3}-3}$, and otherwise $\Pi_2(xy)=\frac{n-2}{3{n \choose 3}-3}$. We then may write:
\begin{align*}
    \sum\limits_{\underset{x \notin [3]}{(1x) \in X(2):}} g(1x) = -\frac{n-3}{n-2}\left(g(12) + g(13) \right),\\
    \sum\limits_{\underset{(xy) \notin [3] \times [3]}{(xy) \in X(2):}} g(xy) = -\frac{n-3}{n-2}(g(12) + g(13) + g(23)).
\end{align*}
Plugging this into \Cref{eq:inner-prod}, the inner product drastically simplifies to depend only on $g(23)$. To see this, we separate the inner product into two terms and deal with each separately:
\begin{align*}
    \sum\limits_{(1xy) \in X} g(1x) + g(1y) + g(xy) & =\left(\sum\limits_{(1xy) \in X} g(1x) + g(1y)\right) + \sum\limits_{(1xy) \in X} g(xy)\\
\end{align*}
We start with the former. Notice that each face $(1z)$ in this term is counted exactly the number of times it appears in a triangle in $X$, and further that this is exactly how $\Pi_2$ is defined. Thus we have:
\[
\left(\sum\limits_{(1xy) \in X} g(1x) + g(1y)\right) \propto \sum\limits_{(1x) \in X(2)} \Pi_2(1x)g(1x)=0.
\]
It is left to analyze the latter term. Since $(123)$ is not in our complex, we may write:
\begin{align*}
    \sum\limits_{(1xy) \in X} g(xy) &= \left(\sum\limits_{\underset{(xy) \notin [3]\times[3]}{(xy) \in X(2):}} g(xy)\right) - \left(\sum\limits_{\underset{x \notin [3]}{(1x) \in X(2):}} g(1x) \right)\\
    &= -\frac{n-3}{n-2}(g(12) + g(13) + g(23)) + \frac{n-3}{n-2}\left(g(12) + g(13) \right)\\
    &= -\frac{n-3}{n-2}g(23).
\end{align*}
Thus it remains to show that there exists $g \in \text{Ker}(D_2)$ such that $g(23) \neq 0$. Note that the kernel of $D_2$ is exactly the space of solutions to the underdetermined linear system of equations given by $D_2g(i)=0$ for all $1 \leq i \leq n$. Thus we can check if a solution exists with $g(23) = c$ for $c \neq 0$ by ensuring that this constraint is linearly independent of the $D_2g(i)$. This can be checked through a direct but tedious computation that we leave to the reader.
\end{proof}

\section{Unique Games}

\subsection{Random-walks and Weighted graphs}\label{sec:ug-rwalk}
Unique games are defined on weighted, \textit{undirected} constraint graphs, unlike most of the walks we analyze in the previous section. However, it is not hard to see that every self-adjoint random walk corresponds to some underlying undirected graph. In particular, recall that given a weighted, undirected graph $G(V,E)$, $G$ induces a random walk on vertices where the transition probability from $x \in V$ to $y \in V$ is given by the normalized weights:
\[
P(x,y) = \frac{W(s,t)}{\sum\limits_{v \in N(s)}W(s,v)}.
\]
In fact, \textit{every} self-adjoint walk can be described as such a process.
\begin{lemma}
Let $M$ be a $k$-dimensional HD-walk on a weighted simplicial complex $(X,\Pi)$. Define $G_M$ to be the graph whose vertex set is $X(k)$ and whose edge set consists of any pair $(s,t)$ such that $M(s,t) \neq 0$. Further, endow the edges of $G_M$ with weight:
\[
W(s,t) = \Pi_k(s)M(s,t).
\]
Then $M$ is the random walk induced by $G_M$.
\end{lemma}
\begin{proof}
The crucial observation for this proof is the following implication due to $M$ being self adjoint:
\[
\forall s,t \in X(k): \Pi_k(s)M(s,t) = \Pi_k(t)M(t,s).
\]
Given this fact, let $P_G(s,t)$ denote the transition probability of the induced walk on $G$. We have:
\begin{align*}
    P_G(s,t) &= \frac{\Pi_k(s)M(s,t)}{\sum\limits_{v \in N(s)}} \Pi_k(s)M(s,t)\\
    &= \frac{M(s,t)}{\sum\limits_{v \in N(s)}} M(s,t)\\
    &= M(s,t)
\end{align*}
\end{proof}

\subsection{The Shift-Partition Potential}\label{app:potential}
The analysis of the Iterated Condition\&Round algorithm relies on analysing a quantity called the Approximate shift-partition potential which was defined in BBKSS. For completeness we include the definitions of the various potential functions used in this section. 

For analysing the Condition\&Round subroutine, define a low-degree polynomial $\Phi_{\beta,\nu}:\Sigma^V \times \Sigma^V \rightarrow [0,\infty)$ called the ``approximate shift partition potential'':

\begin{definition}[Approximate Shift-Partition Potential]\label{def:sp-apx}
For any $\nu,\beta \in (0,1)$, and two assignments $X,X' \in \Sigma^V$ define the {\em approximate shift-partition potential} to be the quantity
\[
\Phi_{\beta,\nu}(X,X') = \sum_{s \in \Sigma} \E{u}{\ind[X_u - X'_u = s] \cdot p_{\beta,\nu}(\val_u(X))}^2,
\]
for $p_{\beta,\nu}(x)$ the degree-$\tO(1/\nu)$ polynomial which SoS-certifiably approximates the indicator $\ind[x\ge\beta]$ for $x \in [0,1]$ up to an error of $\nu$. The exact definition of $p_{\beta,\nu}$ and the properties required of it are described in more detail in Section 7 of BBKSS and we omit them here.

The approximate shift-partition potential on a pseudodistribution $\mu$ of degree at least $O(\deg(\Phi_{\beta,\nu}))$ is defined as:
\[\Phi_{\beta,\nu}(\mu) = \pE_{\mu}[\Phi_{\beta,\nu}(X,X')].\]
\end{definition}

For analysing the Iterated Condition\&Round algorithm we further need to define the potential when restricted to a subgraph which in our case will be a link of the complex.

\begin{definition}[Global shift-potential restricted to links]\label{def:sp-global}
Let $I = (M, \Pi)$ be a UG instance on a complex $X$ where $M$ is a complete random walk on $X(k)$ with stationary distribution $\pi_k$. For any $\nu,\beta \in (0,1)$ and a link $X_\tau$ of the complex $X$, define the {\em approximate global shift-partition potential restricted to $X_\tau$} to be the quantity:
\[
\Phi_{\beta,\nu}(X,X')|_{\tau} = \sum_{s \in \Sigma} \e_{u \sim \pi_k|\tau}[\ind[X_u - X'_u = s] \cdot p_{\beta,\nu}(\val_u(X))]^2,
\]
where $\val_u(X) = \E{(u,v) \sim M}{\ind[X \text{ satisfies } (u,v))}$, and as before $p_{\beta,\nu}(x)$ is the degree-$\tO(1/\nu)$ polynomial that $\nu$-approximates the indicator function $\ind[x \geq \beta]$. The global potential restricted to $X_\tau$ with respect to a pseudodistribution $\mu$ will be denoted by $\Phi_{\beta,\nu}(\mu|_{\tau})$:
\[\Phi_{\beta,\nu}(\mu|_{\tau}) = \pE_{\mu}[\Phi_{\beta,\nu}(X,X')|_{\tau}].\]
\end{definition}

Note that the global shift-partition potential measures the size of the global partition inside $X_\tau$, and although the expectation is taken only over the vertices $u \in X_\tau$, $\val_u(X)$ is a function of \emph{all} the edges in $M$ that are incident on $u$, not just the edges in $X_\tau$. We will also need to consider the potential induced on a link, which is defined as applying the potential function (Definition~\ref{def:sp-apx}) to the graph induced by $X_\tau$:

\begin{definition}[Induced Shift-Potential on a Link]\label{def:sp-induced}
Let $I = (M, \Pi)$ be a UG instance on a complex $X$ where $M$ is a complete random walk on $X(k)$ with stationary distribution $\pi_k$. For any $\nu,\beta \in (0,1)$ and a link $X_\tau$ of the complex $X$, define the {\em approximate induced shift-partition potential on $X_\tau$} to be the quantity:
\[
\Phi_{\beta,\nu}^\tau(X,X') = \sum_{s \in \Sigma} \e_{u \sim \pi_k|\tau}[\ind[X_u - X'_u = s] \cdot p_{\beta,\nu}(\val_u^\tau(X))]^2,
\]
where $\val_u^\tau(X) = \E{(u,v) \sim M|_\tau}{\ind[X \text{ satisfies } (u,v))}$. The induced potential on $X_\tau$ with respect to a pseudodistribution $\mu$ will be denoted by $\Phi_{\beta,\nu}^\tau(\mu)$:
\[\Phi_{\beta,\nu}^\tau(\mu) = \pE_{\mu}[\Phi^\tau_{\beta,\nu}(X,X')].\]
\end{definition}

Note that in the above definition the value of a vertex $u \in X_\tau$ is measured only with respect to the edges incident on $u$ that lie \emph{inside} the link $X_\tau$.

\subsection{Sum of Squares and the HD-Level-Set Decomposition}\label{app:SoS}
This section is devoted to proving \Cref{lemma:f-SoS} which we separate into two parts. First, we examine the relation between $\norm{f_i}$ and $\norm{g_i}$. 
\begin{lemma}[Restated \Cref{lemma:f-SoS} (Part 1)]\label{lemma:f-vs-g}
Let $(X,\Pi)$ be a $d$-dimensional two-sided $\gamma$-local-spectral expander. Then for any $f_i=U^k_ig_i \in V_k^i$:
\[
\vdash_2 \langle f_i, f_i \rangle \in \left(\frac{1}{{k \choose i}} \pm \frac{(k-i)(k+1)}{i+2}\gamma\right) \langle g_i,g_i \rangle
\]
\end{lemma}
\begin{proof}
It is sufficient to prove the following equality:
\begin{align}\label{eq:induction-f-g}
\langle f_i,f_i \rangle = \frac{1}{{k \choose i}}\langle g_i, g_i \rangle + \langle g_i, \Gamma g_i \rangle,
\end{align}
where $\norm{\Gamma} < \frac{(k-i)(k+1)}{i+2}\gamma$. To see why, note that:
\[
\langle g_i, \Gamma g_i \rangle = \left\langle g_i, \frac{\Gamma + \Gamma^*}{2} g_i \right\rangle 
\]
where $\Gamma^*$ is the adjoint of $\Gamma$. Since $\frac{\Gamma + \Gamma^*}{2}$ is self-adjoint and $\norm{\Gamma^*}=\norm{\Gamma}$, we have both:
\begin{enumerate}
    \item $\vdash_2 \left\langle g_i, \frac{\Gamma + \Gamma^*}{2} g_i \right\rangle \leq \frac{(k-i)(k+1)}{i+2}\gamma\langle g_i, g_i \rangle$
    \item $\vdash_2 \left\langle g_i, \frac{\Gamma + \Gamma^*}{2} g_i \right\rangle \geq -\frac{(k-i)(k+1)}{i+2}\gamma\langle g_i, g_i \rangle$
\end{enumerate}
as desired. To prove \Cref{eq:induction-f-g}, we induct on $k$. The base case $k=i$ is trivial. Assume $k>i$, then we may write:
\begin{align*}
    \langle f_i, f_i \rangle = \langle U^{k-1}_ig_i, D_kU^k_ig_i \rangle
\end{align*}
In order to apply the inductive hypothesis, we recall the machinery of \cite[Claim 8.8]{dikstein2018boolean} for pushing $D_k$ through $U^k_i$. In particular:
\[
\norm{D_{k}U^{k}_{i} - \frac{k-i}{k}U^{k-1}_{i} - \frac{i}{k} U^{k-1}_{i-1}D_{i}} \leq (k-i)\gamma.
\]
Since $g_i$ lies in the kernel of $D_i$, combining this with our initial observation gives:
\begin{align*}
    \langle U^{k-1}_ig_i, D_kU^k_ig_i \rangle &= \frac{k-i}{k}\langle U^{k-1}_ig_i, U^{k-1}_{i}g_i \rangle + \langle g_i, \Gamma g_i \rangle,
\end{align*}
where $\norm{\Gamma} \leq (k-i)\gamma$. Since $U^{k-1}_ig_i \in V^{k-1}_i$. Applying the inductive hypothesis, we see that:
\begin{align*}
    \langle f_i,f_i \rangle &= \frac{k-i}{k}\left (\frac{1}{{k-1 \choose i}}\langle g_i, g_i \rangle + \langle g_i,\Gamma'g_i \rangle \right) + \langle g_i, \Gamma g_i \rangle\\
    &=\frac{1}{{k \choose i}}\langle g_i, g_i \rangle + \langle g_i,\left (\frac{k-i}{k}\Gamma' + \Gamma\right) g_i \rangle,
\end{align*}
where: 
\[
\norm{\frac{k-i}{k}\Gamma' + \Gamma} \leq \frac{(k-i)(k+1)}{i+2}\gamma
\]
by the triangle inequality and inductive hypothesis.
\end{proof}
Second, we prove that a version of approximate orthogonality of the HD-Level-Set Decomposition has a low-degree SoS proof. 
\begin{lemma}[Restated \Cref{lemma:f-SoS} (Part 2)]
let $(X,\Pi)$ be a two-sided $\gamma$-local-spectral expaner with $\gamma \leq 2^{-\Omega(k)}$ and $f \in C_k$. Then for any $f_i=U^k_ig_i \in V_k^i$ and $f_j=U^k_jg_j \in V_k^j$ we have:
\[
\vdash_2 \langle f_i, f_i \rangle \leq (1+c_1\gamma) \langle f,f \rangle,
\]
and
\[
\vdash_2 \langle f, f_i \rangle \geq -c_2\gamma \langle f,f \rangle,
\]
where $c_1,c_2 \leq 2^{O(k)}$
\end{lemma}
\begin{proof}
First, note that by \Cref{lemma:body-DU-UD} and the fact that $g_j \in \text{Ker}(D_j)$, we have that:
\[
\langle f_i,f_j \rangle = \langle g_i, \Gamma g_j \rangle
\]
where $\norm{\Gamma} \leq 2^{O(k)}$. We can bound the latter by an SoS version of Cauchy Schwarz. In particular, we have that:
\begin{align*}
\langle f_i,f_j \rangle &= 
\uinner{g_i}{\Gamma g_i}\\
&\leq \frac{\norm{\Gamma}}{2}\uinner{g_i}{g_i} + \frac{1}{2\norm{\Gamma}}\uinner{\Gamma g_j}{\Gamma g_j} \\
&\leq \frac{\norm{\Gamma}}{2}\uinner{g_i}{g_i} + \frac{\norm{\Gamma}}{2}\uinner{g_j}{g_j} \\
&\leq c\gamma(\langle f_i,f_i \rangle + \langle f_j,f_j \rangle)
\end{align*}
where all inequalities are degree 2 SoS and $c \leq 2^{O(k)}$. Further, notice that that by applying the same argument to $\langle -f_i,f_j \rangle$ we get that $\uinner{f_i}{f_j}$ is bounded above and below:
\[
\vdash_2 ~~ -c\gamma(\langle f_i,f_i\rangle + \langle f_j,f_j \rangle) \leq \langle f_i,f_j \rangle \leq c\gamma(\langle f_i,f_i\rangle + \langle f_j,f_j \rangle).
\]
We now apply this fact directly to prove the two desired inequalities. First, we have
\begin{align*}
    \langle f,f \rangle &= \sum\limits_{\ell=0}^k \langle f_\ell, f_\ell \rangle + \sum\limits_{\ell \neq m}\langle f_\ell, f_m \rangle\\
    &\geq \sum\limits_{\ell=0}^k \langle f_\ell, f_\ell \rangle - \sum\limits_{\ell=0}^k kc\gamma\langle f_\ell,f_\ell \rangle\\
    &= (1 - kc\gamma) \sum\limits_{\ell=0}^k \langle f_\ell,f_\ell \rangle\\
    & \geq (1-c_1\gamma)\langle f_i, f_i \rangle,
\end{align*}
where $c_1 \leq 2^{O(k)}$. For small enough $\gamma$, Taylor expanding $(1-c_1\gamma)^{-1}$ gives the desired result. We can now apply the first inequality to easily prove the second inequality:
\begin{align*}
    \langle f,f_i \rangle &= \langle f_i, f_i \rangle + \sum\limits_{j \neq i} \langle f_i, f_j \rangle\\
    &\geq (1-ck\gamma)\langle f_i, f_i \rangle  - \sum\limits_{j \neq i}c\gamma \langle f_j,f_j \rangle\\
    &\geq -ck\gamma(1+c_2\gamma)\langle f,f\rangle
    \geq -c_3\langle f,f \rangle
\end{align*}
for $c_2,c_3 \leq 2^{O(k)}$.
\end{proof}

\subsection{Remaining Proofs from \Cref{sec:alg-analysis}}\label{sec:app-alg-proofs}
Here we prove \Cref{claim:non-lazy-spectra}, restated for convenience.
\begin{claim}[Restated \Cref{claim:non-lazy-spectra}]
Let $M$ be a $k$-dimensional complete HD-walk on a $d$-dimensional, two-sided $\gamma$-local-spectral expander with $d>k$ and $\gamma \leq w(M)^{-1}2^{-\Omega(h(M)+k)}$ with stationary distribution $\pi$. If the total laziness of $M$ $\mathbb{E}[\id{v}^TM\id{v}]$ is at least $.1$, then the spectral gap of $M$ is at least $\Omega(1/k)$ and $\lambda_{k/2}(M) \leq .68$.
\end{claim}
\begin{proof}
To see this, note that any walk can be (approximately) decomposed by \Cref{lemma:body-DU-UD} into an affine combination of pure walks of the form $(U_{k-1}D_k)^i$, that is:
\[
M = \sum\limits_{i=0}^{h(M)} \alpha_i(U_{k-1}D_k)^i + \Gamma,
\]
where $\norm{\Gamma} \leq w(M)2^{O((M)}$ and $(U_{k-1}D_k)^0 = I$. Notice that for any $(U_{k-1}D_k)^i$ for $i>0$, the probability of returning to any given $k$-face is at most $\gamma$. This follows from the fact that the probability of returning to any particular face in the final up step (given that the previous step is at $\sigma \in X(k-1)$)) is at most $|\Pi_{\sigma,1}|_\infty \leq \gamma$. This last inequality follows from $(X,\Pi)$ being a two-sided $\gamma$-local-spectral expander satisfying $d>k$, as the link $(X_\sigma,\Pi_\sigma)$ (at level $k$) is then a $\gamma$-spectral expander, which are easily shown to satisfy this property.

Since the total laziness of $M$ is at most $.1$, by our assumption on $\gamma$ the above analysis implies that the coefficient $\alpha_0$ cannot be too large, say at most $1/5$ (since the remaining parts cannot contribute much to the laziness). We can use this fact to bound the spectral gap of $M$ by the same trick used in \Cref{prop:eig-decrease}: transferring $M$ over to the complete complex. Since the approximate eigenvalues of $M$ are independent of the underlying complex and $(UD)^i$ has a vanishing lazy component on $J(n,d)$ as $n$ goes to infinity, it must be the case that in lazy component of $M$ on the Johnson complex is also small. Since $M$ can be written as a convex combination of partial-swap walks, all of which have spectral gap at least $1/k$ (save for the identity which carries at most a $1/5$ of the weight), we get that $1-\lambda_2(M) \geq \Omega(1/k)$ as desired. The bound on $\lambda_{k/2}$ follows similarly noting that all (non-identity) partial-swap walks satisfy $\lambda_{k/2} \leq 1/2$.
\end{proof}
\subsection{Remaining Proofs from Section~\ref{sec:restricted-con-round}}\label{sec:app-ug-proofs}
In this section we will prove the claims that were stated without proof in Section~\ref{sec:restricted-con-round}. We will use the exact definition of the potential functions from~\Cref{app:potential} and also the propertites of the polynomials $p_{\beta,\nu}(Y)$ that SoS certifiably $\nu$-approximate the indicator function $\ind[Y \geq \beta]$. See Section 7 of BBKSS for a detailed overview about these polynomials and their properties.

\begin{proposition}[Restated \Cref{prop:high-global}]
Let $M$ be a $k$-dimensional complete HD-walk on a $d$-dimensional, two-sided $\gamma$-local-spectral expander satisfying $\gamma w(M)2^{O(h(M)+k)} \leq 1/2^k < \eta < 0.02$ and $d>k$, and $I$ be an affine unique games instance over $M$ with value at least $1-\eta$. Let $r(\eta)=R_{1-16\eta}(M)-1$, and assume $r \leq k/2$. Let $\beta = 19\eta$ and $\nu = \frac{\eta}{56{k \choose r}}$. Then given a degree-$\tO(1/\nu)$ pseudodistribution $\mu$ satisfying the axioms $\mathcal{A}_I$, there exists an $r$-link $X_\tau$ such that the global potential restricted to $X_\tau$ is large: 
\[
\Phi_{\beta,\nu}^{I}(\mu|_\tau) \geq \frac{1}{4{k \choose r}}.
\]
\end{proposition}
\begin{proof}[Proof of \Cref{prop:high-global}]
Before diving into the details of the proof we give a proof overview and some high level intuition. Given two UG assignments $X,X'$, BBKSS define a partition of the graph $G$ into $|\Sigma|$ parts and call this the approximate shift-partition of $G$ with respect to $X,X'$. The partition is given by the functions $f_s^{X,X'}$ for $s \in \Sigma$, where $f_s^{X,X'}(u)$ roughly equals the indicator function of whether $u$ belongs to part $s$. Concretely, we have that $f_s^{X,X'}(u) = \ind[X_u - X'_u = s] \cdot p_{\beta,\nu}(\val_u(X))$, that is, $u$ belongs to part $s$ if $X_u - X'_u = s$ and the value of $u$ measured with respect to assignment $X$ is large. Note that any vertex can only belong to one part $s$ but some vertices may not belong to any part. BBKSS prove that for any instance of affine unique games and any two UG assignments $X,X'$ that have large value, the functions $f_s^{X,X'}$ satisfy three properties: a) they include a large number of vertices in total, b) they form a non-expanding partition of $G$, that is, the fraction of edges that go across parts is small (roughly at most edges violated by $X, X'$), and finally c) the functions $f_s$ are approximately Boolean-valued on average. These properties are proved in Section 4 of their paper. They then apply the structure theorem for non-expanding sets of $G$ to conclude that there exists a subgraph $H$ where the size of the partition restricted to $H$ ($\sum_s \e_{u \sim H}[f_s^{X,X'}(u)]^2$) must be large. The crucial observation then is that the pseudoexpectation of the restricted size on $H$ is by definition the global potential on $H$. Therefore by taking pseudoexpectation over $X, X' \sim \mu$ where $\mu$ is a pseudodistribution with high value, they conclude that there exists a subgraph with large potential. We will use the same proof strategy to conclude that the global potential on a link is large, using the SoS structure theorem (\Cref{thm:structure-SoS}) along the way.

Let us elaborate on the properties (proved in BBKSS Section 4) of the approximate shift-partition given by the functions $f_s^{X,X'}$. We will drop the superscript $X,X'$ for convenience and throughout the proof the pseudoexpectation is taken with respect to $X,X' \sim \mu$. Since $\mu$ is a degree $\tO(1/\nu)$ pseudodistribution with value at least $1-\eta$ we have that:
\begin{enumerate}
    \item The approximate shift-partition includes most of the vertices in pseudoexpectation over $\mu$:
    \[
    \pE_\mu\left[\sum\limits_{s \in \Sigma}\e_u[f_s]\right] \geq 1 - \frac{\eta}{1-\beta-\nu}-\nu
    \]
    \item The approximate shift-partition is non-expanding\footnote{If the functions $f_s(u)$ were exactly Boolean-valued and further exactly partitioned the graph, then the term $\sum_s \uinner{f_s}{(I - A_G)f_s}$ is equal to the fraction of edges in $G$ that have endpoints in different parts. Therefore this expression measures the non-expansion of a partition.} in pseudoexpectation over $\mu$:
    \[
    \pE_\mu\left [\sum\limits_{s \in \Sigma}\langle f_s, (I-A_G)f_s \rangle\right ] \leq 2\eta + 2\frac{\eta}{1-\beta-\nu} + 2 \nu
    \]
    \item The functions $f_s$ are approximately Boolean-valued in pseudoexpectation over $\mu$:
    \[
    \pE_\mu[\sum_s B(f_s)] = \pE_\mu\left[\sum\limits_{s \in \Sigma} \e_u[f_s-f_s^2]\right] \leq \frac{\eta}{1-\beta-\nu} + \nu
    \]
\end{enumerate}

Since the functions $f_s$ are non-expanding (property (2) above), this allows us to use our $\ell_2$ characterization of non-expansion, to get a lower bound on the size of $f_s$ restricted to $r$-links. We will now use the key observation that $\underset{ X_\tau}{\mathbb{E}}[f_s]^2$ is exactly $D^k_if_s(\tau)^2$. In particular, applying \Cref{thm:structure-SoS} to the function $f_s$ with $\ell$ set to $r(\eta) = R_{1-16\eta}(M)$, gives:

\[\langle f_s, (I-M)f_s \rangle \geq (1 - \lambda_{r+1}(M))\left((1-c_1\gamma)\e[f_s] - {k \choose r} \uinner{D_r^k f_s}{D_r^k f_s} + (1-c_1\gamma)B(f_s) \right),
\]
where $B(f) = \e_{u \sim X_k}[f(u) - f(u)^2]$. Now using the key observation that $D_r^kf_s$ is exactly $\e_{X_\tau}[f_s]$ and noting that $1-\lambda_{r+1} \geq 16\eta$ by definition, we can simplify the above to:

\[\frac{1}{16\eta}\langle f_s, (I-M)f_s \rangle \geq \left((1-c_1\gamma)\e[f_s] - {k \choose r} \e_{r \sim X(r)}[\e_{X_\tau}[f_s]^2] + (1-c_1\gamma)B(f_s) \right).
\]

Taking a sum over $s$ and then a pseudoexpectation over $\mu$ yields:
\begin{align*}
\frac{1}{16\eta}\pE_\mu[\sum\limits_{s \in \Sigma} \langle f_s, (I-M)f_s \rangle] &\geq \pE_\mu\left[\sum\limits_{s \in \Sigma} \left((1-c_1\gamma)\mathbb{E}[f_s] - {k \choose r} \underset{\tau \in X(r)}{\mathbb{E}}\left[\underset{X_\tau}{\mathbb{E}}[f_s]^2\right] + (1-c_1\gamma)B(f_s) \right)\right]\\
    &=(1-c_1\gamma)\pE_\mu[\sum\limits_{s \in \Sigma} \mathbb{E}[f_s]]- {k \choose r}\underset{\tau \in X(r)}{\mathbb{E}}\left[\pE_\mu\left[\sum\limits_{s \in \Sigma }\underset{X_\tau}{\mathbb{E}}[f_s]^2\right]\right] + (1-c_1\gamma)\pE_\mu[\sum\limits_{s \in \Sigma} B(f_s)],
\end{align*}
where $c_1 \leq 2^{O(k)}$.


We will now use the observation that the global potential restricted to $X_\tau$ exactly corresponds to
the second term above! Recall that the global potential (Definition~\ref{def:sp-global}) over an $r$-link $X_\tau$ is given by:
\[
\Phi^I_{\beta,\nu}(\mu|_\tau) = \pE_{X,X' \sim \mu}[\sum_{s \in \Sigma} \underset{u \sim X_\tau}{\mathbb{E}}[f_s^{X,X'}(u)]^2].
\]
For convenience of notation we will drop the superscript $I$ in the potential notation henceforth.
Re-arranging gives a lower-bound on the global potential averaged across $r$-links:
\begin{equation}\label{eq:deg-2-pot}
\underset{\tau \in X(r)}{\mathbb{E}}\left[\Phi_{\beta,\nu}(\mu|_\tau)\right]
\geq  \frac{1}{{k \choose r}}\left((1-c_1\gamma)\pE[\sum\limits_{s} \mathbb{E}[f_s]] + (1-c_1\gamma)\pE[\sum\limits_{s}B(f_s)] - \frac{1}{16\eta}\pE[\sum\limits_{s} \langle f_s, (I-M)f_s \rangle]\right).
\end{equation}
Plugging in the properties of the functions $f_s$ discussed at the start of the proof, this gives the following bound on the global potential:
\begin{align*}
   \underset{\tau \in X(r)}{\mathbb{E}}\left[\Phi_{\beta,\nu}(\mu|_\tau)\right]
    \geq  \frac{1}{{k \choose r}}\left (1 - c_1\gamma - 2\frac{\eta}{1-\beta-\nu}-2\nu - \frac{1}{16\eta}\left (2\eta + 2\frac{\eta}{1-\beta-\nu} + 2\nu \right ) \right ).
\end{align*}
Setting $\beta = 19\eta$ and $\nu = \frac{\eta}{56{k \choose r}}$ and recalling our assumptions on $\gamma$ and $\eta$, this may be simplified by direct computation to:
\begin{align*}
\underset{\tau \in X(r)}{\mathbb{E}}\left[\Phi_{\beta,\nu}(\mu|_\tau)\right] \geq  \frac{1}{4{k \choose r}}.
\end{align*}

As a result, we see by averaging that there must exist some $r$-link with high potential:
\begin{align}\label{eq:glob-rest-val}
\exists \tau \in X(r): \Phi_{\beta,\nu}(\mu|_\tau) \geq \frac{1}{4{k \choose r}},
\end{align}
as desired.
\end{proof}

We will now prove \Cref{lem:ug-local-to-global}.

\begin{lemma}[\Cref{lem:ug-local-to-global} restated]\label{lem:ug-local-to-global-restated}
Assume the conditions of~\Cref{prop:high-global} hold. Let $X_\tau$ be an $r$-link for $r =R_{1-16\eta}(M)-1$ with high global potential: $\Phi_{\beta,\nu}^{I}(\mu|_\tau) \geq \frac{1}{4{k \choose r}}$, for $\beta = 19\eta$ and $\nu = \frac{\eta}{56{k \choose r}}$. Then the potential induced on $X_\tau$ is also high:
    \[\Phi_{\eta,\nu}^{I|_\tau}(\mu) \geq \frac{1}{8{k \choose r}}.\]
\end{lemma}
As discussed in \Cref{sec:unique-games}, proving \Cref{lem:ug-local-to-global} requires that a strong expansion property holds for links: expansion must be similar vertex-by-vertex. We now formalize this property and prove it holds for HD-walks.
\begin{definition}
The edge-expansion of a vertex $v$ in a set $S \subseteq X(k)$ with respect to a random walk operator $M$ is given by,
\[\phi_{S}(M,v) = 1 - \id{v}^T M \id{S} .\]
The edge-expansion of $S$ with respect to $M$ is the average edge-expansion taken over vertices in $S$:
\[\phi(M,S) = \E{v \sim S}{\phi_{S}(M,v)},\]
where $v \sim S$ is the distribution $\pi_k$ conditioned on $S$. As before, we drop $M$ from the notation when clear from context.
\end{definition}


We prove that the expansion of vertices in a link cannot vary much from the link's overall expansion (which is small by \Cref{thm:local-vs-global}).
\begin{lemma}[Restatement of \Cref{lem:vert-edge-exp}]
Let $M$ be a $k$-dimensional HD-walk on $d$-dimensional two-sided $\gamma$-local-spectral expander satisfying $\gamma \leq w(M)^{-1}2^{-\Omega(h(M)+k)}$ and $d>k$. Then for every $i$-link $X_\tau$, the deviation of the random variable $\phi_{X_\tau}(v)$ ($v \sim X_\tau$) is small:
\[\e_{v \sim X_\tau}[|\phi_{X_\tau}(v) - \phi(X_\tau)]|] \leq \frac{1}{2^{11k}}.\]
\end{lemma}

We prove this claim by reducing to simpler HD-walks we can analyze directly. With that in mind, let's first prove that a stronger point-wise result holds for any product of ``lower walks'' $(U_{k-1}D_k)^i$. 

\begin{proposition}\label{lem:uniform-exp}
Let $M$ be a $k$-dimensional HD-walk on $d$-dimensional two-sided $\gamma$-local-spectral expander satisfying $\gamma \leq w(M)^{-1}2^{-\Omega(h(M)+k)}$ and $d>k$.
For every $i$-link $X_\tau$, walk $(U_{k-1}D_k)^t$, and $k$-face $v \in X_\tau$, the expansion of $v$ in $X_\tau$ is almost exactly $\phi(X_\tau)$:
\[
|\phi_{X_\tau}((U_{k-1}D_k)^t,v) - \phi((U_{k-1}D_k)^t,X_\tau)| \leq t\gamma.
\]
\end{proposition}
\begin{proof}
It is more convenient to analyze the non-expansion at $v$, $\overline{\phi}_{X_\tau}((U_{k-1}D_k)^t,v) =  1-\phi_{X_\tau}((U_{k-1}D_k)^t,v)$, which for brevity we will denote by $\overline{\phi}(v)$. Below we will prove that for every $v \in X_\tau$, $\overline{\phi}(v) \in [c, c+t\gamma]$, for some fixed constant $c$. This immediately implies the result. Since by definition the non-expansion of $X_\tau$ denoted $\overline{\phi}((U_{k-1}D_k)^t,X_\tau)$, is an average over $\overline{\phi}(v)$, we get that: 
\[|\phi(v) - \phi((U_{k-1}D_k)^t,X_\tau)| = |\overline{\phi}(v) - \overline{\phi}((U_{k-1}D_k)^t,X_\tau)| \leq t\gamma,\] 
as desired.

Therefore let us now prove that $\overline{\phi}(v) \in [c,c+t\gamma]$ for all $v \in X_\tau$. First, notice that the non-expansion at $v$ is lower bounded by the probability that the walk does not remove any element from $\tau$ in any down-step. Since the down step is uniformly random, this probability is the same across all $v \in X_\tau$. Denote this probability by $c$.

Let us now upper bound $\overline{\phi}(v)$. The probability that the walk returns to $X_\tau$ is equal to $c$ plus the probability that starting from $v$, the walk leaves the link $X_\tau$ at some intermediate (down) step but ends back in $X_\tau$ regardless. We will show that the latter probability is at most $t\gamma$. 

Consider the first down-step where the walk leaves $X_\tau$, removing an element $w \in \tau$. To end up in the link of $\tau$, the walk needs to add $w$ back in one of the future up steps. The probability of this occurring at an up step from any $\sigma \in X(k-1)$ is exactly $\Pi_{\sigma,1}(w)$ by definition (where $\Pi_{\sigma,1}(w)=0$ if $w \notin X_\sigma$). Since $(X,\Pi)$ is a two-sided $\gamma$-local-spectral expander, $(X_{\sigma},\Pi_{\sigma})$ is a standard $\gamma$-spectral expander. It is a standard result that such graphs cannot have any vertex of weight greater than $\gamma$, thus $\Pi_{\sigma,1}(w) \leq \gamma$. By a union bound the probability that $w$ is added back in any of the future up steps is therefore at most $t \gamma$. Hence the probability that the walk returns to $X_\tau$ given that it exited in an intermediate step is at most $t \gamma$, so we get that $\overline{\phi}(v) \in [c, c+t\gamma]$.
\end{proof}
It's worth noting that the vertex-by-vertex expansion property actually holds exactly for the canonical walks. However, we can only reduce to canonical walks when the height of $M$ is small, an issue we avoid by analyzing lower walks. We now prove \Cref{lem:vert-edge-exp} by reducing general $M$ to this case.

\begin{proof}[Proof of \Cref{lem:vert-edge-exp}]
By repeated application of \Cref{lemma:body-DU-UD}, any $k$-dimensional HD-walk $M$ on a two-sided $\gamma$-local-spectral expander with $k>0$ can be written as a linear combination of lower products $(U_{k-1}D_k)^i$ up to $O(\gamma)$ error in the following sense:
\[M = \sum_{i = 0}^{h(M)} c_i (U_{k-1}D_k)^i + \Gamma,\]
where $(U_{k-1}D_k)^0$ denotes the identity matrix and $\norm{\Gamma} \leq w(M)2^{O(h(M))}\gamma$. For simplicity of notation, we write $C \coloneqq \norm{\Gamma}$, which is in turn $\leq 1/2^{13k}$ for $\gamma$ chosen to be small enough.

We also know by Lemma~\ref{lem:uniform-exp} that for all $i$ there exists a constant $u_i$ such that for all vertices $v \in X_\tau$, $\phi_{X_\tau}((U_{k-1}D_k)^i,v) = u_i \pm h(M)\gamma$. Hence we get that,

\[1 - \phi_{X_\tau}(M,v) = \sum_{i = 0}^j c_i \id{v}^T (U_{k-1}D_k)^i \id{X_\tau} + \id{v}^T \Gamma \id{X_\tau} = \sum_{i} c_i u_i + \id{v}^T \Gamma \id{X_\tau} \pm w(M)h(M)\gamma.\]

Let $\sum_i c_i u_i = u$. We will show that in expectation over $v \in X_\tau$, $|1 - \phi_{X_\tau}(M,v) - u| \leq |\id{v}^T \Gamma \id{X_\tau}| + w(M)h(M)\gamma$ is small.
Define $err$ as the error vector $\Gamma\id{X_\tau}$. We will bound the 1-norm of $err$ when restricted to coordinates in $X_\tau$. Define the vector $s$ to be $\text{sign}(err(v))$ for $v \in X_\tau$ and $0$ otherwise. First note that,

\[\uinner{s}{err} = \sum_{v \in X_\tau} \Pr_{\pi_k}[v] \cdot |err(v)|.\]

Applying Cauchy-Schwarz on the LHS we get that,

\begin{align*}
\sum_{v \in X_\tau} \Pi_k(v) |err(v)| &\leq \norm{s} \cdot \norm{err}\\ 
&= \norm{\id{X_\tau}} \cdot \norm{\Gamma \id{X_\tau}} \\
&\leq \norm{\id{X_\tau}} \cdot C\norm{\id{X_\tau}} \\ 
&= C \uinner{\id{X_\tau}}{\id{X_\tau}}_{\pi_k}.
\end{align*}

So we get that, 

\[\frac{\sum_{v \in X_\tau} \Pi_k(v) |err(v)|}{\uinner{\id{X_\tau}}{\id{X_\tau}}_{\pi_k}} = \underset{v \sim X_\tau}{\mathbb{E}}[|err(v)|] \leq  C.\]

Using the fact that $\E{v \sim X_\tau}{\phi_{X_\tau}(M,v)} = \phi(M,X_\tau)$, we can conclude with a simple trick:

\[|1 - \phi(M,X_\tau) - u | = |\E{v \sim X_\tau}{1 - \phi_{X_\tau}(M,v) - u}| \leq \E{v \sim X_\tau}{|err(v)|} + w(M)h(M)\gamma  \leq 2C.\]

Since the average deviation of $(1-\phi_{M,\tau}(v))$ from $u$ is small, the average deviation from the mean $(1-\phi_M(X_\tau))$ should also be small by triangle inequality:

\[\E{v \sim X_\tau}{|\phi_{X_\tau}(M,v) - \phi(M,X_\tau)|} \leq \E{v \sim X_\tau}{|1 - \phi_{X_\tau}(M,v) - u|} + C \leq 3C \leq \frac{1}{2^{11k}}.\]

\end{proof}

We are now ready to prove \Cref{lem:ug-local-to-global}/Lemma~\ref{lem:ug-local-to-global-restated}.
\begin{proof}[Proof of \Cref{lem:ug-local-to-global}]
We will use \Cref{lem:vert-edge-exp} to relate the global potential on $X_\tau$ to the potential induced on $X_\tau$. As a reminder, we proved that for every $i$-link $X_\tau$, the deviation of the random variable $\phi_{X_\tau}(M,u)$ ($u \sim X_\tau$) is small:
\[\e_{u \sim X_\tau}[|\phi_{X_\tau}(M,u) - \phi(M,X_\tau)]|] \leq \frac{1}{2^{11k}}.\]

Let us partition $X_\tau$ into two sets, $G$ and $\overline{G},$ where $G$ contains all vertices $u$ for which $\phi_{M,X_\tau}(u)$ is close to $\phi_M(X_\tau)$. We will henceforth drop the subscripts $M, X_\tau$. Formally we define:

\[G = \left\{u \in X_\tau \mid |\phi(u) - \phi(X_\tau)| \leq \frac{1}{2^k}.\right\}\]

Now note that by Markov's inequality on \Cref{lem:vert-edge-exp} we have that $\Pr_{u \sim X_\tau}[u \notin G] := c \leq 1/2^{10k}$. For an assignment $X$, let $Z_{u,s}$ denote $\ind[X_u - X'_u = s]$. Expanding out the definition of the global potential on $X_\tau$ (\Cref{def:sp-global}), for any two assignments $X,X'$ we get:
\begin{align*}
\Phi_{\beta,\nu}(X,X')|_\tau &= \sum_s \e_{u \sim X_\tau}[Z_{u,s}p_{\beta,\nu}(\val_u(X))] \\
&= \sum_s ((1-c) \e_{u \sim G}[Z_{u,s}p_{\beta,\nu}(\val_u(X))] + c \e_{u \sim \overline{G}}[Z_{u,s}p_{\beta,\nu}(\val_u(X))])^2 \\
&\leq \sum_s (1-c)^2 \e_{u \sim G}[Z_{u,s}p_{\beta,\nu}(\val_u(X))]^2 + \sum_s c(1 - c) \e_{u \sim \overline{G}}[Z_{u,s}] \\ &+\sum_s c^2 \e_{u \sim \overline{G}}[Z_{u,s}]^2 \\
&\leq \sum_s (1-c)^2 \e_{u \sim G}[Z_{u,s}p_{\beta,\nu}(\val_u(X))]^2 + c(1 - c) + c^2,
\end{align*}
where $c \leq 1/2^{10k}$ and we used the fact that $p_{\beta,\nu}(y) \leq 1$, when $y \in [0,1]$, $Z_{u,s} \in [0,1]$ and $\sum_s Z_{u,s} = 1$. 

To relate the two potentials we will relate the quantities $\val_u(X)$ and $\val_u^\tau(X)$ for $u \in G$, where $\val_u(X)$ denotes the fraction of edges incident on $u$ that are satisfied by $X$, and further $\val_u^\tau(X)$ denotes the fraction of edges incident on $u$ that lie inside $X_\tau$ and are satisfied by $X$. We will use the fact that vertices in $G$ have small expansion. When $r = R_{1-16\eta}(M)$, by \Cref{thm:local-vs-global} the expansion of the $r$-link $X_\tau$ is at most $1 - \lambda_{r+1} + w(M)h(M)^22^{O(k)}\gamma$ which is less than $17\eta$ because of the way the parameters have been set up. Therefore for all vertices $u$ in $G$, $\phi(u) \leq 18 \eta$ since $\eta > 1/2^k$. This implies that for any assignment $X \in \Sigma^V$ and any $u \in G$:
\[
\ind[\val_u^\tau(X) \ge \beta - 18\eta] \ge \ind[\val_u(X) \ge \beta],
\]
Furthermore by the properties of the polynomials $p_{\beta,\nu}(Y)$, since $\nu < \eta$,
\[
p_{\beta - 18 \eta,\nu}(\val_u^\tau(X)) + \nu \ge p_{\beta,\nu}(\val_u(X)) - \nu.
\]

Finally expanding out the definition of the potential induced on $X_\tau$ (\Cref{def:sp-induced}), for any two assignments $X,X'$ we get:

\begin{align}
    \Phi^\tau_{\beta - 18\eta,\nu}(X,X')
    &= \sum_{s \in \Sigma} \e_{u \sim X_\tau}[Z_{u,s}\cdot p_{\beta-18\eta,\nu}(\val_u^\tau(X))]^2 \nonumber\\
    &= \sum_{s \in \Sigma} \e_{u \sim X_\tau}[Z_{u,s}\cdot \left(p_{\beta,\nu}(\val_u(X)) - 2\nu\right)]^2 \nonumber\\
    &\geq \left(\sum_{s \in \Sigma} \e_{u \sim X_\tau}[Z_{u,s}\cdot \left(p_{\beta,\nu}(\val_u(X))\right)]^2\right) - 4\nu \nonumber\\
    &\ge \sum_{s \in \Sigma} (1-c)^2 \e_{u \sim G}[Z_{u,s}p_{\beta,\nu}(\val_u(X))]^2 - 4\nu \nonumber\\
    &\ge \Phi_{\beta,\nu}(X,X')|_\tau - c(1-c) - c^2 - 4\nu, \label{eq:local-global}
\end{align}
where we again used the facts that $Z_{u,s} \in [0,1]$, $\sum_s Z_{u,s} = 1$ and $p_{\beta,\nu}(y) \in [0,1]$ when $y \in [0,1]$. We also have that $c \leq \frac{1}{2^{10k}} \leq \frac{\eta}{56{k \choose r}}$ since $\eta > 1/2^k$ and therefore $c \leq \nu$ implying that: $\Phi^\tau_{\beta - 18\eta,\nu}(X,X') \geq \Phi_{\beta,\nu}(X,X')|_\tau - 6\nu$.

Now note that all the inequalities above are sum-of-squares inequalities of degree at most $2\deg(p) = \tO(1/\nu)$. We can therefore relate the two potentials when measured with respect to $\mu$, which is a degree-$\tO(1/\nu)$ pseudodistribution, by applying the pseudoexpectation operator $\pE_\mu$ to \Cref{eq:local-global} above:
\[\Phi_{\beta - 18\eta,\nu}^\tau(\mu) = \pE_\mu[\Phi^\tau_{\beta - 18\eta,\nu}(X,X')]  \geq \pE_\mu[\Phi_{\beta,\nu}(X,X')|_\tau] - 6\nu = \Phi_{\beta,\nu}(\mu|_\tau) - 6\nu.\]

By the conditions of the lemma, $\Phi_{\beta,\nu}(\mu|_\tau) = 1/4{k \choose r}$, $\beta = 19\eta$ and $\nu = \frac{\eta}{56{k \choose r}}$ we get $\Phi_{\eta,\nu}^\tau(\mu) \geq \frac{1}{8{k \choose r}}$. This completes the proof of the lemma. 
\end{proof}

\end{document}